\documentclass[reqno]{amsart}
\usepackage{amssymb}
\usepackage[dvips]{epsfig}
\usepackage{graphicx}
\usepackage{color}
\usepackage{amsmath}
\usepackage{amssymb}
\usepackage{amsfonts}
\usepackage{amsthm}
\usepackage{bbm}
\usepackage{mathrsfs}
\newcommand{\D}{\mathbb D}
\newcommand{\R}{\mathbb R}
\newcommand{\p}{\partial}

\newcommand{\Z}{\mathbb Z}

\newcommand{\C}{\mathbb C}

\newcommand{\ep}{\varepsilon}

\newcommand{\g}{\gamma}
\newcommand{\et}{\eta}
\renewcommand{\l}{\lambda}
\renewcommand{\lg}{\Lambda}

\newcommand{\N}{\mathbb{N}}
\newcommand{\T}{\mathbb{T}}

\newcommand{\boA}{\mathcal{A}}
\newcommand{\boB}{\mathcal{B}}
\newcommand{\boD}{\mathcal{D}}
\newcommand{\boE}{\mathcal{E}}
\newcommand{\boW}{\mathcal{W}}
\newcommand{\boG}{\mathcal{G}}
\newcommand{\boT}{\mathcal{T}}
\newcommand{\boI}{\mathcal{I}}

\newcommand{\boH}{\mathcal{H}}

\newcommand{\boM}{\mathcal{M}}

\newcommand{\boL}{\mathcal{L}}
\newcommand{\boK}{\mathcal{K}}

\newcommand{\ji}{\langle}
\newcommand{\jd}{\rangle}

\newcommand{\dist}{{\rm dist}}
\newcommand{\diam}{{\rm diam}}
\newcommand{\good}{{\rm good}}
\newcommand{\tz}{\tilde{\zeta}}
\newcommand{\tO}{\tilde{\Omega}}

\newtheorem{thm}{Theorem}[section]
\newtheorem{lem}[thm]{Lemma}
\newtheorem{stm}[thm]{Statement}

\newtheorem{cor}[thm]{Corollary}
\theoremstyle{remark}
\newtheorem{rem}{\bf Remark}[section]
\theoremstyle{definition}

\numberwithin{equation}{section}

\usepackage[colorlinks=true,pdfstartview=FitV,linkcolor=magenta,citecolor=cyan]{hyperref}
\usepackage{bm}
\begin{document}

	\title[QP localization]{Green's function estimates for long-range  quasi-periodic operators on $\Z^d$ and applications}
	
	%\author[Shi]{Yunfeng Shi}
	%\address[Shi] {School of Mathematics,
	%	Sichuan University,
		%Chengdu 610064,
		%China}
	%\email{yunfengshi@scu.edu.cn}
	
	\author[Wen]{Li Wen}
	\address[Wen] {School of Mathematics,
		Sichuan University,
		Chengdu 610064,
		China}
	\email{liwen.carol98@gmail.com}
	
	\author[Wu]{Yuan Wu}
	\address[Wu]{School of Mathematics and Statistics,
		Huazhong University of Science and Technology,
		Wuhan 430070, 
		China}
		\email{wuyuan@hust.edu.cn}
	\date{\today}
	\maketitle
	\begin{abstract}
		We establish quantitative   Green's function  estimates for a class of quasi-periodic (QP)  operators on $\Z^d$ with certain  slowly  decaying long-range hopping and analytic cosine type potentials.  As applications, we prove  the arithmetic spectral localization,  and  obtain upper bounds on quantum dynamics for all phase parameters.  To deal  with quantum dynamics estimates, we develop an  approach employing  separation property (rather than the sublinear bound)  of resonant blocks in the regime of Green's function estimates. 
	\end{abstract}
	
	\section{Introduction}
	The study of  long-range hopping QP  operators   has attracted great attention over the years, cf. e.g.,   \cite{BJ02,JK16,GY20, JLS20, JL21, Shi22, Liu22,  SW22, GYZ23,Liu23, Shi23, SS23, SW23,SW24,WXZ25}.  
	In this paper, we are concerned with  long-range QP   operators
	\begin{align}\label{model}
		\mathcal{H}(\theta)=\ep \mathcal{W}_\phi+v(\theta+\bm n\cdot\bm \omega)\delta_{\bm n,\bm n'},\ \bm n,\bm n'\in\Z^d,
	\end{align}
	where the off-diagonal  part (i.e., the hopping term) $\mathcal{W}_{\phi}$ is a Toeplitz operator  satisfying
	\begin{align}\label{wphi}
		(\mathcal{W}_\phi \psi)(\bm n)=\sum_{\bm l\in\Z^d}\phi(\bm n-\bm l)\psi(\bm l),\ \phi(\bm 0)=0,\ |\phi(\bm n)|\le e^{-\alpha\log^{\rho}(1+\|\bm n\|)}
	\end{align}
	with some $\alpha>0$, $\rho>1$ and $\|\bm n\|=\sup\limits_{1\le i\le d}|n_i|$.   The potential $v$  is an analytic cosine type function (cf. \eqref{vdefn} for details).  Note  that the hopping term $\mathcal{W}_\phi$  had  previously been used in the construction of almost-periodic solutions for some nonlinear Hamiltonian equations \cite{Pos90}, and was  recently  introduced  by Shi-Wen    \cite{SW22} and Shi-Wen-Yan \cite{SWY25} in the  study of the localization problems for  operators on $\Z^d$ with monotone quasi-periodic potentials  and  random potentials, respectively.   We also mention that the existence of  localized  eigenfunctions whose decay rate is the same as \eqref{wphi}  has been established   in  physics  \cite{ca17}.  The present work aims to prove establish quantitative Green's function estimates for \eqref{model} via multi-scale analysis (MSA)  scheme  in the spirit of Fr\"ohlich-Spencer-Wittwer \cite{FSW90}, Bourgain \cite{Bou00} and Cao-Shi-Zhang \cite{CSZ24a}. As applications, we prove the arithmetic localization and obtain upper bounds on quantum dynamics for all phase parameters. 
	
	The groundbreaking work of Fr\"ohlich-Spencer-Wittwer \cite{FSW90} was the first to extend the Green's function estimation method based on multi-scale analysis (MSA) from random potentials \cite{FS83} to QP potentials. This method is primarily based on eigenvalue variational techniques, thus requiring the operators to be self-adjoint and the potentials to be single-variable functions. In 2000, Bourgain \cite{Bou00} made a breakthrough by proposing an MSA-type Green's function estimation method that does not rely on eigenvalue variation techniques: the resonances are entirely determined by the zeros of the Dirichlet determinant based on the preparation theorem. Subsequently, Bourgain and collaborators \cite{BGS02,Bou07} incorporated matrix-valued Cartan’s estimates and semi-algebraic set theory into the MSA framework, enabling them to derive large deviation theorems (LDT) for the finite volume Green's functions of multi-dimensional QP operators with multivariate analytic potentials. The result of \cite{Bou07} was later extended by Jitormirskaya-Liu-Shi \cite{JLS20} to  operators with arbitrary multi-frequency and exponentially long-range hopping. In 2024, Cao-Shi-Zhang \cite{CSZ24a} successfully generalized Bourgain's method \cite{Bou00} to arbitrary-dimensional lattice $\Z^d$ case. Meanwhile, the eigenvalue-variation-based Green's function estimation method \cite{FSW90} was also extended by Cao-Shi-Zhang \cite{CSZ23,CSZ24b} to the $\Z^d$ case: they overcame the challenges posed by level crossing phenomena.
	
	In fact, the Green's function estimates actually allow us to study Anderson localization (i.e., the pure point spectrum with exponentially decaying eigenfunctions). To derive localization, it is necessary to remove certain parameters ($\omega$ or $\theta$) to overcome the double resonance phenomenon. The proof of localization via LDT-type Green's functions estimate methods typically requires removing a positive-measure set (depending on $\theta$) of $\omega$, resulting in a non-arithmetic result. To achieve arithmetic localization (the removed sets for both $\omega$ and $\theta$  in the localization proof admit an explicit arithmetic description), precise characterization of resonances is required. In higher dimensions, \cite{JK16,GY20} first achieved arithmetic localization for QP operators with the cosine potential via reducibility arguments based on the Aubry duality. Recently, Cao-Shi-Zhang \cite{CSZ23,CSZ24b,CSZ24a} proposed an innovative approach to proving arithmetic localization based on Green's function estimates, which enabled the demonstration of arithmetic localization for multi-dimensional QP operators with $C^2$-cosine like potentials. All the aforementioned arithmetic localization results require operators with exponentially long-range hopping. The primary objective of this work is to verify arithmetic localization for multi-dimensional QP operators featuring slower-decaying (cf. \eqref{wphi}) long-range hopping.
	
	Cao-Shi-Zhang \cite{CSZ23,CSZ24b,CSZ24a} even established arithmetic dynamical localization for these operators. In particular, for fixed Diophantine frequency $\omega$ and a.e. $\theta\in\T$, 
	\begin{align*}
		\sup_{t\in\R}	\left(\left(|\mathscr{X}_{\mathcal{H}(\theta)}|_{\delta_{\bm 0}}^p\right)(t)\right)<+\infty\text{ (cf. \eqref{pthm} for details)}.
	\end{align*}
	A natural question arises: can one demonstrate upper bounds of $\left(|\mathscr{X}_{\mathcal{H}(\theta)}|_{\delta_{\bm 0}}^p\right)(t)$ in $t$ (quantum dynamics estimates) for all $\theta\in\T$? This problem has been resolved by \cite{JL21,SS23,Liu23} in the case of exponentially long-range QP operators. The methodologies in these works are fundamentally built upon large deviation estimates of Green's functions, thus inherently requiring the long-range hopping to exhibit exponential decay. Another motivation of the present work is to develop a novel strategy based on the separation property of resonant blocks (rather than LDT) for investigating quantum dynamical estimates, which enables the treatment of long-range hopping of type \eqref{wphi}. In fact, building upon Green's function estimates of \cite{SW24}, our method can handle quantum dynamical estimates for operators with power-law long-range hopping.
	
	\subsection{Main results}
	%In this paper, we are concerned with  QP   operators
	%\begin{align}\label{model}
		%\mathcal{H}(\theta)=\ep \mathcal{W}_\phi+v(\theta+\bm n\cdot\bm \omega)\delta_{\bm n,\bm n'},\ \bm n,\bm n'\in\Z^d,
	%\end{align}
	%where the off-diagonal  part (i.e., the hopping term) $\mathcal{W}_{\phi}$ is a Toeplitz operator  satisfying
	%\begin{align}\label{wphi}
	%	(\mathcal{W}_\phi \psi)(\bm n)=\sum_{\bm l\in\Z^d}\phi(\bm n-\bm l)\psi(\bm l),\ \phi(\bm 0)=0,\ |\phi(\bm n)|\le e^{-\alpha\log^{\rho}(1+\|\bm n\|)}
	%\end{align}
	%with some $\alpha>0$, $\rho>1$ and $\|\bm n\|=\sup\limits_{1\le i\le d}|n_i|$.   
	Recalling \eqref{model}, 
	the potential $v$  is an analytic function defined on
	\begin{align*}
		\D_R=\left\{z\in\C/\Z: \ |\Im z|\le R\right\},\ R>0,
	\end{align*}
	satisfying the following Morse condition: there exist $\kappa_1,\kappa_2>0$ such that
	\begin{align}\label{vdefn}
		\kappa_1\|z_1-z_2\|_{\T}\|z_1+z_2\|_{\T}\le|v(z_1)-v(z_2)|\le \kappa_2\|z_1-z_2\|_{\T}\|z_1+z_2\|_{\T},
	\end{align}
	for all $z_1,z_2\in\D_R$, where the torus norm is defined by
	\begin{align*}
		\|z\|_\T=\sqrt{\|\Re z\|_\T^2+|\Im z|^2},\ \|x\|_{\T}=\inf_{l\in\Z}|l-x|\ \text{for}\ x\in\R.
	\end{align*}
	We let $\theta\in\T=\R/\Z$, $\bm \omega\in[0,1]^d$, and $
	\bm n\cdot\bm \omega=\sum\limits_{i=1}^{d}n_i\omega_i.$
	In the following, we assume that $\bm \omega\in {DC}_{\tau,\g}$ for some $\tau>d$, $\g>0$ with
	\begin{align}\label{dc}
		{DC}_{\tau,\g}=\left\{\bm\omega\in[0,1]^d:\ \|\bm n\cdot\bm \omega\|_\T\ge\frac{\g}{\|\bm n\|^\tau}\ {\rm for}\ \forall\bm n\in\Z^d\setminus\{\bm 0\}\right\}.
	\end{align}
	
	\subsubsection{Quantitative Green’s function estimates}
	The main result of this paper is a quantitative version of Green's function estimates, which will
	imply both arithmetic spectral localization and sub-polynomial bounds of moments.

	We first introduce the function class of potentials. For fixed $R > 0$, let $\mathscr{V}_R$ be the class of analytic functions $v: \D_R \to \C$ satisfying the following uniform condition: there exist positive constants $\kappa_1 = \kappa_1(v,R)$ and $\kappa_2 = \kappa_2(v,R)$ such that inequality \eqref{vdefn} holds for all $z_1, z_2 \in \D_R$. Let $|v|_R=\sup\limits_{z\in\D_R}|v(z)|$ for $v\in\mathscr{V}_R$ and $\mathscr{V}=\bigcup_{R>0}\mathscr{V}_R$. We remark that more details about $\mathscr{V}$ can be found in \cite{SW24}.

	Given $E\in\C$ and $\lg\subset\Z^d$, the Green's function (if exists) is defined by
	\begin{align*}
		\mathcal{T}_{\lg}^{-1}(E;\theta)=(\mathcal{H}_\lg(\theta)-E)^{-1},\ \mathcal{H}_{\lg}(\theta)=\mathcal{R}_\lg \mathcal{H}(\theta)\mathcal{R}_\lg,
	\end{align*}
	where $\mathcal{H}(\theta)$ is given  by \eqref{model} and $\mathcal R_\Lambda$ denotes the restriction operator.
	
	Recall that $\bm\omega\in DC_{\tau,\g}$ and $\rho>1$. We fix a constant $\rho'$ so that
	\begin{align*}
		1<\rho'<\rho<\rho'+1.
	\end{align*}
	 At the s-th iteration step, let $\delta_{s}^{-1}$ (resp. $N_s$) describe the resonance strength (resp. the size of resonant blocks) defined by
		\begin{align*}
		N_{s+1}=\left[e^{\left|\log\delta_{s}\right|^{\frac{1}{\rho'}}}\right],\ \delta_{s+1}=\delta_s^{10^{5\rho'}},\ \delta_0=\ep_0^{\frac{1}{10}},
	\end{align*}
	where $[x]$ denotes the interger part of $x\in\R$.
	
	Then we have
	\begin{thm}\label{ge}
		Let $1<\rho'<\rho<1+\rho'$, $\bm\omega\in DC_{\tau,\g}$ and $v\in\mathscr{V}_R$. Then there is some $\ep_0=\ep_0(d,\tau,\g,v,R,\alpha,\rho,\rho')>0$ so that for $0<|\ep|\leq \ep_0$ and $E\in v(\D_{R/2})$, there exists a sequence
		\begin{align*}
			\{\theta_s=\theta_s(E)\}_{s=0}^{s'}\subset\C\ (s'\in\N\cup\{+\infty\})
		\end{align*}
		with the following properties. Fix any $\theta\in \T$, if a subset $\lg\subset\Z^d$ is $s$-$\good$ (cf.  $(\bm e)_s$ of Statement \ref{state} for the definition of $s$-$\good$ set, and \S\ref{qgfe} for the definitions of $\{\theta_s\}_{s=0}^{s'}$, sets $P_s$, $\tilde{\Omega}_{\bm k}^s$ and $\tz_s>0$), then
		\begin{align*}
			&\|\mathcal{T}_\lg^{-1}(E; \theta)\|<2\delta_{s-1}^{-3}\cdot\sup\limits_{\{\bm k\in P_s:\ \tilde{\Omega}_{\bm k}^s\subset\lg\}}(\|\theta+\bm k\cdot\bm \omega-\theta_s\|_{\T}^{-1}\cdot\|\theta+\bm k\cdot\bm\omega+\theta_s\|_{\T}^{-1})<\delta_{s}^{-3},\\
			&|\boT_{\lg}^{-1}(E;\theta)(\bm x,\bm y)|<e^{-\frac{1}{2}\alpha\log^{\rho}(1+\|\bm x-\bm y\|)}\ \text{for $\|\bm x-\bm y\|>10\tz_s$}.
		\end{align*}
	\end{thm}
	
	Let us refer to \S\ref{qgfe} for a complete description of Green’s function estimates.

\subsubsection{Arithmetic spectral localization}
In this part, we state our arithmetic spectral localization result. We assume $\boH(\theta)$ is self-adjoint for $\theta\in\T$.
	\begin{thm}\label{asl}
				Define  
		\begin{align*}
			\Theta=\left\{\theta\in\T: \ \|2\theta+\bm n\cdot\bm\omega\|_{\T}\le\frac{1}{\|\bm n\|^{\tau}}\ {\rm holds\  for\  finitely\  many\ } \bm n\in\Z^d\right\}.
		\end{align*}
		Under the assumptions  of Theorem \ref{ge}, for $\theta\in\T\setminus \Theta$, 	$\mathcal{H}(\theta)$ has pure point spectrum and there exists a complete system of eigenfunctions $\psi=\{\psi(\bm n)\}_{\bm n\in\Z^d}$ satisfying
		\begin{align*}
				|\psi(\bm n)|<e^{-\frac{\alpha}{4\cdot 10^{6\rho}}\log^{\rho}\left(1+\|\bm n\|\right)} \text{ for $\|\bm n\|\gg1$}.
		\end{align*}
	\end{thm}

	\subsubsection{Sub-polymonial bounds of moments}
	In this part, we state our quantum dynamics results. We also assume $\mathcal{H}(\theta)$ is self-adjoint for $\theta\in\T$.
	
	For $\psi\in\ell^2(\Z^d)$ and $p>0$, let $(|\mathscr{X}_{\mathcal{H}(\theta)}|_{\psi}^p)(t)$ be the $p$-th moment of $\mathcal{H}(\theta)$,
	\begin{align}\label{pthm}
		\left(|\mathscr{X}_{\mathcal{H}(\theta)}|_{\psi}^p\right)(t)=\sum_{\bm n\in\Z^d}(1+\|\bm n\|)^p\left|\ji e^{-\sqrt{-1}t\mathcal{H}(\theta)}\psi,\delta_{\bm n}\jd\right|^2,
	\end{align}
	and $(|\bar{\mathscr{X}}_{\mathcal{H}(\theta)}|_{\psi}^p)(T)$ be the time-average $p$-th moment of $\mathcal{H}(\theta)$,
\begin{align}\label{tpthm}
	\left(|\bar{\mathscr{X}}_{\mathcal{H}(\theta)}|_{\psi}^p\right)(T)=\frac{2}{T}\int_0^{\infty} e^{-\frac{2t}{T}}\sum_{\bm n\in\Z^d}(1+\|\bm n\|)^p\left|\ji e^{-\sqrt{-1}t\mathcal{H}(\theta)}\psi,\delta_{\bm n}\jd\right|^2dt.
\end{align}
Then we have
\begin{thm}\label{sne}
Let $p>0$. Under the assumptions of Theorem \ref{ge}, there exists $T_0=T_0(d,\tau,\g,v,R,\alpha,\rho,\rho',p)>0$ such that for any $\theta\in\T$ and $t,T\ge T_0$, we have
	\begin{align*}
		\left(|\mathscr{X}_{\mathcal{H}(\theta)}|_{\delta_{\bm 0}}^p\right)(t)\le2^p e^{p(\log t)^{\frac{2}{1+\rho'}}}
	\end{align*}
	and
	\begin{align*}
			\left(|\bar{\mathscr{X}}_{\mathcal{H}(\theta)}|_{\delta_{\bm 0}}^p\right)(T)\le 2^p e^{p(\log T)^{\frac{2}{1+\rho'}}}.
	\end{align*}
\end{thm}

	\subsection{Ideas of the proof and new ingredients}
	Our proof of Theorem \ref{ge} is based on a MSA type induction and combines ideas from \cite{Bou00,CSZ24a,SW24,SWY25}. Once Theorem \ref{ge} was established, the proof of arithmetic localization just follow from standard argument. Building upon Theorem \ref{ge}, we also develop a new scheme to prove sub-polynomial bounds of moments by using separation property of resonant blocks.
	
	The $\ell^2$-norm estimates for the Green's function based on the similar induction hypothesis of Theorem \ref{ge} have already established in \cite{Bou00,CSZ24a,SW24}. In this work, we primarily employ the methods developed in \cite{SWY25} to handle the off-diagonal decay. While our method here is motivated by \cite{SWY25}, there are some major differences: (1). In \cite{SWY25}, by adopting the induction hypothesis from \cite{DK89}, the authors only need to handle ensembles of two adjacent scales. Since we are dealing with QP case, one needs to handle the ensemble of all scales at each step. The iterative framework in \cite{SWY25}, when applied directly, proves insufficient to achieve our target decay rate; (2). Although the quasi-metric property (cf. \eqref{quaeq} for details) is utilized in both this work and \cite{SWY25}, we emphasize that $n\approx \frac{\log^{\rho}N_{k+1}}{\log^{\rho}N_k}$ in \cite{SWY25} which can be controlled by scale-independent constant but $n\approx \frac{log^{\rho}N_{k+1}}{\log^{\rho}N_1}$ maybe present in this work. In contrast to previous studies, we demonstrate the intrinsic necessity of term $C(\rho)\log^{\rho}n$.
	
	In previous works, the proof of Combes-Thomas estimates invariably required the long-range hopping term to take the form $e^{f(\bm x,\bm y)}$, where $f(\bm x,\bm y)$ is a metric. However, the function $f(\bm x,\bm y)=\log^{\rho}(1+\|\bm x-\bm y\|)$ $(\rho>1)$ is not a metric. Fortunately, it was proved in \cite{SWY25} that the function $f(\bm x,\bm y)$ is a quasi-metric. By exploiting quasi-metric property, we construct customized Combes-Thomas type estimates for the present framework.
	
	Finally, we conclude by noting a crucial distinction: while \cite{Liu23} and this work share analogous approaches to bound of moments, our construction of good sets follows a divergent path. In \cite{Liu23}, for any $\|\bm n\|\gg1$, the author relies on the large deviation theorems, postulating the existence of a neighborhood of $\bm n$ satisfying the sublinear bound property. However, in this paper, we can obtain a good set containing $\bm n$ by extending a origin neighborhood of $\bm n$ directly (cf. Lemma \ref{93A} for details). 
	
	\subsection{Structure of the paper}
	The paper is organized as follows. In \S\ref{notation}, we introduce some useful notations. The \S\ref{prelim} contains some important properties. In \S\ref{qgfe}, we complete the proof of Theorem \ref{ge}. In \S\ref{Arith}-\S\ref{subpoly}, we will apply quantitative Green's function estimates to prove spectral locaolization and sub-polynomial bounds of moments. Some useful estimates are given in Appendix \ref{app}.

	\section{The notation}\label{notation}
	\begin{itemize}		
		\item The determinant of a matrix $M$ is denoted by $\det M$.
		
		\item If $a\in\R$, let $\|a\|_{\T}=\dist(a,\Z)=\inf\limits_{l\in\Z}|l-a|$. For $z=a+\sqrt{-1}b\in\C$ with $a,b\in\R$, define $\|z\|_{\T}=\sqrt{\|a\|_{\T}^2+|b|^2}$.
		
		\item For $\bm n\in\R^d$, let
		\begin{align*}
			\|\bm n\|=\sup_{1\le i\le d}|n_i|.
		\end{align*}
		Denote by $\dist(\cdot,\cdot)$ the distance induced by $\|\cdot\|$ on $\R^d$, and define
		\begin{align*}
			\diam(\lg)=\sup_{\bm k,\bm k'\in\lg}\|\bm k-\bm k'\|.
		\end{align*}
		Given $\bm n\in\Z^d$, $\lg'\subset\frac{1}{2}\Z^d$ and $L>0$, define
		\begin{align*}
			\lg_L(\bm n)=\left\{\bm k\in\Z^d:\ \|\bm k-\bm n\|\le L\right\}
		\end{align*}
		and
		\begin{align*}
			\lg_L(\lg')=\left\{\bm k\in\Z^d:\ \dist(\bm k,\lg')\le L\right\}.
		\end{align*}
		In particular, write $\lg_L=\lg_L(\bm 0)$.
		\item For $\et>0$, we define
			\begin{align}\label{Det}
			D(\et)=\sum_{\bm k\in\Z^d}e^{-\et\log^{\rho}(1+\|\bm k\|)}<+\infty.
		\end{align}
		
		\item $\{\delta_{\bm x}\}_{\bm x\in\Z^d}$ is the standard basis of $\ell^2(\Z^d)$.
		
		\item $\mathcal{I}$ typically denotes the identity operator.
		
		\item $\mathcal{R}_\lg$ is the restriction operator  with  $\lg\subset\Z^d$.
		
		\item Let $\mathcal{T}:\ \ell^2(\Z^d)\to\ell^2(\Z^d)$ be a linear operator. Denote by $\ji\cdot,\cdot\jd$ the standard inner product on $\ell^2(\Z^d)$. Set $\mathcal{T}(\bm x,\bm y)=\ji \delta_{\bm x}, \mathcal{T}\delta_{\bm y}\jd$. The spectrum of operator $\mathcal{T}$ is denoted by $\sigma(\mathcal{T})$. Finally, we define $\mathcal{T}_\lg=\mathcal{R}_\lg \mathcal{T}\mathcal{R}_\lg$ and $\mathcal{T}_{\lg}^*$ the adjugate operator of $\mathcal{T}_{\lg}$, where $\lg\subset\Z^d.$
		
	\end{itemize}
	
	\section{Preliminaries}\label{prelim}
	\subsection{Quasi-metric property}
	\begin{lem}\label{qua}
		For $x_i\ge0$, $1\le i\le n$, we have
		\begin{align*}
			\log^{\rho}\left(1+\sum\limits_{i=1}^n x_i\right)\le \sum\limits_{i=1}^n \log^{\rho}(1+x_i)+C(\rho)\log^{\rho}n,
		\end{align*}
		where $C(\rho)>0$   is some constant depending  only on $\rho>1$.
	\end{lem}
	\begin{proof}
	For a detailed proof, we refer to \cite{SWY25}.
	\end{proof}
	\begin{rem}
		We have the {\it quasi-metric} property: for any $\bm x_i\in\Z^d$ ($1\leq i\leq n$),
		\begin{align}\label{quaeq}
			\log^{\rho}\left(1+\left\|\sum_{i=1}^n\bm x_i\right\|\right)\leq \sum_{i=1}^n\log^{\rho}(1+\|\bm x_i\|)+C(\rho) \log^{\rho} n.
		\end{align}
		%which is similar to the triangle inequality.
	\end{rem}
	
	\subsection{Extract lemma}
	Given the frequent need to extract $\log^{\rho}(1+x)$ from $\log^{\rho}(1+x-y)$ in the regime $x\gg y\gg1$, we establish the following lemma:
	\begin{lem}\label{el}
		Let $x>y>0$ and $1+x>2y$. Then
		\begin{align}\label{exl}
			\log^{\rho}(1+x-y)\ge\left(1-2\rho\frac{y}{(1+x)\log(1+x)}\right)\log^{\rho}(1+x).
		\end{align}
	\end{lem}
		\begin{proof}
		We refer to the Appendix \ref{app} for a detailed proof.
	\end{proof}

	\subsection{Hadamard  type estimate}
	We  need to estimate $\ell^2$-norm of the  inverse of some operator $\mathcal{S}_{\lg}$. By the Cramer's rule,   $\mathcal{S}_{\lg}^{-1}=(\det \mathcal{S}_{\lg})^{-1}S_{\lg}^*$, where $\mathcal{S}_{\lg}^*(\bm i,\bm j)$ is a determinant for $\bm i,\bm j\in\lg$. Therefore,  we can apply Hadamard's inequality to estimate $\ell^2$-norm of $\mathcal{S}_{\lg}^*$ and thus  that of $\mathcal{S}_{\lg}^{-1}$.
	\begin{lem}[Hadamard's estimate]\label{chi}
		Let $\mathcal{S}:\ \ell^2(\Z^d)\to\ell^2(\Z^d)$ be a linear operator, and $\lg\subset\Z^d$ a finite subset. Then for any $\bm i,\bm j\in\lg$, we have
		\begin{align*}
			|\ji \delta_{\bm i}, \mathcal{S}_{\lg}^* \delta_{\bm j}\jd|\le \left(\sup_{\bm x\in\lg}\sum_{\bm y\in\lg}|\ji \delta_{\bm x}, \mathcal{S}_{\lg}\delta_{\bm y}\jd|\right)^{\#\lg-1}.
		\end{align*}
		Moreover,
		\begin{align*}
			\|\mathcal{S}_{\lg}^*\|\le(\#\lg)\cdot\left(\sup_{\bm x\in\lg}\sum_{\bm y\in\lg}|\ji \delta_{\bm x}, \mathcal{S}_{\lg}\delta_{\bm y}\jd|\right)^{\#\lg-1}.
		\end{align*}
	\end{lem}
	\begin{proof}
		We refer to the Appendix \ref{app} for a detailed proof.
	\end{proof}

{\subsection{The Combes-Thomas type estimate}
	Aiming to deal with energies outside the spectrum, we need the Combes-Thomas type estimate.
	\begin{lem}\label{cte}
		 Let $ \boA$ be a self-adjoint operator in $ \ell^{2}(\Z^d)$ and $\boG(z):=(\boA- z)^{-1}$ be the Green's function of $\boA$ if it exists. If for some $\lambda>0$,
\begin{align}\label{s21}
S_{\lambda} := \sup_{\bm x\in\Z^d}\sum_{\bm y\ne\bm x}|\boA(\bm x, \bm y)|e^{\lambda\log^{\rho}(1+\|\bm x-\bm y\|)}< \infty,
\end{align}
then for energies $z$ not in the spectrum, $\mathscr{D}:=\dist(z,\sigma(\boA))>0$, we have
\begin{align}\label{s22}
	|\boG(z)(\bm x,\bm y)|\le  \frac{e^{-\lambda'\log^{\rho}(1+\|\bm x-\bm y\|)}}{\mathscr{D}- 2e^{\lambda' C(\rho)\log^{\rho}2}\cdot S_{\lambda'}}
\end{align}
for any $\lambda'\le\lambda$ with $2e^{\lambda' C(\rho)\log^{\rho}2}\cdot S_{\lambda'}<\mathscr{D}$ and all $\bm x,\bm y\in\Z^d$.
\end{lem}
\begin{proof}
	We refer to the Appendix \ref{app} for a detailed proof.
\end{proof}
	
	\section{Quantitative Green's function estimates}\label{qgfe}
	In this section, we fix
	\begin{align*}
		\theta\in\D_{R/2},\ E\in v(\D_{R/2}).
	\end{align*}
	Write
	\begin{align}\label{theta0}
		E=v(\theta_0)
	\end{align}
	for some  $\theta_0\in\D_{R/2}$. Consider
	\begin{align}\label{T}
		\mathcal{T}(E;\theta)=\mathcal{H}(\theta)-E=\mathcal{D}+\ep \mathcal{W}_\phi,
	\end{align}
	where
	\begin{align*}
		\mathcal{D}=\mathcal{D}_{\bm n}\delta_{\bm n,\bm n'},\ \mathcal{D}_{\bm n}=v(\theta+\bm n\cdot\bm\omega)-E.
	\end{align*}
	For simplicity, we may omit the dependence of $\mathcal{T}(E;\theta)$ on $E,\theta$ and that of $\mathcal{W}_\phi$ on $\phi$, respectively.

	Now we introduce  the statement of our main result on the  MSA type Green's function estimates. Define  the induction parameters
	\begin{align}\label{indpa}
		N_{s+1}=\left[e^{\left|\log\delta_{s}\right|^{\frac{1}{\rho'}}}\right],\ \delta_{s+1}=\delta_s^{10^{5\rho'}},\ \delta_0=\ep_0^{\frac{1}{10}}.
	\end{align}
	Thus
	\begin{align*}
		N_s^{10^5}-1\le N_{s+1}\le (N_s+1)^{10^5}
	\end{align*}
		and
		\begin{align*}
		e^{-\log^{\rho'}(N_{s+1}+1)}\le \delta_s\le e^{-\log^{\rho'}N_{s+1}}.
	\end{align*}
	
	We first introduce the following induction  statement.
	\begin{stm}[called $\mathscr{P}_s\ (s\ge1)$]\label{state}
	\end{stm}
	Let
	\begin{align*}
		Q_{s-1}^{\pm}&=\left\{\bm k\in P_{s-1}:\ \|\theta+\bm k\cdot\bm\omega\pm\theta_{s-1}\|_{\T}<\delta_{s-1}\right\},\ Q_{s-1}=Q_{s-1}^+\cup Q_{s-1}^-,\\
		\tilde{Q}_{s-1}^{\pm}&=\left\{\bm k\in P_{s-1}:\ \|\theta+\bm k\cdot\bm\omega\pm\theta_{s-1}\|_{\T}<\delta_{s-1}^{\frac{1}{100}}\right\},\ \tilde{Q}_{s-1}=\tilde{Q}_{s-1}^+\cup \tilde{Q}_{s-1}^-.
	\end{align*}
	We distinguish two cases:
	\begin{align}\label{c1}
		(\bm C1)_{s-1}:\ \dist(\tilde{Q}_{s-1}^-,Q_{s-1}^+)>100N_s^{10}
	\end{align}
	and
	\begin{align}\label{c2}
		(\bm C2)_{s-1}:\ \dist(\tilde{Q}_{s-1}^-,Q_{s-1}^{+})\le100N_s^{10}.
	\end{align}
	
	Let
	\begin{align*}
		\Z^d\ni \bm l_{s-1}=\left\{\begin{array}{lc}
			\bm 0,& \text{if \eqref{c1} holds true},\\
			\bm i_{s-1}-\bm j_{s-1}, & \text{if \eqref{c2} holds true},
		\end{array}\right.
	\end{align*}
	where $\bm i_{s-1}\in Q_{s-1}^+$ and $\bm j_{s-1}\in \tilde{Q}_{s-1}^-$ such that $\|\bm i_{s-1}-\bm j_{s-1}\|\le 100N_s^{10}$ in $(\bm C2)_{s-1}$. Set $\Omega_{\bm k}^0=\{\bm k\}\ (\bm k\in\Z^d)$. Let $\lg\subset\Z^d$ be a finite set. We say $\lg$ is $(s-1)$-$\good$ iff
	\begin{align*}
		\left\{\begin{array}{l}
			\bm k'\in Q_{s'},\ \tilde{\Omega}_{\bm k'}^{s'}\subset\lg,\ \tilde{\Omega}_{\bm k'}^{s'}\subset\Omega_{\bm k}^{s'+1}\Rightarrow\tilde{\Omega}_{\bm k}^{s'+1}\subset\lg,\ \text{for}\ s'<s-1,\\
			\{\bm k\in P_{s-1}:\ \tilde{\Omega}_{\bm k}^{s-1}\subset\lg\}\cap Q_{s-1}=\emptyset.
		\end{array}\right.
	\end{align*}
	
	Then we have
	
	$(\bm a)_s:$ There is $P_s\subset\frac{1}{2}\Z^d$ so that the following holds true. In the case of $(\bm C1)_{s-1}$, we have
	\begin{align}\label{as1}
		P_s=Q_{s-1}\subset\left\{\bm k\in\Z^d+\frac{1}{2}\sum_{i=0}^{s-1}\bm l_i:\ \min_{\sigma=\pm1}\|\theta+\bm k\cdot\bm\omega+\sigma\theta_{s-1}\|_{\T}<\delta_{s-1}\right\}.
	\end{align}
	For the case of $(\bm C2)_{s-1}$, we have
	\begin{align}\label{as2}
		\begin{array}{l}
			P_s\subset\left\{\bm k\in\Z^d+\frac{1}{2}\sum\limits_{i=0}^{s-1}\bm l_i:\ \|\theta+\bm k\cdot\bm\omega\|_{\T}<3\delta_{s-1}^{\frac{1}{100}}\right\},\\
			\text{or}\ P_s\subset\left\{\bm k\in\Z^d+\frac{1}{2}\sum\limits_{i=0}^{s-1}\bm l_i:\ \|\theta+\bm k\cdot\bm\omega+\frac{1}{2}\|_{\T}<3\delta_{s-1}^{\frac{1}{100}}\right\}.
		\end{array}
	\end{align}
	For every $\bm k\in P_s$, we can find resonant blocks $\Omega_{\bm k}^s,2\Omega_{\bm k}^s\subset\Z^d$ and the  enlarged resonant  block $\tilde{\Omega}_{\bm k}^s\subset\Z^d$ with the following properties. If \eqref{c1} holds true, then
	\begin{align*}
		&\lg_{N_s}(\bm k)\subset\Omega_{\bm k}^s\subset\lg_{N_s+50N_{s-1}^{100}}(\bm k),\\
		&\lg_{2N_s}(\bm k)\subset2\Omega_{\bm k}^s\subset\lg_{2N_s+50N_{s-1}^{100}}(\bm k),\\
		&\lg_{N_s^{10}}(\bm k)\subset\tilde{\Omega}_{\bm k}^s\subset\lg_{N_s^{10}+50N_{s-1}^{100}}(\bm k),
	\end{align*}
	and if \eqref{c2} holds true, then
	\begin{align*}
		&\lg_{100N_s^{10}}(\bm k)\subset\Omega_{\bm k}^s\subset\lg_{100N_s^{10}+50N_{s-1}^{100}}(\bm k),\\
		&\lg_{200N_s^{10}}(\bm k)\subset2\Omega_{\bm k}^s\subset\lg_{200N_s^{10}+50N_{s-1}^{100}}(\bm k),\\
		&\lg_{N_s^{100}}(\bm k)\subset\tilde{\Omega}_{\bm k}^s\subset\lg_{N_s^{100}+50N_{s-1}^{100}}(\bm k).
	\end{align*}
	These resonant blocks are constructed to satisfy the following two properties:\\
	$(\bm a1)_s$:
	\begin{align}\label{a1s}
		\left\{\begin{array}{l}
			\Omega_{\bm k}^s\cap\tilde{\Omega}_{\bm k'}^{s'}\ne\emptyset\ (s'<s)\Rightarrow\tilde{\Omega}_{\bm k'}^{s'}\subset\Omega_{\bm k}^s,\\
			2\Omega_{\bm k}^s\cap\tilde{\Omega}_{\bm k'}^{s'}\ne\emptyset\ (s'<s)\Rightarrow\tilde{\Omega}_{\bm k'}^{s'}\subset2\Omega_{\bm k}^s,\\
			\tilde{\Omega}_{\bm k}^s\cap\tilde{\Omega}_{\bm k'}^{s'}\ne\emptyset\ (s'<s)\Rightarrow\tilde{\Omega}_{\bm k'}^{s'}\subset\tilde{\Omega}_{\bm k}^s,\\
			\dist(\tilde{\Omega}_{\bm k}^s,\tilde{\Omega}_{\bm k'}^s)>10\diam(\tilde{\Omega}_{\bm k}^s)\ \text{for}\ \bm k\ne\bm k'\in P_s.
		\end{array}\right.
	\end{align}\\
	$(\bm a2)_s$:  The translation of $\tilde{\Omega}_{\bm k}^s$
	\begin{align*}
		\tilde{\Omega}_{\bm k}^s-\bm k\subset\Z^d+\frac{1}{2}\sum_{i=0}^{s-1}\bm l_i,
	\end{align*}
	is independent of $\bm k\in P_s$ and symmetrical about the origin.
	
	We denote
	\begin{align}\label{zetas}
		\zeta_s=\diam(\Omega_{\bm k}^s), \ \tilde{\zeta}_s=\diam(\tilde{\Omega}_{\bm k}^s).
	\end{align}
	%We sometimes simplify $$ and $\diam(\tilde{\Omega}_{\bm k}^s)$ by $\zeta_s$ and $\tilde{\zeta}_s$ respectively .
	
	$(\bm b)_s$:  $Q_{s-1}$ is covered by $\Omega_{\bm k}^s\ (\bm k\in P_s)$ in the sense that for every $\bm k'\in Q_{s-1}$, there exists a $\bm k\in P_s$ such that
	\begin{align}\label{311}
		\tilde{\Omega}_{\bm k'}^{s-1}\subset \Omega_{\bm k}^s.
	\end{align}
	
	$(\bm c)_s$:  For each $\bm k\in P_s$, $\tilde{\Omega}_{\bm k}^s$ contains a subset $A_{\bm k}^s\subset\Omega_{\bm k}^s$ with $\# A_{\bm k}^s\le 2^s$ such that $\tilde{\Omega}_{\bm k}^s\setminus A_{\bm k}^s$ is $(s-1)$-$\good$. Moreover, $A_{\bm k}^s-\bm k$ is independent of $\bm k$ and is symmetrical about the origin.
	
	$(\bm d)_s$:  There is a $\theta_s=\theta_s(E)\in\C$ with the following properties. Let $\bm k\in P_s$. Replacing $\theta+\bm n\cdot\bm\omega$ by $z+(\bm n-\bm k)\cdot\bm\omega$ for $\bm n\in\Z^d$ and restricting $z$ in
	\begin{align}\label{zs}
		\left\{z\in\C:\ \min_{\sigma=\pm1}\|z+\sigma\theta_s\|_{\T}<\delta_s^{\frac{1}{10^4}}\right\},
	\end{align}
	\begin{rem}
		For notational simplicity, we have made the following simplification: for given $\lg\subset\Z^d$, $\bm k\in \frac{1}{2}\Z^d$ and $\boA(z):\ell^2(\Z^d)\rightarrow\ell^2(\Z^d)$, we define $(\boA(z))_{\lg-\bm k}$ as
		\begin{align*}
			(\boA(z))_{\lg-\bm k}(\bm i,\bm j)=(\boA(z-\bm k\cdot\bm\omega))_{\lg}(\bm i+\bm k,\bm j+\bm k)\text{ for all $\bm i,\bm j\in\lg-\bm k$},
		\end{align*} 
		where $\bm i+\bm k,\bm j+\bm k\in\lg\subset\Z^d$.
	\end{rem}
	We write %see that $\mathcal{T}_{\tilde{\Omega}_{\bm k}^s}$ becomes
	\begin{align*}
		\mathcal{M}_s(z)=(\mathcal{T}(z))_{\tilde{\Omega}_{\bm k}^s-\bm k}=((v(z+\bm n\cdot\bm\omega)-E)\delta_{\bm n,\bm n'}+\ep \mathcal{W})_{\tilde{\Omega}_{\bm k}^s-\bm k}.
	\end{align*}
	Then $(\mathcal{M}_s(z))_{(\tilde{\Omega}_{\bm k}^s\setminus A_{\bm k}^s)-\bm k}$ is invertible and we can define the Schur complement
	\begin{align*}
		\mathcal{S}_s(z)&=(\mathcal{M}_s(z))_{A_{\bm k}^s-\bm k}-\left(\mathcal{R}_{A_{\bm k}^s-\bm k}\mathcal{M}_s(z)\mathcal{R}_{(\tilde{\Omega}_{\bm k}^s\setminus A_{\bm k}^s)-\bm k}((\mathcal{M}_s(z))_{(\tilde{\Omega}_{\bm k}^s\setminus A_{\bm k}^s)-\bm k})^{-1}\right.\\
		&\ \ \left.\times\ \mathcal{R}_{(\tilde{\Omega}_{\bm k}^s\setminus A_{\bm k}^s)-\bm k}\mathcal{M}_s(z)\mathcal{R}_{A_{\bm k}^s-\bm k}\right).
	\end{align*}
	Moreover, if $z$ belongs to the set in  \eqref{zs}, then we have
	\begin{align}\label{ss}
		\max_{\bm x\in A_{\bm k}^s-\bm k}\sum_{\bm y\in A_{\bm k}^s-\bm k}|\mathcal{S}_s(z)(\bm x,\bm y)|<2|v|_R+\sum_{l=0}^{s-1}\delta_l<4|v|_R
	\end{align}
	and
	\begin{align}\label{detss}
		\left|\det \mathcal{S}_s(z)\right|\ge\delta_{s-1}\|z-\theta_s\|_{\T}\cdot\|z+\theta_s\|_{\T}.
	\end{align}
	Combining the Schur complement lemma (cf. Lemma \ref{scl}), we get
	\begin{align}
		\label{tb0}\|\mathcal{T}_{\tilde{\Omega}_{\bm k}^s}^{-1}\|&<\delta_{s-1}^{-2}\|\theta+\bm k\cdot\bm\omega-\theta_s\|_{\T}^{-1}\cdot\|\theta+\bm k\cdot\bm\omega+\theta_s\|_{\T}^{-1}.
	\end{align}
	
	$(\bm e)_s$:  Let
	\begin{align*}
		Q_{s}^{\pm}&=\{\bm k\in P_{s}:\ \|\theta+\bm k\cdot\bm\omega\pm\theta_{s}\|_{\T}<\delta_{s}\},\ Q_{s}=Q_{s}^+\cup Q_{s}^-,\\
		\tilde{Q}_{s}^{\pm}&=\{\bm k\in P_{s}:\ \|\theta+\bm k\cdot\bm\omega\pm\theta_{s}\|_{\T}<\delta_{s}^{\frac{1}{100}}\},\ \tilde{Q}_{s}=\tilde{Q}_{s}^+\cup \tilde{Q}_{s}^-,
	\end{align*}
	and
		\begin{align}\label{alphas}
		\alpha_0=\frac{3}{4}\alpha,\ \alpha_s=\alpha_{s-1}\left(1-\frac{50\times 10^{5\rho'}}{\alpha\log^{\rho-\rho'}N_s}\right).
	\end{align}
	Thus,
	\begin{align*}
		\alpha_s \searrow\alpha_{\infty}\ge\frac{\alpha}{2}.
	\end{align*}
	If $\bm k\in P_s\setminus Q_s$, we have
		\begin{align}
					\label{PsQsnorm}&\|\mathcal{T}_{\tO_{\bm k}^s}^{-1}\|\le\delta_{s-1}^{-2}\delta_s^{-2}<\delta_{s}^{-3},\\
		\label{PsQsdecay}&|\mathcal{T}_{\tO_{\bm k}^s}^{-1}(\bm x,\bm y)|<e^{-\alpha'_{s-1}\log^{\rho}(1+\|\bm x-\bm y\|)}\text{ for $\|\bm x-\bm y\|>\frac{\tz_s}{10}$},
	\end{align}
	where
	\begin{align}\label{alpha's}
	\alpha'_{s-1}=\alpha_{s-1}\left(1-\frac{20\times 10^{5\rho'}}{\alpha\log^{\rho-\rho'}N_{s}}\right)>\alpha_{s}.
\end{align}
	
	We say a finite set $\lg\subset\Z^d$ is $s$-$\good$ iff
	\begin{align*}
		\left\{\begin{array}{l}
			\bm k'\in Q_{s'},\ \tilde{\Omega}_{\bm k'}^{s'}\subset\lg,\ \tilde{\Omega}_{\bm k'}^{s'}\subset\Omega_{\bm k}^{s'+1}\Rightarrow\tilde{\Omega}_{\bm k}^{s'+1}\subset\lg\ \text{for}\ s'<s,\\
			\{\bm k\in P_s:\ \tilde{\Omega}_{\bm k}^s\subset\lg\}\cap Q_s=\emptyset.
		\end{array}\right.
	\end{align*}
	Assume that $\lg$ is $s$-$\good$. Then
	\begin{align}
		\label{tsgnorm}&\|\mathcal{T}_\lg^{-1}\|\le2\delta_{s-1}^{-3}\sup_{\{\bm k\in P_s:\ \tO_{\bm k}^s\subset\lg\}}\left(\|\theta+\bm k\cdot\bm\omega-\theta_s\|_{\T}^{-1}\cdot\|\theta+\bm k\cdot\bm\omega+\theta_s\|_{\T}^{-1}\right)<\delta_{s}^{-3},\\
		\label{tsgdecay}&|\mathcal{T}_{\lg}^{-1}(\bm x,\bm y)|<e^{-\alpha_s\log^{\rho}(1+\|\bm x-\bm y\|)}\text{ for $\|\bm x-\bm y\|>10\tz_s$}.
	\end{align}

	$(\bm f)_s$:  We have
	\begin{align}\label{fs}
		\left\{\bm k\in\Z^d+\frac{1}{2}\sum_{i=0}^{s-1}\bm l_i:\ \min_{\sigma=\pm1}\|\theta+\bm k\cdot\bm\omega+\sigma\theta_s\|_{\T}<10\delta_s^{\frac{1}{100}}\right\}\subset P_s.
	\end{align}
	
	The main theorem of this section is
	\begin{thm}\label{ind}
		Let $\bm\omega\in DC_{\tau,\g}$. Then there is some $\ep_0=\ep_0(d,\tau,\g,v,R,\alpha,\rho,\rho')>0$ so that for $0<|\ep|<\ep_0$, the statement $\mathscr{P}_s$ holds true  for all $s\ge1$.
	\end{thm}
	
	The following three  subsections are devoted to proving  Theorem \ref{ind}.
	
	\subsection{The initial step}
	Recalling $v(\theta_0)=E$, $\mathcal{D}_{\bm n}=v(\theta+\bm n\cdot\bm\omega)-E$ and \eqref{vdefn}, we have
	\begin{align*}
		|\mathcal{D}_{\bm n}|=|v(\theta+\bm n\cdot\bm\omega)-v(\theta_0)|\ge \kappa_1
		\|\theta+\bm n\cdot\bm\omega+\theta_0\|_{\T}\cdot\|\theta+\bm n\cdot\bm\omega-\theta_0\|_{\T}.
	\end{align*}
	
	Denote %$\delta_0=\ep_0^{\frac{1}{30}}$ and
	\begin{align*}
		P_0=\Z^d,\ Q_0=\{\bm k\in P_0:\ \min(\|\theta+\bm k\cdot\bm\omega+\theta_0\|_{\T},\|\theta+\bm k\cdot\bm\omega-\theta_0\|_{\T})<\delta_0\}.
	\end{align*}
	
	We say a finite set $\lg\subset \Z^d$ is $0$-$\good$ iff $\lg\cap Q_0=\emptyset$. Then we have
	
	\begin{lem}\label{0g}
		If the set $\lg\subset\Z^d$ is $0$-$\good$, we have
		\begin{align}\label{0gl2}
			\|\mathcal{T}_{\lg}^{-1}\|\le2\kappa_1^{-1}\delta_0^{-2},
		\end{align}
and
		\begin{align}\label{0gde}
			|\boT_{\lg}^{-1}(\bm x,\bm y)|\le e^{-\frac{3}{4}\alpha\log^{\rho}(1+\|\bm x-\bm y\|)}\ \text{for $\bm x\ne \bm y$}.
		\end{align}
	\end{lem}
	\begin{proof}
		Since  $\lg$ is $0$-good, for $\bm n\in\lg$,  we have
		\begin{align*}
			|\mathcal{D}_{\bm n}|\ge\kappa_1\|\theta+\bm n\cdot\bm\omega+\theta_0\|_{\T}\cdot\|\theta+\bm n\cdot\bm\omega-\theta_0\|_{\T}\ge\kappa_1\delta_0^{2}.
		\end{align*}
		Then
		\begin{align}\label{D-1}
			\|\mathcal{D}_{\lg}^{-1}\|\le\kappa_1^{-1}\delta_0^{-2},
		\end{align}
		which implies (since $|\ep|<\ep_0$ and $\ep_0=\delta_0^{10}$)
		\begin{align*}
			\|\ep \mathcal{D}_{\lg}^{-1}\mathcal{W}_{\lg}\|\le |\ep|\cdot\|\mathcal{D}_{\lg}^{-1}\|\cdot\|\mathcal{W}_{\lg}\|<\frac{1}{2}.
		\end{align*}
		Thus
		\begin{align*}
			(\mathcal{I}_{\lg}+\ep \mathcal{D}_{\lg}^{-1}\mathcal{W}_{\lg})^{-1}=\sum_{i=0}^{\infty}(-\ep \mathcal D_{\lg}^{-1}\mathcal W_{\lg})^i
		\end{align*}
		exist and
		\begin{align*}
			\|(\mathcal{I}_{\lg}+\ep \mathcal{D}_{\lg}^{-1}\mathcal{W}_{\lg})^{-1}\|\le \sum_{i=0}^{\infty}|(|\ep|\cdot\|\mathcal{D}_{\lg}^{-1}\|\cdot\|\mathcal{W}_{\lg}\|)^i<2.
		\end{align*}
	At this time, we obtain
		\begin{align*}
			\mathcal{T}_{\lg}^{-1}=(\mathcal{I}_\lg+\ep \mathcal{D}_{\lg}^{-1}\mathcal{W}_{\lg})^{-1}\mathcal{D}_\lg^{-1}
		\end{align*}
	exist and
	\begin{align*}
		\|	\mathcal{T}_{\lg}^{-1}\|\le \|(\mathcal{I}_\lg+\ep \mathcal{D}_{\lg}^{-1}\mathcal{W}_{\lg})^{-1}\|\cdot\|\mathcal{D}_\lg^{-1}\|\le 2\kappa_1^{-1}\delta_0^{-2}.
	\end{align*}
	Moreover, if $\bm x\ne\bm y$, we learn that
		\begin{align}
		\nonumber	|\mathcal{T}_\lg^{-1}(\bm x,\bm y)|&\le \sum_{i=1}^{\infty}|\ep|^i\cdot|((\boD_{\lg}^{-1}\boW_{\lg})^{i}\boD_{\lg}^{-1})(\bm x,\bm y)|\\
	\nonumber	&\le |\ep|\cdot\|\boD_{\lg}^{-1}\|^{2}\cdot|\boW_{\lg}(\bm x,\bm y)|\\
	\label{5801}&\ +\left(\sum_{i=2}^{\infty}|\ep|^i\cdot\|\boD_{\lg}^{-1}\|^{i+1}\cdot|\boW_{\lg}^i(\bm x,\bm y)|\right).
		\end{align}
		Next, we will control $|\boW_{\lg}^i(\bm x,\bm y)|$ for $i\ge2$. From \eqref{wphi} and \eqref{quaeq}, we have
		\begin{align}
			\nonumber|\boW_{\lg}^i(\bm x,\bm y)|&\le \sum_{\bm k_1,\bm k_2,\cdots,\bm k_{i-1}\in\lg}|\boW_{\lg}(\bm x,\bm k_1)|\cdot|\boW_{\lg}(\bm k_1,\bm k_2)|\cdots|\boW_{\lg}(\bm k_{i-1},\bm y)|\\
			\nonumber&\le \sum_{\bm k_1,\bm k_2,\cdots,\bm k_{i-1}\in\lg} e^{-\alpha\log^{\rho}(1+\|\bm x-\bm k_1\|)}\cdot e^{-\alpha\log^{\rho}(1+\|\bm k_1-\bm k_2\|)}\cdots e^{-\alpha\log^{\rho}(1+\|\bm k_{i-1}-\bm y\|)}\\
			\nonumber&\le \sum_{\bm k_1,\bm k_2,\cdots,\bm k_{i-1}\in\lg} \left(\vphantom{e^{-\frac{\alpha}{4}\log^{\rho}(1+\|\bm x-\bm k_1\|)}}e^{-\frac{\alpha}{4}\log^{\rho}(1+\|\bm x-\bm k_1\|)}\cdot e^{-\frac{\alpha}{4}\log^{\rho}(1+\|\bm k_1-\bm k_2\|)}\cdots e^{-\frac{\alpha}{4}\log^{\rho}(1+\|\bm k_{i-1}-\bm y\|)}\right.\\
			\nonumber&\ \ \left.\vphantom{e^{-\frac{\alpha}{4}\log^{\rho}(1+\|\bm x-\bm k_1\|)}}\times e^{-\frac{3}{4}\alpha\log^{\rho}(1+\|\bm x-\bm k_1\|)}\cdot e^{-\frac{3}{4}\alpha\log^{\rho}(1+\|\bm k_1-\bm k_2\|)}\cdots e^{-\frac{3}{4}\alpha\log^{\rho}(1+\|\bm k_{i-1}-\bm y\|)}\right)\\
			\label{W^i}&\le \left(D\left(\frac{\alpha}{4}\right)\right)^{i-1}e^{-\frac{3}{4}\alpha\log^{\rho}(1+\|\bm x-\bm y\|)+\frac{3}{4}\alpha C(\rho)\log^{\rho}},
		\end{align}
		where $e^{-\frac{\alpha}{4}\log^{\rho}(1+\|\bm k_{i-1}-\bm y\|)}\le1$ and \eqref{Det} are used.
		 Since $|\ep|<\ep_0$, $\ep_0=\delta_0^{10}$, \eqref{wphi}, \eqref{D-1} and \eqref{W^i}, we obtain
		 \begin{align}\label{5802}
		 	|\ep|\cdot\|\boD_{\lg}^{-1}\|^{2}\cdot|\boW_{\lg}(\bm x,\bm y)|\le\frac{1}{2}e^{-\alpha\log^{\rho}(1+\|\bm x-\bm y\|)},
		 \end{align}
		 and
		\begin{align}
		\nonumber	&\ \ \sum_{i=2}^{\infty}|\ep|^i\cdot\|\boD_{\lg}^{-1}\|^{i+1}\cdot|\boW_{\lg}^i(\bm x,\bm y)|\\
			 \nonumber\le&\ \ \left(\sum_{i=2}^{\infty}|\ep|^i\cdot\|\boD_{\lg}^{-1}\|^{i+1}\cdot\left(D\left(\frac{\alpha}{4}\right)\right)^{i-1}e^{\frac{3}{4}\alpha C(\rho)\log^{\rho}i}\right)e^{-\frac{3}{4}\alpha\log^{\rho}(1+\|\bm x-\bm y\|)}\\
			 \label{5803}\le &\ \ \frac{1}{2}e^{-\frac{3}{4}\alpha\log^{\rho}(1+\|\bm x-\bm y\|)}.
		\end{align}
		Combining \eqref{5801}, \eqref{5802} and \eqref{5803} gives
		\begin{align*}
			|\mathcal{T}_\lg^{-1}(\bm x,\bm y)|\le e^{-\frac{3}{4}\alpha\log^{\rho}(1+\|\bm x-\bm y\|)}\ \text{for $\bm x\ne\bm y$}.
		\end{align*}
		We complete this proof.
	\end{proof}
	\begin{rem}
		If we directly obtain
		\begin{align*}
			&\ \ \sum_{\bm k_1,\bm k_2,\cdots,\bm k_{i-1}\in\lg}|\boW_{\lg}(\bm x,\bm k_1)|\cdot|\boW_{\lg}(\bm k_1,\bm k_2)|\cdots|\boW_{\lg}(\bm k_{i-1},\bm y)|\\
			\le&\ \ \sum_{\bm k_1,\bm k_2,\cdots,\bm k_{i-1}\in\lg}e^{-\alpha\log^{\rho}(1+\|\bm x-\bm y\|)+\alpha C(\rho)\log^{\rho}i}\\
			\le &\ \ (\#\lg)^{i-1}e^{-\alpha\log^{\rho}(1+\|\bm x-\bm y\|)+\alpha C(\rho)\log^{\rho}i}
		\end{align*}
	in \eqref{W^i} by using \eqref{wphi} and \eqref{quaeq}, the relevant series in \eqref{5803} will diverge when $\#\lg$ is sufficiently large. In practice, we utilize the partial decay of $\boW$ to overcome this difficulty, and this technique will be widely applied in subsequent sections. The loss term $\frac{\alpha}{4}$ is not essential and can be replaced by any constant $c \in (0, \alpha)$.
	\end{rem}
	
	Based on Lemma \ref{0g}, we derive the following frequently used corollary:
	\begin{cor}\label{ite0g}
		Let $\lg',\lg\subset\Z^d$ be finite sets. If $\lg'\subset\lg$ is $0$-{\rm good}, $\bm u\in\lg'$ and $\bm v\in\lg\setminus\lg'$, then there is some $\bm u'\in\lg\setminus\lg'$ such that
		\begin{align}\label{ite0geq1}
			|\boT_{\lg}^{-1}(\bm u,\bm v)|\le (\#\lg')\cdot e^{-\frac{3}{4}\alpha\log^{\rho}(1+\|\bm u-\bm u'\|)}\cdot|\boT_{\lg}^{-1}(\bm u',\bm v)|\cdot\delta_0^{-3}.
		\end{align}
		Morover, if $\lg'=\lg_{\frac{1}{2}N_1}(\bm u)\cap\lg$ is $0$-{\rm good}, then
			\begin{align}\label{ite0g0}
				|\boT_{\lg}^{-1}(\bm u,\bm v)|\le e^{-\frac{3}{4}\alpha\left(1-\frac{6\times 10^{5\rho'}}{\alpha\log^{\rho-\rho'}N_1}\right)\log^{\rho}(1+\|\bm u-\bm u'\|)}\cdot|\boT_{\lg}^{-1}(\bm u',\bm v)|.
			\end{align}
	\end{cor}
	\begin{rem}
		If $\bm v\in \lg'$ and $\bm u \in\lg\setminus\lg'$, we have a similar argument: there is some $\bm v'\in \lg\setminus\lg'$ such that
		\begin{align}\label{iteeq12}
			|\boT_{\lg}^{-1}(\bm u,\bm v)|\le (\#\lg')\cdot e^{-\frac{3}{4}\alpha\log^{\rho}(1+\|\bm v-\bm v'\|)}\cdot|\boT_{\lg}^{-1}(\bm u,\bm v')|\cdot\delta_0^{-3}.
		\end{align}
	\end{rem}
	\begin{proof}
		Using the resolvent identity implies
		\begin{align*}
			\boT_{\lg}^{-1}(\bm u,\bm v)=-\ep \sum_{\bm w\in\lg'\atop \bm w'\in\lg\setminus\lg'}\boT_{\lg'}^{-1}(\bm u,\bm w)\cdot\boW(\bm w,\bm w')\cdot\boT_{\lg}^{-1}(\bm w',\bm v),
		\end{align*}
		and then
		\begin{align}
			\nonumber|\boT_{\lg}^{-1}(\bm u,\bm v)|&\le \max_{\bm w\in\lg'\atop \bm w'\in\lg\setminus\lg'}\left|\boT_{\lg'}^{-1}(\bm u,\bm w)\cdot\boW(\bm w,\bm w')\cdot e^{\frac{\alpha}{10}\log^{\rho}(1+\|\bm w-\bm w'\|)}\cdot\boT_{\lg}^{-1}(\bm w',\bm v)\right|\\
			\nonumber&\ \ \times\left(\sum_{\bm w\in\lg'\atop \bm w'\in\lg\setminus\lg'}e^{-\frac{\alpha}{10}\log^{\rho}(1+\|\bm w-\bm w'\|)}\right)\\
			\nonumber&\le \max_{\bm w\in\lg'\atop \bm w'\in\lg\setminus\lg'}\left|\boT_{\lg'}^{-1}(\bm u,\bm w)\cdot\boW(\bm w,\bm w')\cdot e^{\frac{\alpha}{10}\log^{\rho}(1+\|\bm w-\bm w'\|)}\cdot\boT_{\lg}^{-1}(\bm w',\bm v)\right|\\
			\label{0gri}&\ \ \times(\#\lg')\cdot D\left(\frac{\alpha}{10}\right),
		\end{align}
		where $D\left(\frac{\alpha}{10}\right)$ is defined in \eqref{Det}. Since $\lg'$ and $\lg\setminus\lg'$ are finite sets, there are $\bm u^*\in\lg'$ and $\bm u'\in \lg\setminus\lg'$ such that
		\begin{align}
			\nonumber&|\boT_{\lg'}^{-1}(\bm u,\bm u^*)\cdot\boW(\bm u^*,\bm u')\cdot e^{\frac{\alpha}{10}\log^{\rho}(1+\|\bm u^*-\bm u'\|)}\cdot\boT_{\lg}^{-1}(\bm u',\bm v)|\\
			=&\ \
			\label{0gmax}\max_{\bm w\in\lg'\atop \bm w'\in\lg\setminus\lg'}|\boT_{\lg'}^{-1}(\bm u,\bm w)\cdot\boW(\bm w,\bm w')\cdot e^{\frac{\alpha}{10}\log^{\rho}(1+\|\bm w-\bm w'\|)}\cdot\boT_{\lg}^{-1}(\bm w',\bm v)|.
		\end{align}
		If $\bm u^*=\bm u$, from $\lg'$ is $0$-good (cf. \eqref{0gl2}) and \eqref{wphi}, we can get
		\begin{align}
			\nonumber&\ \ |\boT_{\lg'}^{-1}(\bm u,\bm u^*)\cdot\boW(\bm u^*,\bm u')\cdot e^{\frac{\alpha}{10}\log^{\rho}(1+\|\bm u^*-\bm u'\|)}|\\
			\label{0gTW01}\le &\ \ 2\kappa_1^{-1}\delta_0^{-2}\cdot e^{-\frac{9}{10}\alpha\log^{\rho}(1+\|\bm u-\bm u'\|)}.
		\end{align}
		If $\bm u^*\ne \bm u$, since $\lg'$ is $0$-good (cf. \eqref{0gde}), \eqref{wphi} and \eqref{quaeq}, we have
		\begin{align}
			\nonumber&\ \ |\boT_{\lg'}^{-1}(\bm u,\bm u^*)\cdot\boW(\bm u^*,\bm u')\cdot e^{\frac{\alpha}{10}\log^{\rho}(1+\|\bm u^*-\bm u'\|)}|\\
			\nonumber\le&\ \  2e^{-\frac{3}{4}\alpha\log^{\rho}(1+\|\bm u-\bm u^*\|)}\cdot e^{-\frac{9}{10}\alpha\log^{\rho}(1+\|\bm u^*-\bm u'\|)}\\
			\label{0gTW02}\le&\ \  2e^{-\frac{3}{4}\alpha\log^{\rho}(1+\|\bm u-\bm u'\|)+\frac{3}{4}\alpha C(\rho)\log^{\rho}2}.
		\end{align}
		Combining \eqref{0gTW01}, \eqref{0gTW02} and $\delta_0\ll1$ gives
		\begin{align}
			\nonumber&\ \ |\boT_{\lg'}^{-1}(\bm u,\bm u^*)\cdot\boW(\bm u^*,\bm u')\cdot e^{\frac{\alpha}{10}\log^{\rho}(1+\|\bm u^*-\bm u'\|)}|\\
			\label{0gTW03}\le&\ \ 2\kappa_1^{-1}\delta_0^{-2}\cdot e^{-\frac{3}{4}\alpha\log^{\rho}(1+\|\bm u-\bm u'\|)}.
		\end{align}
		By \eqref{0gri}, \eqref{0gmax}, \eqref{0gTW03} and $\delta_0\ll1$, we obtain
		\begin{align}\label{51101}
			|\boT_{\lg}^{-1}(\bm u,\bm v)|\le (\#\lg')\cdot e^{-\frac{3}{4}\alpha\log^{\rho}(1+\|\bm u-\bm u'\|)}\cdot|\boT_{\lg}^{-1}(\bm u',\bm v)|\cdot\delta_0^{-3}.
		\end{align}
		
		Morover, if $\lg'=\lg_{\frac{1}{2}N_1}(\bm u)\cap\lg$ is 0-good, $\|\bm u-\bm u'\|\ge\frac{1}{2}N_1$ (since $\bm u'\in\lg\setminus\lg'$) and $(\#\lg')\le (N_1+1)^d$. Therefore, by \eqref{51101}, \eqref{indpa} and $N_1\gg1$, we have
		\begin{align*}
				|\boT_{\lg}^{-1}(\bm u,\bm v)|\le e^{-\frac{3}{4}\alpha\left(1-\frac{6\times 10^{5\rho'}}{\alpha\log^{\rho-\rho'}N_1}\right)\log^{\rho}(1+\|\bm u-\bm u'\|)}\cdot|\boT_{\lg}^{-1}(\bm u',\bm v)|.
		\end{align*}
		We finish this proof.
	\end{proof}

	\subsection{Verification of $\mathscr{P}_1$}\label{vo1}
	If $\lg\cap Q_0\ne\emptyset$, the Neumann series argument from the previous subsection fails. We therefore estimate $\mathcal{T}_\lg^{-1}$ using the resolvent identity argument whenever  $\lg$ is $1$-$\good$.
	
	Recall that
	\begin{align*}
		N_{1}=\left[e^{\left|\log\delta_{0}\right|^{\frac{1}{\rho'}}}\right]
	\end{align*}
	and define (cf. \eqref{theta0})
	\begin{align*}
		Q_0^{\pm}&=\left\{\bm k\in\Z^d:\ \|\theta+\bm k\cdot\bm\omega\pm\theta_0\|_\T<\delta_0\right\},\ Q_0=Q_0^+\cup Q_0^-,\\
		\tilde{Q}_0^{\pm}&=\left\{\bm k\in\Z^d:\ \|\theta+\bm k\cdot\bm\omega\pm\theta_0\|_\T<\delta_0^{\frac{1}{100}}\right\},\ \tilde{Q}_0=\tilde{Q}_0^+\cup \tilde{Q}_0^-.
	\end{align*}
	
	We distinguish the verification into three steps.
	
	\begin{itemize}
		\item[\textbf{Step 1}]: \textbf{Estimates of $\|\mathcal{T}_{\tilde{\Omega}_{\bm k}^1}^{-1}\|$}.
	\end{itemize}
	In this step, we will find $\theta_1=\theta_1(E)$ so that
	\begin{align*}
		\|\mathcal{T}_{\tilde{\Omega}_{\bm k}^1}^{-1}\|&<\delta_0^{-2}\|\theta+\bm k\cdot\bm\omega-\theta_1\|_{\T}^{-1}\cdot\|\theta+\bm k\cdot\bm\omega+\theta_1\|_{\T}^{-1}.
	\end{align*}
	We  again divide the discussions   into two cases.
	\begin{itemize}
		\item[\textbf{Case 1}]:  {The case $(\bm C1)_0$ occurs}, i.e.,
	\end{itemize}
	\begin{align}\label{-+1}
		\text{dist}\left(\tilde{Q}_0^-,Q_0^+\right)>100N_1^{10}.
	\end{align}
	\begin{rem}\label{dsq1}
		We have in fact
		\begin{align*}
			\text{dist}\left(\tilde{Q}_0^-,Q_0^+\right)=\text{dist}\left(\tilde{Q}_0^+,Q_0^-\right).
		\end{align*}
		Thus \eqref{-+1} also implies
		\begin{align*}
			\text{dist}\left(\tilde{Q}_0^+,Q_0^-\right)>100N_1^{10}.
		\end{align*}
		We refer to \cite{CSZ24a} for a detailed proof.
	\end{rem}
	Assuming \eqref{-+1} holds true, we define
	\begin{align}\label{P11}
		P_1=Q_0=\{\bm k\in\Z^d:\ \min(\|\theta+\bm k\cdot\bm\omega+\theta_0\|_{\T},\|\theta+\bm k\cdot\bm\omega-\theta_0\|_{\T})<\delta_0\}.
	\end{align}
	Associate each  $\bm k\in P_1$ with  $\Omega_{\bm k}^1:=\lg_{N_1}(\bm k)$, $2\Omega_{\bm k}^{1}=\lg_{2N_1}(\bm k)$ and  $\tilde{\Omega}_{\bm k}^1:=\lg_{N_1^{10}}(\bm k)$. Then $\tilde{\Omega}_{\bm k}^1-\bm k\subset\Z^d$ is independent of $\bm k\in P_1$ and symmetrical about the origin. If $\bm k\ne\bm k'\in P_1$,  then
	\begin{align*}
		\|\bm k-\bm k'\|\ge\min\left(100N_1^{10},\left(\frac{\g}{2\delta_0}\right)^{\frac{1}{\tau}}\right)\ge100N_1^{10}.
	\end{align*}
	Thus
	\begin{align*}
		\text{dist}\left(\tilde{\Omega}_{\bm k}^1,\tilde{\Omega}_{\bm k'}^1\right)>10\tz_1\ \text{for}\ \bm k\ne \bm k'\in P_1.
	\end{align*}
	For $\bm k\in Q_0^-$, we consider
	\begin{align*}
		\mathcal{M}_1(z):=(\mathcal{T}(z))_{\tilde{\Omega}_{\bm k}^1-\bm k}=\left((v(z+\bm n\cdot\bm\omega)-E)\delta_{\bm n,\bm n'}+\ep \mathcal{W}\right)_{\tilde{\Omega}_{\bm k}^1-\bm k}
	\end{align*}
	defined in
	\begin{align}\label{z0-}
		\left\{z\in\C:\ |z-\theta_0|\le\delta_0^{\frac{1}{10}}\right\}.
	\end{align}
	For $\bm n\in\left(\tilde{\Omega}_{\bm k}^1-\bm k\right)\setminus\{\bm 0\}$, we have for $0<\delta_0\ll 1$,
	\begin{align*}
		\|z+\bm n\cdot\bm\omega-\theta_0\|_{\T}&\ge\|\bm n\cdot\bm\omega\|_{\T}-|z-\theta_0|\\
		&\ge\frac{\g}{(N_1^{10})^{\tau}}-\delta_0^{\frac{1}{10}}\\
		&>\delta_0^{\frac{1}{10^4}}\ \text{(since \eqref{indpa})}.
	\end{align*}
	For $\bm n\in\tilde{\Omega}_{\bm k}^1-\bm k$, we have (since $\dist(\bm k,\tilde{Q}_0^{+})>100N_1^{10}$ and $\bm k\in Q_0^{-}$)
	\begin{align*}
		\|z+\bm n\cdot\bm\omega+\theta_0\|_{\T}&\ge\|\theta+(\bm n+\bm k)\cdot\bm\omega+\theta_0\|_{\T}-|z-\theta_0|-\|\theta+\bm k\cdot\bm\omega-\theta_0\|_{\T}\\
		&\ge\delta_0^{\frac{1}{100}}-\delta_0^{\frac{1}{10}}-\delta_0>\frac{1}{2}\delta_0^{\frac{1}{100}}.
	\end{align*}
	Hence for $\bm n\in\left(\tilde{\Omega}_{\bm k}^1-\bm k\right)\setminus\{\bm 0\}$,
	\begin{align*}
		|v(z+\bm n\cdot\bm\omega)-E|\ge \frac{\kappa_1}{2}\delta_{0}^{\frac{1}{100}+\frac{1}{10^4}}> \delta_0^{\frac{1}{50}}\gg\ep.
	\end{align*}
	By Neumann series argument,  we have
	\begin{align}\label{m1z01}
		\left\|\left((\mathcal{M}_1(z))_{(\tilde{\Omega}_{\bm k}^1-\bm k)\setminus\{\bm 0\}}\right)^{-1}\right\|\le2\delta_0^{-\frac{1}{50}}.
	\end{align}
	We now apply the Schur complement lemma to derive the desired estimates. According to Lemma \ref{scl}, the inverse $(\mathcal{M}1(z))^{-1}$ is governed by the Schur complement (of $(\tilde{\Omega}{\bm k}^1-\bm k)\setminus{\bm 0}$):
	\begin{align*}
		\mathcal{S}_1(z)&=(\mathcal{M}_1(z))_{\{\bm 0\}}-\left(\mathcal{R}_{\{\bm 0\}}\mathcal{M}_1(z)\mathcal{R}_{(\tilde{\Omega}_{\bm k}^1-\bm k)\setminus\{\bm 0\}}\left((\mathcal{M}_1(z))_{(\tilde{\Omega}_{\bm k}^1-\bm k)\setminus\{\bm 0\}}\right)^{-1}\right.\\
		&\ \ \left.\vphantom{\left((\mathcal{M}_1(z))_{(\tilde{\Omega}_{\bm k}^1-\bm k)\setminus\{\bm 0\}}\right)^{-1}}\times\mathcal{R}_{(\tilde{\Omega}_{\bm k}^1-\bm k)\setminus\{\bm 0\}}\mathcal{M}_1(z)\mathcal{R}_{\{\bm 0\}}\right)\\
		&=v(z)-E+r(z)=g(z)((z-\theta_0)+r_1(z)),
	\end{align*}
	where $g(z)$ and $r_1(z)$ are analytic functions in the domain specified by \eqref{z0-}. These functions satisfy the bounds:  $|g(z)|>\kappa_1\|z+\theta_0\|_{\T}>\frac{\kappa_1}{2}\delta_0^{\frac{1}{100}}$ and $|r_1(z)|\le \ep^2 D^2(\alpha)\delta_0^{-\frac{1}{50}}<\ep$. Since
	\begin{align*}
		|r_1(z)|<|z-\theta_0|\ \text{for}\ |z-\theta_0|=\delta_0^{\frac{1}{10}},
	\end{align*}
	using the Rouch\'e's  theorem implies
	\begin{align*}
		(z-\theta_0)+r_1(z)=0
	\end{align*}
	has a unique root $\theta_1$ in  the set of \eqref{z0-} satisfying
	\begin{align}\label{t0-t1}
		|\theta_0-\theta_1|=|r_1(\theta_1)|<\ep.
	\end{align}
	Since $|r_1(z)|<\ep$ and \eqref{t0-t1}, we get for $|z|=\delta_0^{\frac{1}{10}}$,
	\begin{align*}
		\frac{|r_1(\theta_1)-r_1(z)|}{|z-\theta_0+r_1(z)|}\le 4\ep \delta_0^{-\frac{1}{10}},
	\end{align*}
	which combined with $\theta_1-\theta_0+r_1(\theta_1)=0$ shows
	\begin{align*}
		\frac{|z-\theta_1|}{|z-\theta_0+r_1(z)|}=\frac{|z-\theta_0+r_1(\theta_1)|}{|z-\theta_0+r_1(z)|}\in[1-4\ep \delta_0^{-\frac{1}{10}},1+4\ep \delta_0^{-\frac{1}{10}}].
	\end{align*}
	By the maximum modulus principle,  we have
	\begin{align*}
	\frac{1}{2}\le\frac{|z-\theta_1|}{|z-\theta_0+r_1(z)|}\le 2.
	\end{align*}
	Moreover, $\theta_1$ is the unique root of $\det \mathcal{M}_1(z)=0$ in the set of  \eqref{z0-}. Since $\|z+\theta_0\|_{\T}>\frac{1}{2}\delta_0^{\frac{1}{100}}$ and $|\theta_0-\theta_1|<\ep$, we get
	\begin{align*}
	\frac{1}{2}\|z+\theta_0\|_{\T}\le \|z+\theta_1\|_{\T}\le 2\|z+\theta_0\|_{\T}.
	\end{align*}
	Then by \eqref{m1z01} and \eqref{sc}, we have  for $z$ being in the set of \eqref{z0-},
	\begin{align}
		|\mathcal{S}_1(z)|&\ge\frac{\kappa_1}{4}\|z+\theta_1\|_{\T}\cdot\|z-\theta_1\|_{\T}\ge\delta_0\|z+\theta_1\|_{\T}\cdot\|z-\theta_1\|_{\T}, \label{S1}\\
		\nonumber\|(\mathcal{M}_1(z))^{-1}\|&<4\left(1+\left\|\left(\mathcal{M}_1(z)_{(\tilde{\Omega}_{\bm k}^1-\bm k)\setminus\{\bm 0\}}\right)^{-1}\right\|\right)^2(1+|\mathcal{S}_1(z)|^{-1})\\
		&< \delta_0^{-2}\|z+\theta_1\|_{\T}^{-1}\cdot\|z-\theta_1\|_{\T}^{-1}.\label{M-1}
	\end{align}
	Now, for $\bm k\in Q_0^+$, we consider $\mathcal{M}_1(z)$ in the set
	\begin{align}\label{z0+}
		\left\{z\in\C:\ |z+\theta_0|\le\delta_0^{\frac{1}{10}}\right\}.
	\end{align}
	
	Applying an analogous argument demonstrates that $\det \mathcal{M}_1(z)=0$ has a unique root $\theta_1'$ in the region defined by \eqref{z0+}. We claim that $\theta_1 + \theta_1' = 0$. Indeed, Lemma \ref{ef} establishes that $\det \mathcal{M}_1(z)$ is an even function of $z$, which, combined with the uniqueness of the root, implies $\theta_1' = -\theta_1$. Consequently, for $z$ in \eqref{z0+}, both \eqref{S1} and \eqref{M-1} remain valid. Finally, we note that \eqref{S1} and \eqref{M-1} continue to hold for
	\begin{align}\label{z0-+}
		\left\{z\in\C:\ \min\limits_{\sigma=\pm1}|z+\sigma\theta_0|\le \delta_0^{\frac{1}{10}}\right\}.
	\end{align}
	From \eqref{P11}, we conclude that $\theta+\bm k\cdot\bm\omega$ belongs to the region defined in \eqref{z0-+}. Therefore, for $\bm k\in P_1$, we obtain
	\begin{align}
		\nonumber\|\mathcal{T}_{\tilde{\Omega}_{\bm k}^1}^{-1}\|&=\|(\mathcal{M}_1(\theta+\bm k\cdot\bm\omega))^{-1}\|\\
		&<\delta_0^{-2}\|\theta+\bm k\cdot\bm\omega+\theta_1\|_{\T}^{-1}\cdot\|\theta+\bm k\cdot\bm\omega-\theta_1\|_{\T}^{-1}\label{T-11}.
	\end{align}
	
	\begin{itemize}
		\item[\textbf{Case 2}]:  {The case $(\bm C2)_0$ occurs},  i.e.,
	\end{itemize}
	\begin{align}\label{-+2}
		\text{dist}\left(\tilde{Q}_0^-,Q_0^+\right)\le100N_1^{10}.
	\end{align}
	Then there exist $\bm i_0\in Q_0^+$ and $\bm j_0\in\tilde{Q}_0^-$ with $\|\bm i_0-\bm j_0\|\le100N_1^{10}$ such that
	\begin{align*}
		\|\theta+\bm i_0\cdot\bm \omega+\theta_0\|_{\T}<\delta_0,\ \|\theta+\bm j_0\cdot\bm \omega-\theta_0\|_{\T}<\delta_0^{\frac{1}{100}}.
	\end{align*}
	Set $\bm l_0=\bm i_0-\bm j_0$. Then
	\begin{align*}
		\|\bm l_0\|=\text{dist}\left(\tilde{Q}_0^-,Q_0^+\right)=\text{dist}\left(\tilde{Q}_0^+,Q_0^-\right).
	\end{align*}
	Define
	\begin{align*}
		O_1=Q_0^-\cup(Q_0^+-\bm l_0).
	\end{align*}
	For $\bm k\in Q_0^+$, we have
	\begin{align*}
		\|\theta+(\bm k-\bm l_0)\cdot\bm \omega-\theta_0\|_{\T}&<\|\theta+\bm k\cdot\bm\omega+\theta_0\|_{\T}+\|\bm l_0\cdot\bm\omega+2\theta_0\|_{\T}\\
		&<\delta_0+\delta_0+\delta_0^{\frac{1}{100}}<2\delta_0^{\frac{1}{100}}.
	\end{align*}
	Thus
	\begin{align*}
		O_1\subset\left\{\bm o\in\Z^d:\ \|\theta+\bm o\cdot\bm \omega-\theta_0\|_{\T}<2\delta_0^{\frac{1}{100}}\right\}.
	\end{align*}
	For every $\bm o\in O_1$, define its mirror point
	\begin{align*}
		\bm o^*=\bm o+\bm l_0.
	\end{align*}
	Next, define
	\begin{align}\label{P12}
		P_1=\left\{\frac{1}{2}(\bm o+\bm o^*):\ \bm o\in O_1\right\}=\left\{\bm o+\frac{\bm l_0}{2}:\ \bm o\in O_1\right\}.
	\end{align}
	Associate each  $\bm k\in P_1$ with  $\Omega_{\bm k}^1:=\lg_{100N_1^{10}}(\bm k)$, $2\Omega_{\bm k}^1:=\lg_{200N_1^{10}}(\bm k)$ and   $\tilde{\Omega}_{\bm k}^1:=\lg_{N_1^{100}}(\bm k)$. Thus
	\begin{align*}
		Q_0\subset\bigcup_{\bm k\in P_1}\Omega_{\bm k}^1
	\end{align*}
	and $\tilde{\Omega}_{\bm k}^1-\bm k\subset\Z^d+\frac{\bm l_0}{2}$ is independent of $\bm k\in P_1$ and symmetrical about origin. Notice that
	\begin{align*}
		\min&\left(\left\|\frac{\bm l_0}{2}\cdot\bm\omega+\theta_0\right\|_{\T},\left\|\frac{\bm l_0}{2}\cdot\bm \omega+\theta_0-\frac{1}{2}\right\|_{\T}\right)\\
		&=\frac{1}{2}\|\bm l_0\cdot\bm\omega+2\theta_0\|_{\T}\\
		&\le\frac{1}{2}(\|\theta+\bm i_0\cdot\bm\omega+\theta_0\|_{\T}+\|\theta+\bm j_0\cdot\bm \omega-\theta_0\|_{\T})<\delta_0^{\frac{1}{100}}.
	\end{align*}
	Since $\delta_0\ll1$, only one of
	\begin{align*}
		\left\|\frac{\bm l_0}{2}\cdot\bm\omega+\theta_0\right\|_{\T}<\delta_0^{\frac{1}{100}}\ \text{and}\ \left\|\frac{\bm l_0}{2}\cdot\bm \omega+\theta_0-\frac{1}{2}\right\|_{\T}<\delta_0^{\frac{1}{100}}
	\end{align*}
	holds true. First, we consider the case of
	\begin{align}\label{l01}
		\left\|\frac{\bm l_0}{2}\cdot\bm\omega+\theta_0\right\|_{\T}<\delta_0^{\frac{1}{100}}.
	\end{align}
	Let $\bm k\in P_1$. Since $\bm k=\frac{1}{2}(\bm o+\bm o^*)=\bm o+\frac{\bm l_0}{2}$ (for some $\bm o\in O_1$), we have
	\begin{align}\label{t+ko1}
		\|\theta+\bm k\cdot\bm\omega\|_{\T}\le\|\theta+\bm o\cdot\bm\omega-\theta_0\|_{\T}+\left\|\frac{\bm l_0}{2}\cdot\bm\omega+\theta_0\right\|_{\T}<3\delta_0^{\frac{1}{100}}.
	\end{align}
	Thus if $\bm k\ne \bm k'\in P_1$, we obtain
	\begin{align*}
		\|\bm k-\bm k'\|\ge\left(\frac{\g}{6\delta_0^{\frac{1}{100}}}\right)^{\frac{1}{\tau}}\gg 100N_1^{100},
	\end{align*}
	which implies
	\begin{align*}
		\text{dist}\left(\tilde{\Omega}_{\bm k}^1,\tilde{\Omega}_{\bm k'}^1\right)>10\tz_1\ \text{for}\ \bm k\ne\bm k'\in P_1.
	\end{align*}
	Consider
	\begin{align*}
		\mathcal{M}_1(z):=(\mathcal{T}(z))_{\tilde{\Omega}_{\bm k}^1-\bm k}=\left((v(z+\bm n\cdot\bm\omega)-E)\delta_{\bm n,\bm n'}+\ep \mathcal{W}\right)_{\tilde{\Omega}_{\bm k}^1-\bm k}
	\end{align*}
	in the set of
	\begin{align}\label{z02}
		\left\{z\in\C:\ |z|\le\delta_0^{\frac{1}{10^3}}\right\}.
	\end{align}
	For $\bm n\ne\pm\frac{\bm l_0}{2}$ and $\bm n\in\tilde{\Omega}_{\bm k}^1-\bm k$, we have
	\begin{align*}
		\|\bm n\cdot\bm\omega\pm\theta_0\|_{\T}&\ge\left\|\left(\bm n\mp\frac{\bm l_0}{2}\right)\cdot\bm \omega\right\|_{\T}-\left\|\frac{\bm l_0}{2}\cdot\bm\omega+\theta_0\right\|_{\T}\\
		&>\frac{\g}{(2N_1^{100})^\tau}-\delta_0^{\frac{1}{10}}\gg\delta_0^{\frac{1}{10^4}}.
	\end{align*}
	Thus for $z$ being in the set of \eqref{z02} and $\bm n\ne\pm\frac{\bm l_0}{2}$, we have
	\begin{align*}
		\|z+\bm n\cdot\bm\omega\pm\theta_0\|_{\T}\ge\|\bm n\cdot\bm\omega\pm\theta_0\|_{\T}-|z|\gg\delta_0^{\frac{1}{10^4}}.
	\end{align*}
	Hence for $\bm n\in(\tilde{\Omega}_{\bm k}^1-\bm k)\setminus\left\{\pm\frac{\bm l_0}{2}\right\},$ we have
	\begin{align*}
		|v(z+\bm n\cdot\bm\omega)-E|> \kappa_1\cdot\delta_0^{2\times\frac{1}{10^4}}\gg\ep.
	\end{align*}
	Using Neumann series argument concludes
	\begin{align}\label{M-12}
		\left\|\left((\mathcal{M}_1(z))_{(\tilde{\Omega}_{\bm k}^1-\bm k)\setminus\left\{\pm\frac{\bm l_0}{2}\right\}}\right)^{-1}\right\|<\delta_0^{-3\times\frac{1}{10^4}}.
	\end{align}
	Thus by Lemma \ref{scl}, $(\mathcal{M}_1(z))^{-1}$ is controlled by   the Schur complement of $(\tilde{\Omega}_{\bm k}^1-\bm k)\setminus\left\{\pm\frac{\bm l_0}{2}\right\}$, i.e.,
	\begin{align*}
		\mathcal{S}_1(z)&=(\mathcal{M}_1(z))_{\left\{\pm\frac{\bm l_0}{2}\right\}}-\left(\vphantom{\left((\mathcal{M}_1(z))_{(\tilde{\Omega}_{\bm k}^1-\bm k)\setminus\left\{\pm\frac{\bm l_0}{2}\right\}}\right)^{-1}}\mathcal{R}_{\left\{\pm\frac{\bm l_0}{2}\right\}}\mathcal{M}_1(z)\mathcal{R}_{(\tilde{\Omega}_{\bm k}^1-\bm k)\setminus\left\{\pm\frac{\bm l_0}{2}\right\}}\right.\\
		&\ \ \left.\times\left((\mathcal{M}_1(z))_{(\tilde{\Omega}_{\bm k}^1-\bm k)\setminus\left\{\pm\frac{\bm l_0}{2}\right\}}\right)^{-1}\mathcal{R}_{(\tilde{\Omega}_{\bm k}^1-\bm k)\setminus\left\{\pm\frac{\bm l_0}{2}\right\}}\mathcal{M}_1(z)\mathcal{R}_{\left\{\pm\frac{\bm l_0}{2}\right\}}\right).
	\end{align*}
	Then by \eqref{wphi}, \eqref{M-12} and \eqref{detd},  we  get
	\begin{align}
		\nonumber\max_{\bm x\in \left\{\pm\frac{\bm l_0}{2}\right\}}\sum_{\bm y\in \left\{\pm\frac{\bm l_0}{2}\right\}}|\mathcal{S}_1(z)(\bm x,\bm y)|&\le\max_{\bm x\in\left\{\pm\frac{\bm l_0}{2}\right\}}\sum_{\bm y\in \left\{\pm\frac{\bm l_0}{2}\right\}}\left|(\mathcal{M}_1(z))_{\left\{\pm\frac{\bm l_0}{2}\right\}}(\bm x,\bm y)\right|\\
		\nonumber&\ \ \ \ +2(D(\alpha))^2\cdot\ep^2\delta_0^{-\frac{3}{10^4}}\\
		\label{s10}&\le 2|v|_R+\delta_0<4|v|_R,
	\end{align}
	and
	\begin{align*}
		\det \mathcal{S}_1(z)&=\det\left((\mathcal{M}_1(z))_{\left\{\pm\frac{\bm l_0}{2}\right\}}\right)+O(\ep^2\delta_0^{-\frac{3}{10^4}})\\
		&=\left(v\left(z+\frac{\bm l_0}{2}\cdot\bm\omega\right)-E\right)\left(v\left(z-\frac{\bm l_0}{2}\cdot\bm\omega\right)-E\right)+O(\ep^2\delta_0^{-\frac{3}{10^4}}).
	\end{align*}
	When $\bm l_0=\bm0$, the argument simplifies considerably, so we omit this case. For $\bm l_0\ne\bm 0$, combining \eqref{l01} with \eqref{z02} yields
	\begin{align*}
		\left\|z+\frac{\bm l_0}{2}\cdot\bm\omega-\theta_0\right\|_{\T}&\ge\|\bm l_0\cdot\bm\omega\|_{\T}-\left\|\frac{\bm l_0}{2}\cdot\bm\omega+\theta_0\right\|_{\T}-|z|\\
		&>\frac{\g}{(100N_1^{10})^\tau}-\delta_0^{\frac{1}{100}}-\delta_0^{\frac{1}{10^3}}\\
		&>\delta_0^{\frac{1}{10^4}},
	\end{align*}
	and
	\begin{align*}
		\left\|z-\frac{\bm l_0}{2}\cdot\bm\omega+\theta_0\right\|_{\T}&\ge\|\bm l_0\cdot\bm\omega\|_{\T}-\left\|\frac{\bm l_0}{2}\cdot\bm\omega+\theta_0\right\|_{\T}-|z|\\
		&>\frac{\g}{(100N_1^{10})^\tau}-\delta_0^{\frac{1}{100}}-\delta_0^{\frac{1}{10^3}}\\
		&>\delta_0^{\frac{1}{10^4}}.
	\end{align*}
	Let $z_1$ satisfy
	\begin{align}\label{z1}
		z_1\equiv\frac{\bm l_0}{2}\cdot\bm\omega+\theta_0\ (\text{mod}\ \Z),\ |z_1|=\left\|\frac{\bm l_0}{2}\cdot\bm\omega+\theta_0\right\|_{\T}<\delta_0^{\frac{1}{100}}.
	\end{align}
	Then
	\begin{align*}
		|\det \mathcal{S}_1(z)|&\ge \left\|z+\frac{\bm l_0}{2}\cdot\bm\omega-\theta_0\right\|_{\T}\cdot\left\|z-\frac{\bm l_0}{2}\cdot\bm\omega+\theta_0\right\|_{\T}\\
		&\ \ \cdot|(z-z_1)(z+z_1)+r_1(z)|\\
		&\ge\delta_0^{\frac{2}{10^4}}|(z-z_1)(z+z_1)+r_1(z)|,
	\end{align*}
	where $r_1(z)$ is an analytic function in the set of  \eqref{z02} with
	\begin{align}\label{r1}
		|r_1(z)|\ll\ep\ll\delta_0^{\frac{1}{10^3}}.
	\end{align}
	By Rouch\'e's  theorem, the equation
	\begin{align*}
		(z-z_1)(z+z_1)+r_1(z)=0
	\end{align*}
	has exactly two roots $\theta_1$ and $\theta_1'$ in the region defined by \eqref{z02}, which are small perturbations of $\pm z_1$. Suppose, for contradiction, that both
	\begin{align*}
		|z_1-\theta_1|>|r_1(\theta_1)|^{\frac{1}{2}}\ \text{and}\ |z_1+\theta_1|>|r_1(\theta_1)|^{\frac{1}{2}}.
	\end{align*}
	This would imply
	\begin{align*}
		|r_1(\theta_1)|=|z_1-\theta_1|\cdot|z_1+\theta_1|>|r_1(\theta_1)|,
	\end{align*}
	which is impossible. Therefore, without loss of generality, we may assume
	\begin{align*}
		|z_1-\theta_1|\le|r_1(\theta_1)|^{\frac{1}{2}}\le\ep^{\frac{1}{2}}.
	\end{align*}
	We observe that the zero sets coincide:
	\begin{align*}
		\left\{|z|\le \delta_0^{\frac{1}{10^3}}:\ \det \mathcal{M}_1(z)=0\right\}=\left\{|z|\le \delta_0^{\frac{1}{10^3}}:\ \det \mathcal{S}_1(z)=0\right\}
	\end{align*}
	and that $\det \mathcal{M}_1(z)$ is an even function (see Lemma \ref{ef}). Consequently, we must have
	\begin{align*}
		\theta_1'=-\theta_1.
	\end{align*}
	Furthermore, from \eqref{z1} and \eqref{r1}, we obtain for $|z| = \delta_0^{\frac{1}{10^3}}$:
	\begin{align*}
		\frac{|r_1(z)-r_1(\theta_1)|}{|(z-z_1)(z+z_1)+r_1(\theta_1)|}\le2\ep^2\delta_0^{-\frac{2}{10^3}}.
	\end{align*}
	Combined with the identity $\theta_1^2 - z_1^2 + r_1(\theta_1) = 0$, this yields:
	\begin{align*}
		\frac{|(z-z_1)(z+z_1)+r_1(z)|}{|(z-\theta_1)(z+\theta_1)|}&=\frac{|(z-z_1)(z+z_1)+r_1(z)|}{|(z-z_1)(z+z_1)+r_1(\theta_1)|}\\
		&\in\left[1-2\ep^2\delta_0^{-\frac{2}{10^3}},1+2\ep^2\delta_0^{-\frac{2}{10^3}}\right].
	\end{align*}
	The maximum modulus principle then gives:
	\begin{align*}
		\frac{1}{2}\le \frac{|(z-z_1)(z+z_1)+r_1(z)|}{|(z-\theta_1)(z+\theta_1)|}\le 2.
	\end{align*}
	Consequently, for $z$ in the region \eqref{z02}, we establish the determinant bound:
	\begin{align}\label{detS11}
		|\det \mathcal{S}_1(z)|\ge\frac{1}{2}\delta_0^{\frac{2}{10^4}}\|z-\theta_1\|_{\T}\cdot\|z+\theta_1\|_{\T}.
	\end{align}
	Applying Cramer's rule with Lemma \ref{chi}, \eqref{s10}, and \eqref{detS11} yields the inverse estimate:
	\begin{align}\label{S1-11}
		\|(\mathcal{S}_1(z))^{-1}\|=\frac{\|(\mathcal{S}_1(z))^*\|}{|\det \mathcal{S}_1(z)|}\le \delta_0^{-1}\|z-\theta_1\|_{\T}^{-1}\cdot\|z+\theta_1\|_{\T}^{-1}.
	\end{align}
	Combining \eqref{M-12}, \eqref{S1-11}, and Lemma \ref{scl}, we derive the fundamental bound:
	\begin{align}\label{M1z-11}
		\nonumber\|(\mathcal{M}_1(z))^{-1}\|&<4\left(1+\left\|\left((\mathcal{M}_1(z))_{(\tilde{\Omega}_{\bm k}^1-\bm k)\setminus\left\{\pm\frac{\bm l_0}{2}\right\}}\right)^{-1}\right\|\right)^2(1+\|(\mathcal{S}_1(z))^{-1}\|)\\
		&<\delta_0^{-2}\|z-\theta_1\|_{\T}^{-1}\cdot\|z+\theta_1\|_{\T}^{-1}.
	\end{align}
	Thus for \eqref{l01}, these estimates \eqref{S1-11} and \eqref{M1z-11} hold for all $z$ in the periodic domain:
	\begin{align*}
		\{z\in\C:\ \|z\|_{\T}\le \delta_0^{\frac{1}{10^3}}\}
	\end{align*}
	due to the 1-periodicity of $\mathcal{M}_1(z)$. Finally, for $\bm{k} \in P_1$ via \eqref{t+ko1}, we obtain:
	\begin{align*}
		\|\mathcal{T}_{\tilde{\Omega}_{\bm k}^1}^{-1}\|&=\|(\mathcal{M}_1(\theta+\bm k\cdot\bm\omega))^{-1}\|\\
		&<\delta_0^{-2}\|\theta+\bm k\cdot\bm\omega-\theta_1\|_{\T}^{-1}\cdot\|\theta+\bm k\cdot\bm\omega+\theta_1\|_{\T}^{-1}.
	\end{align*}
	For the case of
	\begin{align}\label{l02}
		\left\|\frac{\bm l_0}{2}\cdot\bm\omega+\theta_0-\frac{1}{2}\right\|_{\T}<\delta_0^{\frac{1}{100}},
	\end{align}
	we have for $\bm k\in P_1$,
	\begin{align}\label{t+ko2}
		\left\|\theta+\bm k\cdot\bm\omega-\frac{1}{2}\right\|_{\T}<3\delta_0^{\frac{1}{100}}.
	\end{align}
	Consider
	\begin{align*}
		\mathcal{M}_1(z):=\mathcal{T}_{\tilde{\Omega}_{\bm k}^1-\bm k}(z)=\left((v(z+\bm n\cdot\bm\omega)-E)\delta_{\bm n,\bm n'}+\ep \mathcal{W}\right)_{\bm n\in\tilde{\Omega}_{\bm k}^1-\bm k}
	\end{align*}
	in
	\begin{align}\label{z03}
		\left\{z\in\C:\ \left|z-\frac{1}{2}\right|\le \delta_0^{\frac{1}{10^3}}\right\}.
	\end{align}
	
	An analogous argument demonstrates that the equation $\det \mathcal{M}_1(z) = 0$ has two roots $\theta_1$ and $\theta'_1$ within the region defined by \eqref{z03}, such that estimates \eqref{s10}--\eqref{M1z-11} remain valid for all $z$ in \eqref{z03}. Consequently, under condition \eqref{l02}, these estimates \eqref{s10}--\eqref{M1z-11} extend to the shifted domain:
	\begin{align*}
	\left\{z\in\C:\ \left\|z-\frac{1}{2}\right\|_{\T}\le \delta_0^{\frac{1}{10^3}}\right\}.
	\end{align*}
	Furthermore, for any $\bm{k} \in P_1$ through \eqref{t+ko2}, we derive the $\ell^2$-norm bound:
	\begin{align}
		\nonumber\|\mathcal{T}_{\tilde{\Omega}_{\bm k}^1}^{-1}\|&=\|(\mathcal{M}_1(\theta+\bm k\cdot\bm\omega))^{-1}\|\\
		\label{to1-1}&<\delta_0^{-2}\|\theta+\bm k\cdot\bm\omega-\theta_1\|_{\T}^{-1}\cdot\|\theta+\bm k\cdot\bm\omega+\theta_1\|_{\T}^{-1}.
	\end{align}
	
	We have thus established the desired norm estimates for $\|\mathcal{T}_{\tilde{\Omega}_{\bm k}^1}^{-1}\|$ in both cases $(\bm C1)_0$ and $(\bm C2)_0$. For each $\bm{k} \in P_1$, we define the subset $A_{\bm k}^1\subset \Omega_{\bm k}^1$ as follows:
	\begin{align*}
		A_{\bm k}^1:=\left\{\begin{array}{lc}
			\{\bm k\}, & \text{Case}\ (\bm C1)_0\\
			\{\bm o\}\cup\{\bm o^*\}, & \text{Case}\ (\bm C2)_0
		\end{array}\right.,
	\end{align*}
	
	where in Case $(\bm{C}2)_0$, the vector $\bm{k}$ admits the decomposition $\bm{k} = \frac{1}{2}(\bm{o} + \bm{o}^*)$ for some $\bm{o} \in O_1$.
	\\
	
	\  \\
	
	\begin{itemize}
		\item[\textbf{Step 2}]: \textbf{Off-diagonal estimates of $\boT_{\tilde{\Omega}_{\bm k}^1}^{-1}$}.
	\end{itemize}
	
	The main result of this step is Theorem \ref{psq1}, which establishes that despite the presence of singular 0-scale sites within the block $\tilde{\Omega}_{\bm k}^1$, the Green's function $\mathcal{T}_{\tilde{\Omega}_{\bm k}^1}^{-1}$ remains controllable provided no 1-scale singular sites exist. Recalling
	\begin{align*}
		\delta_{1}=\delta_0^{10^{5\rho'}}
	\end{align*}
	and
	\begin{align*}
		Q_1^{\pm}=\{\bm k\in P_1:\ \|\theta+\bm k\cdot\bm\omega\pm\theta_1\|_{\T}<\delta_1\},\ Q_1=Q_1^+\cup Q_1^-,
	\end{align*}
	we have
	\begin{thm}\label{psq1}
		For $\bm k\in P_1\setminus Q_1$, we have
		\begin{align}\label{tg1a}
			|\mathcal{T}_{\tilde{\Omega}_{\bm k}^1}^{-1}(\bm x,\bm y)|<e^{-\alpha'_0\log^{\rho}(1+\|\bm x-\bm y\|)}\ \text{for $\|\bm x-\bm y\|\ge\frac{\tilde{\zeta_1}}{10}$},
		\end{align}
		where $\alpha'_0=\frac{3}{4}\alpha\left(1-\frac{20\times 10^{5\rho'}}{\alpha\log^{\rho-\rho'}N_1}\right)$ is defined in \eqref{alpha's}.
	\end{thm}
	\begin{proof}
From our construction, we have
		\begin{align*}
			Q_0\subset \bigcup_{\bm k\in P_1}A_{\bm k}^1\subset\bigcup_{\bm k\in P_1}2\Omega_{\bm k}^1.
		\end{align*}
		Thus
		\begin{align*}
			(\tilde{\Omega}_{\bm k}^1\setminus 2\Omega_{\bm k}^1)\cap Q_0=\emptyset,
		\end{align*}
		which shows that $\tilde{\Omega}_{\bm k}^1\setminus 2\Omega_{\bm k}^1$ is $0$-good. Since \eqref{to1-1} and $\bm k\notin Q_1$, we have
		\begin{align}\label{tk0}
			\|\mathcal{T}_{\tilde{\Omega}_{\bm k}^1}^{-1}\|<\delta_0^{-2}\|\theta+\bm k\cdot\bm\omega-\theta_1\|_{\T}^{-1}\cdot\|\theta+\bm k\cdot\bm\omega+\theta_1\|_{\T}^{-1}<\delta_0^{-2}\delta_1^{-2}<\delta_1^{-3}.
		\end{align}
		
		To obtain the desired estimates, we divide the remaining proof into several cases.		
	\begin{itemize}
		\item[\textbf{Case 1}:] $\bm x\in 2\Omega_{\bm k}^{1}$ or $\bm y\in 2\Omega_{\bm k}^1$. Without loss of generality, we assume that $\bm y\in2\Omega_{\bm k}^{1}$. From $\|\bm x-\bm y\|\ge\frac{\tilde{\zeta_1}}{10}$ and $\diam(2\Omega_{\bm k}^1)\ll\tilde{\zeta_1}$, we have $\bm x\in \tilde{\Omega}_{\bm k}^{1}\setminus 2\Omega_{\bm k}^1$. Since $\tilde{\Omega}_{\bm k}^{1}\setminus 2\Omega_{\bm k}^1$ is $0$-good and \eqref{ite0geq1}, there is some $\bm x'\in 2\Omega_{\bm k}^{1}$ such that
			\begin{align}\label{51201}
			|\boT_{\tilde{\Omega}_{\bm k}^{1}}^{-1}(\bm x,\bm y)|\le (\#(\tilde{\Omega}_{\bm k}^{1}\setminus 2\Omega_{\bm k}^1))\cdot e^{-\frac{3}{4}\alpha\log^{\rho}(1+\|\bm x-\bm x'\|)}\cdot|\boT_{\tilde{\Omega}_{\bm k}^{1}}^{-1}(\bm x',\bm y)|\cdot\delta_0^{-3}.
		\end{align}
		Next, we will extract $\log^{\rho}(1+\|\bm x-\bm y\|)$ from $\log^{\rho}(1+\|\bm x-\bm x'\|)$.  Since $\bm x',\bm y\in 2\Omega_{\bm k}^{1}$, $\|\bm y-\bm x'\|\le 4\zeta_1$. From $\|\bm x-\bm y\|\ge \frac{\tilde{\zeta}_1}{10}\gg4\zeta_1\gg1$ and \eqref{exl}, we have
		\begin{align}
			\nonumber&\ \ \log^{\rho}(1+\|\bm x-\bm x'\|)\ge\log^{\rho}(1+\|\bm x-\bm y\|-4\zeta_1)\\
		\nonumber	\ge&\ \ \left(1-2\rho\frac{4\zeta_1}{(1+\|\bm x-\bm y\|)\log(1+\|\bm x-\bm y\|)}\right)\log^{\rho}(1+\|\bm x-\bm y\|)\\
			\label{803}	\ge&\ \ \left(1-\frac{80\rho\zeta_1}{\tz_1\log N_1}\right)\log^{\rho}(\|\bm x-\bm y\|+1).
		\end{align}
		Therefore, from \eqref{tk0}--\eqref{803}, $1<\rho'<\rho<\rho'+1$ and $\|\bm x-\bm y\|\ge \frac{\tilde{\zeta}_1}{10}$, we get
		\begin{align*}
			|\boT_{\tilde{\Omega}_{\bm k}^{1}}^{-1}(\bm x,\bm y)|&\le (2\tilde{\zeta}_1+1)^d\cdot e^{-\frac{3}{4}\alpha\left(1-\frac{80\rho\zeta_1}{\tz_1\log N_1}\right)\log^{\rho}(\|\bm x-\bm y\|+1)}\cdot\delta_0^{-5}\cdot\delta_1^{-2}\\
			&\le e^{-\frac{3}{4}\alpha\left(1-\frac{6\times 10^{5\rho'}}{\alpha\log^{\rho-\rho'} N_1}\right)\log^{\rho}(\|\bm x-\bm y\|+1)}.
		\end{align*}
		\item[\textbf{Case 2}:] $\bm x\in \tilde{\Omega}_{\bm k}^{1}\setminus 2\Omega_{\bm k}^{1}$ and $\bm y\in \tilde{\Omega}_{\bm k}^{1}\setminus 2\Omega_{\bm k}^1$. If $\bm z\in \tilde{\Omega}_{\bm k}^{1}\setminus\left(2\Omega_{\bm k}^{1}\bigcup \lg_{\frac{1}{2}N_1}(\bm y)\right)$, there is a $0$-good set $\tilde{\Omega}_{\bm k}^{1}\cap\lg_{\frac{1}{2}N_1}(\bm z)$ such that 
		\begin{align*}
			\bm z\in \tilde{\Omega}_{\bm k}^{1}\cap\lg_{\frac{1}{2}N_1}(\bm z)\text{ and }\bm y\in\tilde{\Omega}_{\bm k}^{1}\setminus\lg_{\frac{1}{2}N_1}(\bm z).
		\end{align*}
		For $\bm z\in\tilde{\Omega}_{\bm k}^{1}$, we can define
				\begin{align*}
			\hat{\bm z}=\left\{\begin{array}{ll}
				\bm z, & \bm z\in 2\Omega_{\bm k}^{1}\bigcup \lg_{\frac{1}{2}N_1}(\bm y),\\
				\bm z', & \bm z\in \tilde{\Omega}_{\bm k}^{1}\setminus\left(2\Omega_{\bm k}^{1}\bigcup \lg_{\frac{1}{2}N_1}(\bm y)\right),
			\end{array}\right.
		\end{align*}
		by Corollary \ref{ite0g} with $\bm v=\bm y$ and $\lg'=\tilde{\Omega}_{\bm k}^{1}\cap\lg_{\frac{1}{2}N_1}(\bm z)$.

	Let $\bm x_0:=\bm x$ and $\bm x_{l+1}=\hat{\bm x}_{l}$, $l\ge0$. For given $\{\bm x_{l}\}_{l\in\N}$, we define $l_1\ge1$ to be the smallest integer so that $\bm x_{l_1}\in 2\Omega_{\bm k}^{1}\bigcup \lg_{\frac{1}{2}N_1}(\bm y)$. We then have 
	\begin{align*}
		\bm x_{i}\in \tilde{\Omega}_{\bm k}^{1}\setminus\left(2\Omega_{\bm k}^{1}\bigcup \lg_{\frac{1}{2}N_1}(\bm y)\right)\text{ for }0\le i<l_1.
	\end{align*}
	 We also divide the discussion into 3 cases:
	\item[\textbf{Case 2-1}:] $l_1>\left[\frac{\alpha \log^{\rho}(\|\bm x-\bm y\|+1)+4\times 10^{5\rho'}\log^{\rho'}(N_1+1)}{\alpha\left(1-\frac{6\times 10^{5\rho'}}{\alpha\log^{\rho-\rho'}N_1}\right)\log^{\rho}\left(\frac{N_1}{2}+1\right)}\right]+1:=l^*$. Since $\bm x_{i}\in \tilde{\Omega}_{\bm k}^{1}\setminus\left(2\Omega_{\bm k}^{1}\bigcup \lg_{\frac{1}{2}N_1}(\bm y)\right)$ for $0\le i<l_1$, we have 
	\begin{align*}
		\bm x_{i+1}=\bm x'_{i}\in \tilde{\Omega}_{\bm k}^{1}\setminus\lg_{\frac{1}{2}N_1}(\bm x_i)\ \text{(cf. Corollary \ref{ite0g})}
	\end{align*}
	and then
	\begin{align*}
		\|\bm x_{i+1}-\bm x_{i}\|\ge\frac{N_1}{2}\ \text{for $0\le i<l_1$.}
	\end{align*}
	 Thus, from \eqref{indpa}, \eqref{ite0g0} and \eqref{tk0}, we get
	\begin{align*}
		|\boT_{\tilde{\Omega}_{\bm k}^{1}}^{-1}(\bm x,\bm y)|&\le \prod_{i=0}^{l^*-1}\left(e^{-\frac{3}{4}\alpha\left(1-\frac{6\times 10^{5\rho'}}{\alpha\log^{\rho-\rho'} N_1}\right)\log^{\rho}(\|\bm x_{i+1}-\bm x_{i}\|+1)}\right)|\boT_{\tilde{\Omega}_{\bm k}^{1}}^{-1}(\bm x_{l^*},\bm y)|\\
		&\le e^{-\frac{3}{4}\alpha\cdot l^*\left(1-\frac{6\times 10^{5\rho'}}{\alpha\log^{\rho-\rho'} N_1}\right)\log^{\rho}\left(\frac{N_1}{2}+1\right)}\delta_1^{-3}\\
		&\le e^{-\frac{3}{4}\alpha\cdot l^*\left(1-\frac{6\times 10^{5\rho'}}{\alpha\log^{\rho-\rho'} N_1}\right)\log^{\rho}\left(\frac{N_1}{2}+1\right)}\cdot e^{3\times 10^{5\rho'}\log^{\rho'}(N_1+1)}\\
		&\le e^{-\frac{3}{4}\alpha\log^{\rho}(\|\bm x-\bm y\|+1)}.
	\end{align*}
		\item[\textbf{Case 2-2:}] $l_1\le l^*$ and $\bm x_{l_1}\in \lg_{\frac{1}{2}N_1}(\bm y)$. According to $\|\bm x-\bm y\|\ge\frac{\tz_1}{10}$, $N_1\gg1$ and $1<\rho'<\rho$, we obtain
		\begin{align}
			\nonumber l^*&=\left[\frac{\alpha \log^{\rho}(\|\bm x-\bm y\|+1)+4\times 10^{5\rho'}\log^{\rho'}(N_1+1)}{\alpha\left(1-\frac{6\times 10^{5\rho'}}{\alpha\log^{\rho-\rho'}N_1}\right)\log^{\rho}\left(\frac{N_1}{2}+1\right)}\right]+1\\
		\label{l*}	&\le \log^{\rho}(1+\|\bm x-\bm y\|).
		\end{align}
	Then from \eqref{quaeq}, \eqref{ite0g0}, \eqref{tk0} and $\|\bm x-\bm y\|\ge\frac{\tz_1}{10}$, we have
	\begin{align}
		\nonumber	|\boT_{\tilde{\Omega}_{\bm k}^{1}}^{-1}(\bm x,\bm y)|&\le \prod_{i=0}^{l_1-1}\left(e^{-\frac{3}{4}\alpha\left(1-\frac{6\times 10^{5\rho'}}{\alpha\log^{\rho-\rho'} N_1}\right)\log^{\rho}(\|\bm x_{i+1}-\bm x_{i}\|+1)}\right)|\boT_{\tilde{\Omega}_{\bm k}^{1}}^{-1}(\bm x_{l_1},\bm y)|\\
		\nonumber	&\le e^{-\frac{3}{4}\alpha\left(1-\frac{6\times 10^{5\rho'}}{\alpha\log^{\rho-\rho'} N_1}\right)\left(\log^{\rho}(\|\bm x_{l_1}-\bm x\|+1)-C(\rho)\log^{\rho}l_1\right)}\delta_1^{-3}\\
		\nonumber	&\le e^{-\frac{3}{4}\alpha\left(1-\frac{6\times 10^{5\rho'}}{\alpha\log^{\rho-\rho'} N_1}\right)\left(\log^{\rho}(\|\bm x-\bm y\|-\|\bm x_{l_1}-\bm y\|+1)-C(\rho)\log^{\rho}l_1\right)}\\
		\label{102}&\ \ \cdot e^{3\times 10^{5\rho'}\log^{\rho'}(N_1+1)}.
	\end{align}
By $\bm x_{l_1}\in\lg_{\frac{1}{2}N_1}(\bm y)$, $\|\bm x_{l_1}-\bm y\|\le\frac{N_1}{2}$. Since $\|\bm x-\bm y\|\ge\frac{\tz_1}{10}\gg N_1\gg1$ and \eqref{exl}, we can get
\begin{align}
	\nonumber&\ \ \log^{\rho}(\|\bm x-\bm y\|-\|\bm x_{l_1}-\bm y\|+1)\\
	\nonumber\ge&\ \ \log^{\rho}\left(\|\bm x-\bm y\|-\frac{N_1}{2}+1\right)\\
	\nonumber\ge&\ \ \left(1-\rho\frac{N_1}{(1+\|\bm x-\bm y\|)\log(1+\|\bm x-\bm y\|)}\right)\log^{\rho}(1+\|\bm x-\bm y\|)\\
\label{101}	\ge&\ \ \left(1-\frac{10\rho N_1}{\tz_1\log N_1}\right)\log^{\rho}(\|\bm x-\bm y\|+1).
\end{align}
From \eqref{l*}, \eqref{101}, $\|\bm x-\bm y\|\ge\frac{\tz_1}{10}\gg1$ and $1<\rho'<\rho<\rho'+1$, we get
\begin{align*}
	&\ \ \log^{\rho}(\|\bm x-\bm y\|-\|\bm x_{l_1}-\bm y\|+1)-C(\rho)\log^{\rho}l_1\\
	\ge&\ \  \left(1-\frac{10^{5\rho'}}{\alpha\log^{\rho-\rho'}N_1}\right)\log^{\rho}(\|\bm x-\bm y\|+1)
\end{align*}
and then
\begin{align}
\nonumber	&\ \ \left(1-\frac{10^{5\rho'}}{\alpha\log^{\rho-\rho'} N_1}\right)\left(\log^{\rho}(\|\bm x-\bm y\|-\|\bm x_{l_1}-\bm y\|+1)-C(\rho)\log^{\rho}l_1\right)\\
\label{103}	\ge&\ \ \left(1-\frac{3\times 10^{5\rho'}}{\alpha\log^{\rho-\rho'}N_1}\right)\log^{\rho}(\|\bm x-\bm y\|+1).
\end{align}
Since $\|\bm x-\bm y\|\ge\frac{\tz_1}{10}\gg N_1\gg1$, \eqref{102} and \eqref{103}, we have
\begin{align*}
		|\boT_{\tilde{\Omega}_{\bm k}^{1}}^{-1}(\bm x,\bm y)|&\le e^{-\frac{3}{4}\alpha\left(1-\frac{10\times 10^{5\rho'}}{\alpha\log^{\rho-\rho'}N_1}\right)\log^{\rho}(\|\bm x-\bm y\|+1)}\cdot e^{3\times 10^{5\rho'}\log^{\rho'}(N_1+1)}\\
		&\le e^{-\frac{3}{4}\alpha\left(1-\frac{20\times 10^{5\rho'}}{\alpha\log^{\rho-\rho'}N_1}\right)\log^{\rho}(\|\bm x-\bm y\|+1)}.
\end{align*}
\item[\textbf{Case 2-3}:] $l_1\le l^*$ and $\bm x_{l_1}\in 2\Omega_{\bm k}^{1}$. From a similar argument of \eqref{102}, we obtain
\begin{align}\label{105}
			\nonumber	|\boT_{\tilde{\Omega}_{\bm k}^{1}}^{-1}(\bm x,\bm y)|&\le e^{-\frac{3}{4}\alpha\left(1-\frac{6\times 10^{5\rho'}}{\alpha\log^{\rho-\rho'} N_1}\right)\left(\log^{\rho}(\|\bm x-\bm y\|-\|\bm x_{l_1}-\bm y\|+1)-C(\rho)\log^{\rho}l_1\right)}\\
			&\ \ \cdot|\boT_{\tilde{\Omega}_{\bm k}^{1}}^{-1}(\bm x_{l_1},\bm y)|
\end{align}
By Corollary \ref{ite0g} (since $\tO_{\bm k}^{1}\setminus2\Omega_{\bm k}^{1}$ is $0$-good), there is some $\bm y'\in2\Omega_{\bm k}^{1}$ such that
\begin{align}\label{106}
	|\boT_{\tilde{\Omega}_{\bm k}^{1}}^{-1}(\bm x_{l_1},\bm y)|\le (\#\tO_{\bm k}^{1})\cdot e^{-\frac{3}{4}\alpha\log^{\rho}(1+\|\bm y-\bm y'\|)}\cdot|\boT_{\tilde{\Omega}_{\bm k}^{1}}^{-1}(\bm x_{l_1},\bm y')|\cdot\delta_0^{-3}.
\end{align}
According to $\bm x_{l_1}, \bm y'\in2\Omega_{\bm k}^{1}$, we have $\|\bm x_{l_1}-\bm y'\|\le 4\zeta_1$. Since \eqref{quaeq} and $\|\bm x-\bm y\|\ge\frac{\tz_1}{10}\gg1$,
\begin{align}
\nonumber&\ \ \log^{\rho}(\|\bm x_{l_1}-\bm x\|+1)+\log^{\rho}(1+\|\bm y-\bm y'\|)\\
\nonumber\ge&\ \ \log^{\rho}(1+\|\bm x_{l_1}-\bm x\|+\|\bm y-\bm y'\|)-C(\rho)\log^{\rho}2\\
\nonumber\ge&\ \ \log^{\rho}(1+\|\bm x-\bm y\|-\|\bm x_{l_1}-\bm y'\|)-C(\rho)\log^{\rho}2\\
\label{107}\ge&\ \ \log^{\rho}(1+\|\bm x-\bm y\|-4\zeta_1)-C(\rho)\log^{\rho}2.
\end{align}
From \eqref{803}, $\tz_1\gg\zeta_1$, $1<\rho'<\rho<\rho'+1$ and $N_1\gg1$, we also get
\begin{align}
\nonumber	&\ \ \log^{\rho}(1+\|\bm x-\bm y\|-4\zeta_1)\\
\nonumber\ge&\ \ \left(1-\frac{80\rho\zeta_1}{\tz_1\log N_1}\right)\log^{\rho}(\|\bm x-\bm y\|+1)\\
\label{108}	\ge&\ \ \left(1-\frac{10^{5\rho'}}{\alpha\log^{\rho-\rho'}N_1}\right)\log^{\rho}(\|\bm x-\bm y\|+1).
\end{align}
Combining \eqref{tk0}, \eqref{105}--\eqref{108}, $1<\rho'<\rho<\rho'+1$ and $\|\bm x-\bm y\|\ge\frac{\tz_1}{10}\gg1$ gives
\begin{align*}
	|\boT_{\tilde{\Omega}_{\bm k}^{1}}^{-1}(\bm x,\bm y)|&\le e^{-\frac{3}{4}\alpha\left(1-\frac{20\times 10^{5\rho'}}{\alpha\log^{\rho-\rho'}N_1}\right)\log^{\rho}(\|\bm x-\bm y\|+1)}.
\end{align*}
\end{itemize}
This finishes the proof.
	\end{proof}

	\ \\

	\begin{itemize}
		\item[\textbf{Step 3}]: \textbf{Estimates of general $1$-good $\lg$}.
	\end{itemize}
	In this final step, we complete the verification of property $\mathscr{P}_1$. Recall that a finite set  	$\lg\subset\Z^d$ is called 1-$\good$ if it satisfies the following conditions:
	\begin{align}\label{1gdef}
		\left\{\begin{array}{l}
			\lg\cap Q_0\cap\Omega_{\bm k}^1\ne\emptyset\Rightarrow\tilde{\Omega}_{\bm k}^1\subset\lg,\\
			\{\bm k\in P_1:\ \tilde{\Omega}_{\bm k}^1\subset\lg\}\cap Q_1=\emptyset.
		\end{array}\right.
	\end{align}
	
	To establish $\mathscr{P}_1$, we will synthesize three key components: (a) The norm estimates of $\mathcal{T}_{\tilde{\Omega}_{\bm k}^1}^{-1}$ obtained previously; (b) Schur's test for operator bounds; (c) The resolvent identity technique.
	\begin{thm}\label{1g}
		If $\lg$ is $1$-$\good$, then
		\begin{align}
			\nonumber\|\boT_\lg^{-1}\|&\le2\delta_0^{-3}\sup_{\{\bm k\in P_1:\ \tO_{\bm k}^1\subset\lg\}}\left(\|\theta+\bm k\cdot\bm\omega-\theta_1\|_{\T}^{-1}\cdot\|\theta+\bm k\cdot\bm\omega+\theta_1\|_{\T}^{-1}\right)\\
			\label{1gnorm}&<\delta_1^{-3},\\
		\label{1gdecay}	|\boT_{\lg}(\bm x,\bm y)|&<e^{-\alpha_1\|\bm x-\bm y\|}\ \text{for $\|\bm x-\bm y\|\ge10\tz_1$,}
		\end{align}
	where $\alpha_1=\frac{3}{4}\alpha\left(1-\frac{50\times 10^{5\rho'}}{\alpha\log^{\rho-\rho'}N_1}\right)$ is defined in \eqref{alphas}.
	\end{thm}
	\begin{proof}
	First,  we prove \eqref{1gnorm} by Schur's test. Define
		\begin{align*}
			\widetilde{P}_1=\left\{\bm k\in P_1:\ \lg\cap\Omega_{\bm k}^1\cap Q_0\ne\emptyset\right\}\subset P_1.
		\end{align*}
		From the definition of $1$-good set (cf. \eqref{1gdef}), we have $\widetilde{P}_1\cap Q_1=\emptyset$. For every $\bm w\in\lg$, define its block in $\lg$:
		\begin{align}\label{neighbor1}
			O(\bm w)=\left\{\begin{array}{ll}
				\lg_{\frac{1}{2}N_1}(\bm w)\cap \lg, & \text{if $\bm w\notin\bigcup_{\bm k\in\widetilde{P}_1}2\Omega_{\bm k}^{1}$,}\\
				\tilde{\Omega}_{\bm k}^1, & \text{if $\bm w\in2\Omega_{\bm k}^{1}$ for some $\bm k\in\widetilde{P}_1\setminus Q_1$.}
			\end{array}\right.
		\end{align}
		Since $\lg$ is $1$-good, $\boT_{O(\bm x)}^{-1}$ exists for each $\bm x\in\lg$. Hence we can define $\boL$ and $\boK$ as
		\begin{align}\label{defL1}
			\boL(\bm x,\bm y)=\left\{\begin{array}{ll}
				\boT_{O(\bm x)}^{-1}(\bm x,\bm y), & \text{for $\bm x\in\lg$ and $\bm y\in O(\bm x)$,}\\
				0, & \text{else,}
			\end{array}\right.
		\end{align}
		and
		\begin{align}\label{defK1}
				\boK(\bm x,\bm y)=\left\{\begin{array}{ll}
				\sum\limits_{\bm z\in O(\bm x)}\boL(\bm x,\bm z)\boW(\bm z,\bm y), & \text{for $\bm x\in\lg$ and $\bm y\in\lg\setminus O(\bm x)$,}\\
				0, & \text{else}.
			\end{array}\right.
		\end{align}
		Direct computations shows
		\begin{align*}
			\boL \boT_{\lg}=\boI_{\lg}+\ep\boK.
		\end{align*}
			Next, we will prove the invertibility of $\boI_{\lg}+\ep\boK$ through the following steps: (a) Schur's test yields $\|\boK\|\le D\left(\frac{5}{16}\alpha\right)<+\infty$, where $D(\et)$ is defined in \eqref{Det}; (b) Neumann series expansion is valid for $\|\ep\boK\|<1$. This ensures the existence of $\boT_{\lg}^{-1}=(\boI_{\lg}+\ep\boK)^{-1}\boL$.
			
			Now, we need to estimate $\sup_{\bm x\in\lg}\sum_{\bm y\in\lg}|\boK(\bm x,\bm y)|$ and $\sup_{\bm y\in\lg}\sum_{\bm x\in\lg}|\boK(\bm x,\bm y)|$ respectively for controlling $\|\boK\|$. We estimate $|\boK(\bm x,\bm y)|$  first. We divide the discussion into three cases:
		\begin{itemize}
			\item[\textbf{Case 1}:] $\bm x\notin\bigcup_{\bm k\in\widetilde{P}_1}2\Omega_{\bm k}^{1}$ and $\bm y\in \lg\setminus O(\bm x)$. In this case,  we have $O(\bm x)$ is $0$-good. From \eqref{wphi}, \eqref{0gl2} and \eqref{0gde}, we have
			\begin{align}
				\nonumber|\boK(\bm x,\bm y)|&\le\sum_{\bm z\in O(\bm x)}|\boT_{O(\bm x)}^{-1}(\bm x,\bm z)|\cdot|\boW(\bm z,\bm y)|\\
				\label{1K12}&\le ({\rm I})+({\rm II}),
			\end{align}
						where
			\begin{align*}
				({\rm I})=	\sum_{\bm z\in O(\bm x)\setminus\{\bm x\}}e^{-\frac{3}{4}\alpha\log^{\rho}(1+\|\bm x-\bm z\|)}\cdot e^{-\alpha\log^{\rho}(1+\|\bm z-\bm y\|)}
			\end{align*}
			and
			\begin{align*}
				({\rm II})=(2\kappa_1^{-1}\delta_0^{-2})\cdot e^{-\alpha\log^{\rho}(1+\|\bm x-\bm y\|)}.
			\end{align*}
			For $({\rm I})$, by \eqref{quaeq} and $\bm y\notin O(\bm x)$ ($\|\bm x-\bm y\|\ge\frac{1}{2}N_1\gg1$), we get
			\begin{align}
				\nonumber({\rm I})&=\sum_{\bm z\in O(\bm x)\setminus\{\bm x\}}e^{-\frac{3}{4}\alpha\log^{\rho}(1+\|\bm x-\bm z\|)}\cdot e^{-\alpha\log^{\rho}(1+\|\bm z-\bm y\|)}\\
				\nonumber&\le  e^{-\frac{3}{4}\alpha\log^{\rho}(1+\|\bm x-\bm y\|)+\frac{3}{4}\alpha\cdot C(\rho)\log^{\rho}2}\sum_{\bm z\in O(\bm x)\setminus\{\bm x\}}e^{-\frac{\alpha}{10}\log^{\rho}(1+\|\bm z-\bm y\|)}\\
				\nonumber&\le D\left(\frac{\alpha}{10}\right)\cdot e^{-\frac{3}{4}\alpha\log^{\rho}(1+\|\bm x-\bm y\|)+\frac{3}{4}\alpha\cdot C(\rho)\log^{\rho}2}\\
			\label{1K1}	&\le \frac{1}{2}e^{-\frac{5}{8}\alpha\log^{\rho}(1+\|\bm x-\bm y\|)},
			\end{align}
			where $D\left(\frac{\alpha}{10}\right)$ is defined in \eqref{Det}. 		For $({\rm II})$, since $\|\bm x-\bm y\|\ge\frac{1}{2}N_1\gg1$, \eqref{indpa} and $1<\rho'<\rho$, we obtain
			\begin{align}
			\label{1K2}({\rm II})&= 2\kappa_1^{-1}\delta_0^{-2}e^{-\alpha\log^{\rho}\left(1+\|\bm x-\bm y\|\right)}\le \frac{1}{2}e^{-\frac{5}{8}\alpha\log^{\rho}\left(1+\|\bm x-\bm y\|\right)}.
			\end{align}
			Therefore, from \eqref{1K12}, \eqref{1K1} and \eqref{1K2}, we have
			\begin{align}\label{Kle11}
			|\boK(\bm x,\bm y)|\le e^{-\frac{5}{8}\alpha\log^{\rho}\left(1+\|\bm x-\bm y\|\right)}.
			\end{align}
			
			\item[\textbf{Case 2}:] $\bm x\in 2\Omega_{\bm k}^{1}$ for some $\bm k\in\widetilde{P}_1\setminus Q_1$ and $\bm y\in \lg\setminus O(\bm x)$. In this case, we define $X:=\lg_{\frac{\tz_1}{9}}(\bm k)\subset O(\bm x)$. Hence
					\begin{align}
				\nonumber|\boK(\bm x,\bm y)|&\le\sum_{\bm z\in O(\bm x)}|\boT_{O(\bm x)}^{-1}(\bm x,\bm z)|\cdot|\boW(\bm z,\bm y)|\\
				\label{1K34}&\le ({\rm III})+({\rm IV}),
			\end{align}
			where
			\begin{align*}
				({\rm III})=\sum_{\bm z\in O(\bm x)\setminus X}|\boT_{O(\bm x)}^{-1}(\bm x,\bm z)|\cdot|\boW(\bm z,\bm y)|
			\end{align*}
			and
			\begin{align*}
				({\rm IV})=\sum_{\bm z\in X}|\boT_{O(\bm x)}^{-1}(\bm x,\bm z)|\cdot|\boW(\bm z,\bm y)|.
			\end{align*}
					For $({\rm III})$, since \eqref{wphi}, \eqref{quaeq}, $\bm z\in O(\bm x)\setminus X$ ($\|\bm z-\bm x\|\ge\frac{\tz_1}{10}$), $\bm y\notin O(\bm x)$ ($\|\bm x-\bm y\|\ge\frac{\tz_1}{3}$) and \eqref{tg1a}, we have
			\begin{align}
				\nonumber({\rm III})=&\ \ \sum_{\bm z\in O(\bm x)\setminus X}|\boT_{O(\bm x)}^{-1}(\bm x,\bm z)|\cdot|\boW(\bm z,\bm y)|\\
				\nonumber\le &\ \ \sum_{\bm z\in O(\bm x)\setminus X} e^{-\alpha'_0\log^{\rho}(1+\|\bm x-\bm z\|)}\cdot e^{-\alpha\log^{\rho}(1+\|\bm z-\bm y\|)}\\
				\nonumber\le &\ \ e^{-\alpha'_0\log^{\rho}(1+\|\bm x-\bm y\|)+\alpha'_0\cdot C(\rho)\log^{\rho}2}\sum_{\bm z\in O(\bm x)\setminus X}e^{-\frac{1}{10}\alpha\log^{\rho}(1+\|\bm z-\bm y\|)}\\
					\nonumber\le &\ \  D\left(\frac{\alpha}{10}\right)\cdot e^{-\alpha'_0\log^{\rho}(1+\|\bm x-\bm y\|)+\alpha'_0\cdot C(\rho)\log^{\rho}2}\\
				\label{1K3}	\le&\ \  \frac{1}{2}e^{-\frac{5}{8}\alpha'_0\log^{\rho}(1+\|\bm x-\bm y\|)},
			\end{align}
			where $D\left(\frac{\alpha}{10}\right)$ is defined in \eqref{Det}.  For $({\rm IV})$, from \eqref{wphi}, \eqref{to1-1} and $\bm k\notin Q_1$, we obtain
			\begin{align}
				\label{51202}\sum_{\bm z\in X}|\boT_{O(\bm x)}^{-1}(\bm x,\bm z)|\cdot|\boW(\bm z,\bm y)|\le \delta_0^{-2}\delta_1^{-2}\sum_{\bm z\in X} e^{-\alpha\log^{\rho}(1+\|\bm z-\bm y\|)}.
			\end{align}
			By $\bm z\in X$, $\bm y\in\lg\setminus O(\bm x)$ ($\|\bm y-\bm x\|\ge\frac{\tz_1}{3}\ge2\|\bm x-\bm z\|$) and \eqref{exl}, we get
			\begin{align}
				\nonumber&\ \ \log^{\rho}(1+\|\bm y-\bm z\|)\\
				\nonumber\ge&\ \ 	\log^{\rho}(1+\|\bm y-\bm x\|-\|\bm x-\bm z\|)\\
				\nonumber	\ge&\ \ \left(1-2\rho\frac{\|\bm x-\bm z\|}{(1+\|\bm y-\bm x\|)\log(1+\|\bm y-\bm x\|)}\right)\log^{\rho}(1+\|\bm y-\bm x\|)\\
				\label{51203}\ge&\ \ \left(1-\frac{1}{\log N_1}\right)\log^{\rho}(1+\|\bm y-\bm x\|).
			\end{align}
			Hence, combining \eqref{51202}, \eqref{51203},  $\|\bm x-\bm y\|\ge\frac{\tz_1}{3}\gg1$, \eqref{indpa} and $1<\rho'<\rho$ gives
			\begin{align}
				\nonumber({\rm IV})=&\ \ \sum_{\bm z\in X}|\boT_{O(\bm x)}^{-1}(\bm x,\bm z)|\cdot|\boW(\bm z,\bm y)|\\
				\nonumber\le&\ \  \delta_0^{-2}\delta_1^{-2}\cdot e^{-\frac{3}{4}\alpha\log^{\rho}(1+\|\bm y-\bm x\|)}\cdot\sum_{\bm z\in X}e^{-\frac{1}{10}\alpha\log^{\rho}(1+\|\bm z-\bm y\|)}\\
			\nonumber	\le&\ \  \delta_0^{-2}\delta_1^{-2}\cdot D\left(\frac{\alpha}{10}\right)\cdot e^{-\frac{3}{4}\alpha\log^{\rho}(1+\|\bm y-\bm x\|)}\\
				\label{1K4}\le &\ \ \frac{1}{2}e^{-\frac{5}{8}\alpha'_0\log^{\rho}(1+\|\bm x-\bm y\|)},
			\end{align}
			where $D\left(\frac{\alpha}{10}\right)$ is defined in \eqref{Det}.	From \eqref{1K34}, \eqref{1K3} and \eqref{1K4}, we have
			\begin{align}\label{Kle12}
			|\boK(\bm x,\bm y)|\le e^{-\frac{5}{8}\alpha'_0\log^{\rho}(1+\|\bm x-\bm y\|)}.
			\end{align}
			
			\item[\textbf{Case 3}:] $\bm x\in\lg$ and $\bm y\in O(\bm x)$. From \eqref{defK1}, we get
			\begin{align}\label{Kle13}
				|\boK(\bm x,\bm y)|=0\le e^{-\frac{5}{8}\alpha'_0\log^{\rho}(1+\|\bm x-\bm y\|)}.
			\end{align}
		\end{itemize}
	In summary, we obtain
	\begin{align*}
			|\boK(\bm x,\bm y)|\le e^{-\frac{5}{8}\alpha'_0\log^{\rho}(1+\|\bm x-\bm y\|)}\text{ for all $\bm x,\bm y\in\lg$}.
	\end{align*}
	Therefore,
	\begin{align*}
		\sup_{\bm x\in\lg}\sum_{\bm y\in\lg}|\boK(\bm x,\bm y)|&\le D\left(\frac{5}{8}\alpha'_0\right)<+\infty,\\
		\sup_{\bm y\in\lg}\sum_{\bm x\in\lg}|\boK(\bm x,\bm y)|&\le D\left(\frac{5}{8}\alpha'_0\right)<+\infty,
	\end{align*}
	where $D\left(\frac{5}{8}\alpha'_0\right)$ is defined in \eqref{Det}. By Schur's test, we can get
	\begin{align*}
		\|\boK\|\le D\left(\frac{5}{8}\alpha'_0\right)\le D\left(\frac{5}{16}\alpha\right)<+\infty.
	\end{align*}
	From $\ep\ll1$, we have $\boI_{\lg}+\ep\boK$ is invertible and
	\begin{align}\label{I+eK-1}
		\|(\boI_{\lg}+\ep\boK)^{-1}\|\le 2.
	\end{align}
	At this time,
	\begin{align}\label{tlg-1}
	\boT_{\lg}^{-1}=(\boI_{\lg}+\ep\boK)^{-1}\boL
\end{align}
	 exists. Then, we estimate $\|\boL\|$ in order to estimate $\|\boT_{\lg}^{-1}\|$ by using similar methods.
	
	We deal with $\sup_{\bm x\in\lg}\sum_{\bm y\in\lg}|\boL(\bm x,\bm y)|$ first. We divide the discussion into two cases:
	\begin{itemize}
		\item[\textbf{Case 1}:] $\bm x\in\lg\setminus\bigcup_{\bm k\in\widetilde{P}_1}2\Omega_{\bm k}^{1}$. From \eqref{0gl2} and \eqref{defK1}, we have
			\begin{align}
			\nonumber\sum_{\bm y\in\lg}|\boL(\bm x,\bm y)|&=\sum_{\bm y\in O(\bm x)}|\boL(\bm x,\bm y)|\le (\# O(\bm x))\cdot\|\boT_{O(\bm x)}^{-1}\|\\
			\label{1L11}&<(N_1+1)^d\cdot2\kappa_1^{-1}\delta_0^{-2}<\delta_0^{-\frac{5}{2}}.
		\end{align}
		\item[\textbf{Case 2}:] $\bm x\in 2\Omega_{\bm k}^{1}$ for some $\bm k\in\widetilde{P}_1\setminus Q_1$. In this case, by \eqref{tk0} and \eqref{defK1}, we get
			\begin{align}
		\nonumber	\sum_{\bm y\in\lg}|\boL(\bm x,\bm y)|&=\sum_{\bm y\in O(\bm x)}|\boL(\bm x,\bm y)|\le (\# (O(\bm x)))\cdot\|\boT_{O(\bm x)}^{-1}\|\\
		\label{1L12}	&<\delta_0^{-3}\|\theta+\bm k\cdot\bm\omega-\theta_1\|_{\T}^{-1}\cdot\|\theta+\bm k\cdot\bm\omega+\theta_1\|_{\T}^{-1}.
		\end{align}
	\end{itemize}
	Since $\theta,\theta_1\in \D_R$, we obtain
	\begin{align*}
		\|\theta+\bm k\cdot\bm\omega\pm\theta_1\|\le\sqrt{4R^2+\frac{1}{4}}
	\end{align*}
	and
	\begin{align}\label{1L13}
		\delta_0^{-3}\|\theta+\bm k\cdot\bm\omega-\theta_1\|_{\T}^{-1}\cdot\|\theta+\bm k\cdot\bm\omega+\theta_1\|_{\T}^{-1}>\delta_0^{-\frac{5}{2}}.
	\end{align}
	Combining \eqref{1L11}--\eqref{1L13}, we have
		\begin{align}
		\nonumber&\ \ \sup_{\bm x\in\lg}\sum_{\bm y\in\lg}|\boL(\bm x,\bm y)|\\
	\label{L1}	\le&\ \  \delta_0^{-3}\sup_{\{\bm k\in P_1:\ \tO_{\bm k}^1\subset\lg\}}\left(\|\theta+\bm k\cdot\bm\omega-\theta_1\|_{\T}^{-1}\cdot\|\theta+\bm k\cdot\bm\omega+\theta_1\|_{\T}^{-1}\right).
	\end{align}
	
	Now, we estimate $\sup_{\bm y\in\lg}\sum_{\bm x\in\lg}|\boL(\bm x,\bm y)|$. We again divide the discussion into two cases:
	\begin{itemize}
		\item[\textbf{Case 1}:] $\bm y\in\lg\setminus\bigcup_{\bm k\in\widetilde{P}_1}\tO_{\bm k}^{1}$. In this case, 
		\begin{align*}
		\bm y\in O(\bm x)\text{ iff }\bm x\in \lg\cap\lg_{\frac{1}{2}N_1}(\bm y).
	\end{align*}
		 At this time, if $\bm x\in \lg\cap\lg_{\frac{1}{2}N_1}(\bm y)$, then $\bm x\in\lg\setminus\bigcup_{\bm k\in\widetilde{P}_1}2\Omega_{\bm k}^{1}$ ($O(\bm x)=\lg\cap\lg_{\frac{1}{2}N_1}(\bm x)$ is $0$-good). Hence by \eqref{indpa}, \eqref{0gl2}, we get
		\begin{align}
		\nonumber	\sum_{\bm x\in\lg}|\boL(\bm x,\bm y)|&=\sum_{\bm x\in \lg\cap\lg_{\frac{1}{2}N_1}(\bm y)}|\boT_{O(\bm x)}^{-1}(\bm x,\bm y)|\\
		\label{1L1}	&\le (\# (\lg\cap\lg_{\frac{1}{2}N_1}(\bm y)))\cdot2\kappa_1^{-1}\delta_0^{-2}<\delta_0^{-\frac{5}{2}}.
		\end{align}
		\item[\textbf{Case 2}:] $\bm y\in \tO_{\bm k}^{1}$ for some $\bm k\in\widetilde{P}_1\setminus Q_1$. In this case, 
		\begin{align*}
		\bm y\in O(\bm x)\text{ iff }\bm x\in 2\Omega_{\bm k}^{1}\bigcup\left(\lg\cap\lg_{\frac{1}{2}N_1}(\bm y)\right).
	\end{align*}
		 Therefore, from \eqref{indpa}, \eqref{0gl2} and \eqref{tk0}, we obtain
		\begin{align}
\nonumber\sum_{\bm x\in\lg}	|\boL(\bm x,\bm y)|&\le \sum_{\bm x\in2\Omega_{\bm k}^1}|\boT_{O(\bm x)}^{-1}(\bm x,\bm y)|+\sum_{\bm x\in \left(\lg\cap\lg_{\frac{1}{2}N_1}(\bm y)\right)\setminus(2\Omega_{\bm k}^1)}|\boT_{O(\bm x)}^{-1}(\bm x,\bm y)|\\
\nonumber&\le (\#(2\Omega_{\bm k}^{1}))\delta_0^{-2}\|\theta+\bm k\cdot\bm\omega-\theta_1\|_{\T}^{-1}\cdot\|\theta+\bm k\cdot\bm\omega+\theta_1\|_{\T}^{-1}\\
\nonumber&\ \ +(\# (\lg\cap\lg_{\frac{1}{2}N_1}(\bm y)))\cdot2\kappa_1^{-1}\delta_0^{-2}\\
\label{1L2}&\le \delta_0^{-3}\|\theta+\bm k\cdot\bm\omega-\theta_1\|_{\T}^{-1}\cdot\|\theta+\bm k\cdot\bm\omega+\theta_1\|_{\T}^{-1}.
\end{align}
	\end{itemize}
	Combining \eqref{1L13}, \eqref{1L1} and \eqref{1L2} gives
	\begin{align}
	\nonumber	&\ \ \sup_{\bm y\in\lg}\sum_{\bm x\in\lg}|\boL(\bm x,\bm y)|\\
		\label{L2}\le&\ \  \delta_0^{-3}\sup_{\{\bm k\in P_1:\ \tO_{\bm k}^1\subset\lg\}}\left(\|\theta+\bm k\cdot\bm\omega-\theta_1\|_{\T}^{-1}\cdot\|\theta+\bm k\cdot\bm\omega+\theta_1\|_{\T}^{-1}\right).
	\end{align}
	Then by \eqref{L1}, \eqref{L2} and Schur's test, we know
	\begin{align}\label{Lnorm1}
		\|\boL\|\le\delta_0^{-3}\sup_{\{\bm k\in P_1:\ \tO_{\bm k}^1\subset\lg\}}\left(\|\theta+\bm k\cdot\bm\omega-\theta_1\|_{\T}^{-1}\cdot\|\theta+\bm k\cdot\bm\omega+\theta_1\|_{\T}^{-1}\right).
	\end{align}
	Finally, from \eqref{1gdef}, \eqref{I+eK-1}, \eqref{tlg-1} and \eqref{Lnorm1}, we prove
	\begin{align*}
		\|\boT_{\lg}^{-1}\|\le 2\delta_0^{-3}\sup_{\{\bm k\in P_1:\ \tO_{\bm k}^1\subset\lg\}}\left(\|\theta+\bm k\cdot\bm\omega-\theta_1\|_{\T}^{-1}\cdot\|\theta+\bm k\cdot\bm\omega+\theta_1\|_{\T}^{-1}\right)<\delta_1^{-3}.
	\end{align*}
	This completes the proof of \eqref{1gnorm}.
	
	We begin the proof of the off-diagonal decay estimate in \eqref{1gdecay}. The argument requires the following lemma:
	\begin{lem}\label{ite1}
		Assuming $\bm u\in 2\Omega_{\bm k}^{1}$ for some $\bm k\in\widetilde{P}_1\setminus Q_1$ and $\bm v\in\lg\setminus O(\bm u)=\lg\setminus\tO_{\bm k}^{1}$, then there exists some $\bm u'\in \lg\setminus O(\bm u)$ such that
		\begin{align}\label{ite1eq}
			|\boT_{\lg}^{-1}(\bm u,\bm v)|\le e^{-\alpha''_0\log^{\rho}(1+\|\bm u-\bm u'\|)}\cdot	|\boT_{\lg}^{-1}(\bm u',\bm v)|,
		\end{align}
		where $\alpha''_0=\frac{3}{4}\alpha\left(1-\frac{30\times 10^{5\rho'}}{\alpha\log^{\rho-\rho'}N_1}\right)$.
	\end{lem}
		\begin{proof}
		For a detailed proof, we refer to the Appendix \ref{app}.
	\end{proof}
	
	In the following, we always assume that $\|\bm x-\bm y\|\ge10\tz_1$. The proof of \eqref{1gdecay} proceeds via an iterative application of \eqref{ite0g0} and \eqref{ite1eq}. For $\bm z\in\lg$, from Corollary \ref{ite0g} and Lemma \ref{ite1}, there exists $\bm z'\in\lg\setminus O(\bm z)$ (cf. \eqref{neighbor1}) such that
	\begin{align}\label{1iteeq1}
		|\boT_{\lg}^{-1}(\bm z,\bm y)|\le e^{-\alpha''_0\log^{\rho}(1+\|\bm z-\bm z'\|)}\cdot|\boT_{\lg}^{-1}(\bm z',\bm y)|,
	\end{align}
	where $\alpha''_0=\frac{3}{4}\alpha\left(1-\frac{30\times 10^{5\rho'}}{\alpha\log^{\rho-\rho'}N_1}\right)$.
	
	Next, iterating \eqref{1iteeq1} for $m\ge2$ steps leads to the following: there exist $\bm x_1,\bm x_2,\cdots,\bm x_m\in\lg$ such that
	\begin{align}\label{1iteeq2}
		|\boT_{\lg}^{-1}(\bm x,\bm y)|\le e^{-\alpha''_0\left(\sum_{k=0}^{m-1}\log^{\rho}(1+\|\bm x_{k+1}-\bm x_{k}\|)\right)}	|\boT_{\lg}^{-1}(\bm x_m,\bm y)|,
	\end{align}
	where $\bm x_0=\bm x$ and $\bm x_{k+1}=\bm x'$, $k=0,1,\cdots,m-1$.  We define $n\ge1$ to be the smallest integer so that $\bm x_n\in O(\bm y)$ (cf. \eqref{neighbor1}). We then have $\bm x_i\notin O(\bm y)$ for $i=0,1,\cdots,n-1$. We divide the discussion into two cases:
	\begin{itemize}
		\item[\textbf{Case 1}:] $n\le \frac{\alpha\log^{\rho}(1+\|\bm x-\bm y\|)+3\times 10^{5\rho'}\log(N_1+1)}{\frac{\alpha}{2}\log^{\rho}\left(\frac{N_1}{2}+1\right)}$. Using Lemma \ref{qua} and the triangle inequality implies
		\begin{align}
			\nonumber&\ \ \sum_{k=0}^{n-1}\log^{\rho}(1+\|\bm x_{k+1}-\bm x_{k}\|)\\
			\nonumber\ge&\ \ \log^{\rho}\left(1+\sum_{k=0}^{n-1}\|\bm x_{k+1}-\bm x_k\|\right)-C(\rho)\log^{\rho}n\\
			\label{51802}\ge&\ \  \log^{\rho}(1+\|\bm x_n-\bm x\|)-C(\rho)\log^{\rho}n.
		\end{align}
		Since $\bm x_n\in O(\bm y)$ and $\diam(O(\bm y))\le \tz_1$, we have
		\begin{align*}
			\|\bm x-\bm x_n\|\ge\|\bm x-\bm y\|-\|\bm x_n-\bm y\|\l\ge\|\bm x-\bm y\|-\tz_1.
		\end{align*}
		By $\|\bm x-\bm y\|\ge10\tz_1\gg N_1\gg1$ and $1<\rho'<\rho<\rho'+1$, applying Lemma \ref{exl} with $x=\|\bm x-\bm y\|$ and $y=\tz_1$ gives
		\begin{align}
		\nonumber	&\ \ \log^{\rho}(1+\|\bm x_n-\bm x\|)\\
		\nonumber\ge&\ \ \log^{\rho}(1+\|\bm x-\bm y\|-\tz_1)\\
			\nonumber\ge&\ \  \left(1-2\rho\frac{\tz_1}{(1+\|\bm x-\bm y\|)\log(1+\|\bm x-\bm y\|)}\right)\log^{\rho}(1+\|\bm x-\bm y\|)\\
			\label{51803}\ge&\ \ \left(1-\frac{10^{5\rho'}}{\alpha\log^{\rho-\rho'} N_1}\right)\log^{\rho}(1+\|\bm x-\bm y\|).
		\end{align}
		Under the same conditions, we have
		\begin{align}
		\nonumber	C(\rho)\log^{\rho}n&\le C(\rho)\log^{\rho}(2\alpha\log^{\rho}(1+\|\bm x-\bm y\|))\le \log(1+\|\bm x-\bm y\|)\\
		\label{51804}	&\le \frac{ 10^{5\rho'}\log^{\rho}(1+\|\bm x-\bm y\|)}{\alpha\log^{\rho-\rho'}N_1}.
		\end{align}
Combining \eqref{51802}, \eqref{51803} and \eqref{51804} shows
\begin{align}\label{51805}
	\sum_{k=0}^{n-1}\log^{\rho}(1+\|\bm x_{k+1}-\bm x_{k}\|)\ge \left(1-\frac{2\times 10^{5\rho'}}{\alpha\log^{\rho-\rho'} N_1}\right)\log^{\rho}(1+\|\bm x-\bm y\|).
\end{align}
From $\|\bm x-\bm y\|\ge10\tz_1\gg N_1\gg1$, \eqref{indpa}, \eqref{1gnorm}, \eqref{1iteeq2} and  \eqref{51805}, we have
\begin{align*}
		|\boT_{\lg}^{-1}(\bm x,\bm y)|&\le e^{-\alpha''_0 \left(1-\frac{2\times 10^{5\rho'}}{\alpha\log^{\rho-\rho'} N_1}\right)\log^{\rho}(1+\|\bm x-\bm y\|)}	\delta_1^{-3}\\
		&\le e^{-\alpha_1\log^{\rho}(1+\|\bm x-\bm y\|)},
\end{align*}
		where $\alpha_1=\alpha''_0-\frac{15\times 10^{5\rho'}}{\log^{\rho-\rho'} N_1}=\frac{3}{4}\alpha\left(1-\frac{50\times 10^{5\rho'}}{\alpha\log^{\rho-\rho'}N_1}\right)$.

		\item[\textbf{Case 2}:] $n>\frac{\alpha\log^{\rho}(1+\|\bm x-\bm y\|)+3\times 10^{5\rho'}\log(N_1+1)}{\frac{\alpha}{2}\log^{\rho}\left(\frac{N_1}{2}+1\right)}$. In this case, from \eqref{1gnorm}, \eqref{1iteeq2}, $\|\bm x_{k+1}-\bm x_k\|\ge\frac{N_1}{2}$ (since $\bm x_{k+1}\in\lg\setminus O(\bm x_k)$), $k=0,1\cdots,n-1$ and
		\begin{align*}
		\alpha''_0=\frac{3}{4}\alpha\left(1-\frac{30\times 10^{5\rho'}}{\alpha\log^{\rho-\rho'}N_1}\right)>\frac{\alpha}{2},
	\end{align*}
		it follows that
		\begin{align*}
				|\boT_{\lg}^{-1}(\bm x,\bm y)|&\le e^{-\alpha''_0\left(\sum_{k=0}^{n-1}\log^{\rho}(1+\|\bm x_{k+1}-\bm x_{k}\|)\right)}	|\boT_{\lg}^{-1}(\bm x_n,\bm y)|\\
				&\le e^{-n\cdot\alpha''_0\log^{\rho}\left(\frac{N_1}{2}+1\right)}e^{3\times 10^{5\rho'}\log^{\rho}(N_1+1)}\\
				&\le e^{-\alpha\log^{\rho}(1+\|\bm x-\bm y\|)}\le e^{-\alpha_1\log^{\rho}(1+\|\bm x-\bm y\|)}.
		\end{align*}
	\end{itemize}
	This concludes the proof of Lemma \ref{1g}.
	\end{proof}
	
	\subsection{The proof of Theorem \ref{ind}: (from $\mathscr{P}_s$ to $\mathscr{P}_{s+1}$)}
	We have completed the proof of $\mathscr{P}_1$ in Subsection \ref{vo1}. Now, assuming the validity of $\mathscr{P}_t$ $(1\le t\le s)$, the proof of Theorem \ref{ind} reduces to establishing.

	Recall that
	\begin{align*}
		Q_{s}^{\pm}&=\left\{\bm k\in P_{s}:\ \|\theta+\bm k\cdot\bm\omega\pm\theta_{s}\|_{\T}<\delta_{s}\right\},\ Q_{s}=Q_{s}^+\cup Q_{s}^-,\\
		\tilde{Q}_{s}^{\pm}&=\left\{\bm k\in P_{s}:\ \|\theta+\bm k\cdot\bm\omega\pm\theta_{s}\|_{\T}<\delta_{s}^{\frac{1}{100}}\right\},\ \tilde{Q}_{s}=\tilde{Q}_{s}^+\cup \tilde{Q}_{s}^-.
	\end{align*}
	
	We distinguish the verification into three steps.
	
	\begin{itemize}
		\item[\textbf{Step 1}]: \textbf{Estimates of $\|\mathcal{T}_{\tilde{\Omega}_{\bm k}^{s+1}}^{-1}\|$}.
	\end{itemize}
	In this step, we apply the resolvent indentity and Rouch\'e's theorem to construct $\theta_{s+1}=\theta_{s+1}(E)$ such that
	\begin{align}
		\label{s+1T-1}\|\mathcal{T}_{\tilde{\Omega}_{\bm k}^{s+1}}^{-1}\|&<\delta_{s}^{-2}\|\theta+\bm k\cdot\bm\omega-\theta_{s+1}\|_{\T}^{-1}\cdot\|\theta+\bm k\cdot\bm\omega+\theta_{s+1}\|_{\T}^{-1}.
	\end{align}
	We again divide the discussion into two cases.
	
	\begin{itemize}
		\item[\textbf{Case 1}]: The case $(\bm C1)_s$ occurs, i.e.,
	\end{itemize}
	\begin{align}\label{dqs1-}
		\dist(\tilde{Q}_s^-,Q_s^+)>100N_{s+1}^{10}.
	\end{align}
	\begin{rem}
		We can prove similar to Remark \ref{dsq1} that
		\begin{align*}
			\dist(\tilde{Q}_s^-,Q_s^+)=\dist(\tilde{Q}_s^+,Q_s^-).
		\end{align*}
		Thus \eqref{dqs1-} also implies that
		\begin{align}\label{dqs1+}
			\dist(\tilde{Q}_s^+,Q_s^-)>100N_{s+1}^{10}.
		\end{align}
		%The proof is similar to that of Remark \ref{dsq1} and we omit the details.
	\end{rem}
	
	By \eqref{fs} and the definitions  of $Q_s^{\pm}$ and $\tilde{Q}_s^{\pm}$, we obtain
	\begin{align}
		\label{qspm}Q_s^{\pm}&=\left\{\bm k\in\Z^d+\frac{1}{2}\sum_{i=0}^{s-1}\bm l_i:\ \|\theta+\bm k\cdot\bm\omega\pm\theta_{s}\|_{\T}<\delta_{s}\right\},\\
		\nonumber\tilde{Q}_{s}^{\pm}&=\left\{\bm k\in\Z^d+\frac{1}{2}\sum_{i=0}^{s-1}\bm l_i:\ \|\theta+\bm k\cdot\bm\omega\pm\theta_{s}\|_{\T}<\delta_{s}^{\frac{1}{100}}\right\}.
	\end{align}
	
	Assuming \eqref{dqs1-} holds true, we define
	\begin{align}\label{ps+1}
		P_{s+1}=Q_s,\ \bm l_s=\bm 0.
	\end{align}
	By \eqref{qspm}, we have
	\begin{align}\label{ps+11}
		P_{s+1}\subset\left\{\bm k\in\Z^d+\frac{1}{2}\sum_{i=0}^{s}\bm l_i:\ \min_{\sigma=\pm1}(\|\theta+\bm k\cdot\bm\omega+\sigma\theta_{s}\|_{\T})<\delta_{s}\right\},
	\end{align}
	which proves \eqref{as1} in the case $(\bm C1)_{s+1}$. Thus from \eqref{dqs1+}, we obtain for $\bm k,\bm k'\in P_{s+1}$ with $\bm k\ne \bm k'$,
	\begin{align}\label{dk}
		\|\bm k-\bm k'\|\ge\min\left(100N_{s+1}^{10},\left(\frac{\g}{2\delta_s}\right)^{\frac{1}{\tau}}\right)\ge100N_{s+1}^{10}.
	\end{align}
	In the following, we associate each $\bm k\in P_{s+1}$ with blocks $\Omega_{\bm k}^{s+1}$, $2\Omega_{\bm k}^{s+1}$ and $\tilde{\Omega}_{\bm k}^{s+1}$ so that
	\begin{align*}
		&\lg_{N_{s+1}}(\bm k)\subset\Omega_{\bm k}^{s+1}\subset\lg_{N_{s+1}+50N_s^{100}}(\bm k),\\
		&\lg_{2N_{s+1}}(\bm k)\subset2\Omega_{\bm k}^{s+1}\subset\lg_{2N_{s+1}+50N_s^{100}}(\bm k),\\
		&\lg_{N_{s+1}^{10}}(\bm k)\subset\tilde{\Omega}_{\bm k}^{s+1}\subset\lg_{N_{s+1}^{10}+50N_s^{100}}(\bm k),
	\end{align*}
	and
	\begin{align}\label{sb}
		\left\{\begin{array}{l}
			\Omega_{\bm k}^{s+1}\cap\tilde{\Omega}_{\bm k'}^{s'}\ne\emptyset\ (s'<s+1)\Rightarrow\tilde{\Omega}_{\bm k'}^{s'}\subset\Omega_{\bm k}^{s+1},\\
			2\Omega_{\bm k}^{s+1}\cap\tilde{\Omega}_{\bm k'}^{s'}\ne\emptyset\ (s'<s+1)\Rightarrow\tilde{\Omega}_{\bm k'}^{s'}\subset2\Omega_{\bm k}^{s+1},\\
			\tilde{\Omega}_{\bm k}^{s+1}\cap\tilde{\Omega}_{\bm k'}^{s'}\ne\emptyset\ (s'<s+1)\Rightarrow\tilde{\Omega}_{\bm k'}^{s'}\subset\tilde{\Omega}_{\bm k}^{s+1},\\
			\dist(\tilde{\Omega}_{\bm k}^{s+1},\tilde{\Omega}_{\bm k'}^{s+1})>10\tilde{\zeta}_{s+1}\ \text{for}\ \bm k\ne \bm k'\in P_{s+1}.
		\end{array}\right.
	\end{align}
	Moreover, the translated set of $\tilde{\Omega}_{\bm k}^{s+1}$
	\begin{align}\label{trans1}
		\tilde{\Omega}_{\bm k}^{s+1}-\bm k\subset\Z^d+\frac{1}{2}\sum_{i=0}^s \bm l_i
	\end{align}
	is both independent of $\bm k\in P_{s+1}$ and symmetric about the origin. For full proofs of \eqref{sb} and \eqref{trans1}, we refer to page 23 of \cite{CSZ24a}. In summary, we have now established $(\bm a)_{s+1}$ and $(\bm b)_{s+1}$ in the case $(\bm C1)_{s}$.
	
	Now we turn to the proof of $(\bm c)_{s+1}$. For any $\bm k'\in Q_s\ (=P_{s+1})$, observe that
	\begin{align*}
		\tilde{\Omega}_{\bm k'}^s\subset\tilde{\Omega}_{\bm k'}^{s+1}.
	\end{align*}
	For each $\bm k\in P_{s+1}$, define
	\begin{align*}
		A_{\bm k}^{s+1}=A_{\bm k}^s.
	\end{align*}
	By construction,
	\begin{align*}
		A_{\bm k}^{s+1}\subset\Omega_{\bm k}^s\text{ and }\# A_{\bm k}^{s+1}=\# A_{\bm k}^s\le 2^s.
	\end{align*}
	 It remains to show that $\tilde{\Omega}_{\bm k}^{s+1}\setminus A_{\bm k}^{s+1}$ is $s$-$\good$ and the set $A_{\bm k}^{s+1}-\bm k$ is independent of $\bm k\in P_{s+1}$ and symmetrical about the origin. Detailed arguments for these claims appear on page 26 of \cite{CSZ24a}. This completes the proof of $(\bm c)_{s+1}$ in the case $(\bm C1)_s$.

	We now proceed to establish properties $(\bm d)_{s+1}$ and $(\bm f)_{s+1}$ under the assumption of case $(\bm C1)_s$. For $\bm k\in Q_s^-$, we consider
	\begin{align*}
		\mathcal{M}_{s+1}(z):=(\mathcal{T}(z))_{\tilde{\Omega}_{\bm k}^{s+1}-\bm k}=((v(z+\bm n\cdot\bm\omega)-E)\delta_{\bm n,\bm n'}+\ep \mathcal{W})_{\tilde{\Omega}_{\bm k}^{s+1}-\bm k}
	\end{align*}
	defined in
	\begin{align}\label{zs1}
		\{z\in\C:\ |z-\theta_{s}|\le\delta_s^{\frac{1}{10}}\}.
	\end{align}
	If $\bm k'\in P_s$ and $\tilde{\Omega}_{\bm k'}^s\subset(\tilde{\Omega}_{\bm k}^{s+1}\setminus A_{\bm k}^{s+1})$, then $0\ne\|\bm k'-\bm k\|\le3N_{s+1}^{10}$. Thus from \eqref{indpa},
	\begin{align*}
		\|\theta+\bm k'\cdot\bm\omega-\theta_s\|_{\T}&\ge\|(\bm k-\bm k')\cdot\bm\omega\|_{\T}-\|\theta+\bm k\cdot\bm\omega-\theta_s\|_{\T}\\
		&\ge\frac{\g}{(3N_{s+1}^{10})^\tau}-\delta_{s}>\delta_s^{\frac{1}{10^4}}.
	\end{align*}
	By \eqref{dqs1+}, we have $\bm k'\notin\tilde{Q}_s^{+}$, and thus
	\begin{align*}
		\|\theta+\bm k'\cdot\bm\omega+\theta_s\|_{\T}>\delta_s^{\frac{1}{100}}.
	\end{align*}
	From $\tilde{\Omega}_{\bm k}^{s+1}\setminus A_{\bm k}^{s+1}$ is $s$-$\good$ (cf. $(\bm c)_{s+1}$) and \eqref{tsgnorm}, we obtain
	\begin{align}
		\nonumber\|\mathcal{T}_{\tilde{\Omega}_{\bm k}^{s+1}\setminus A_{\bm k}^{s+1}}^{-1}\|	&< 2\delta_{s-1}^{-3}\times\sup\limits_{\{\bm k'\in P_s:\ \tilde{\Omega}_{\bm k'}^s\subset(\tilde{\Omega}_{\bm k}^{s+1}\setminus A_{\bm k}^{s+1})\}}(\|\theta+\bm k'\cdot\bm\omega-\theta_s\|_{\T}^{-1}\cdot\|\theta+\bm k'\cdot\bm\omega+\theta_s\|_{\T}^{-1})\\
		\label{ts0}&<\frac{1}{2}\delta_s^{-2\times\frac{1}{100}}.
	\end{align}
	One may restate as
	\begin{align*}
		\|((\mathcal{M}_{s+1}(\theta+\bm k\cdot\bm\omega))_{(\tilde{\Omega}_{\bm k}^{s+1}\setminus A_{\bm k}^{s+1})-\bm k})^{-1}\|&<\frac{1}{2}\delta_s^{-2\times\frac{1}{100}}.
	\end{align*}
	Notice that
	\begin{align}\label{qs-}
		\nonumber\|z-(\theta+\bm k\cdot\bm\omega)\|_{\T}&\le|z-\theta_s|+\|\theta+\bm k\cdot\bm\omega-\theta_s\|_{\T}\\
		&<\delta_s^{\frac{1}{10}}+\delta_s<2\delta_s^{\frac{1}{10}}.
	\end{align}
	Thus by Neumann series argument, we can show
	\begin{align}\label{375}
		\|((\mathcal{M}_{s+1}(z))_{(\tilde{\Omega}_{\bm k}^{s+1}\setminus A_{\bm k}^{s+1})-\bm k})^{-1}\|&<\delta_s^{-2\times\frac{1}{100}}.
	\end{align}
	We now apply the Schur complement lemma to derive the required estimates. Lemma \eqref{scl} asserts that the inverse $(\mathcal{M}_{s+1}(z))^{-1}$ is governed by the Schur complement associated with $((\tilde{\Omega}_{\bm k}^{s+1}\setminus A_{\bm k}^{s+1})-\bm k)$:
	\begin{align*}
		\mathcal{S}_{s+1}(z)&=(\mathcal{M}_{s+1}(z))_{A_{\bm k}^{s+1}-\bm k}-\left(\vphantom{\left((\mathcal{M}_{s+1}(z))_{(\tilde{\Omega}_{\bm k}^{s+1}\setminus A_{\bm k}^{s+1})-\bm k}\right)^{-1}}\mathcal{R}_{A_{\bm k}^{s+1}-\bm k}\mathcal{M}_{s+1}(z)\mathcal{R}_{(\tilde{\Omega}_{\bm k}^{s+1}\setminus A_{\bm k}^{s+1})-\bm k}\right.\\
		&\ \  \left.\times\left((\mathcal{M}_{s+1}(z))_{(\tilde{\Omega}_{\bm k}^{s+1}\setminus A_{\bm k}^{s+1})-\bm k}\right)^{-1} \mathcal{R}_{(\tilde{\Omega}_{\bm k}^{s+1}\setminus A_{\bm k}^{s+1})-\bm k}\mathcal{M}_{s+1}(z)\mathcal{R}_{A_{\bm k}^{s+1}-\bm k}\right).
	\end{align*}
	Our primary objective is to analyze $\det \mathcal{S}_{s+1}(z)$. Observing that
	\begin{align*}
		A_{\bm k}^{s+1}-\bm k=A_{\bm k}^s-\bm k\subset\Omega_{\bm k}^s-\bm k,
	\end{align*}
	 we decompose the restriction operator as follows
	\begin{align*}
		&\ \mathcal{R}_{A_{\bm k}^{s+1}-\bm k}\mathcal{M}_{s+1}(z)\mathcal{R}_{(\tilde{\Omega}_{\bm k}^{s+1}\setminus A_{\bm k}^{s+1})-\bm k}\\
		=&\ \mathcal{R}_{A_{\bm k}^{s}-\bm k}\mathcal{M}_{s+1}(z)\mathcal{R}_{(\tilde{\Omega}_{\bm k}^{s+1}\setminus A_{\bm k}^{s})-\bm k}\\
		=&\ \mathcal{R}_{A_{\bm k}^{s}-\bm k}\mathcal{M}_{s+1}(z)\mathcal{R}_{(\tilde{\Omega}_{\bm k}^{s}\setminus A_{\bm k}^{s})-\bm k}+\mathcal{R}_{A_{\bm k}^{s}-\bm k}\mathcal{M}_{s+1}(z)\mathcal{R}_{(\tilde{\Omega}_{\bm k}^{s+1}\setminus \tilde{\Omega}_{\bm k}^{s})-\bm k}.
	\end{align*}
	Similarly, we have
	\begin{align*}
		\mathcal{R}_{(\tilde{\Omega}_{\bm k}^{s+1}\setminus A_{\bm k}^{s+1})-\bm k}\mathcal{M}_{s+1}(z)\mathcal{R}_{A_{\bm k}^{s+1}-\bm k}=&\ \ \mathcal{R}_{(\tilde{\Omega}_{\bm k}^{s}\setminus A_{\bm k}^{s})-\bm k}\mathcal{M}_{s+1}(z)\mathcal{R}_{A_{\bm k}^{s}-\bm k}\\
		&\ \ +\mathcal{R}_{(\tilde{\Omega}_{\bm k}^{s+1}\setminus \tilde{\Omega}_{\bm k}^{s})-\bm k}\mathcal{M}_{s+1}(z)\mathcal{R}_{A_{\bm k}^{s}-\bm k}.
	\end{align*}
	Then
	\begin{align}
	\nonumber	\mathcal{S}_{s+1}(z)&=(\mathcal{M}_{s+1}(z))_{A_{\bm k}^{s}-\bm k}+\boE_s(z)-\left(\vphantom{\left((\mathcal{M}_{s+1}(z))_{(\tilde{\Omega}_{\bm k}^{s+1}\setminus A_{\bm k}^{s+1})-\bm k}\right)^{-1}}\mathcal{R}_{A_{\bm k}^{s}-\bm k}\mathcal{M}_{s+1}(z)\mathcal{R}_{(\tilde{\Omega}_{\bm k}^{s}\setminus A_{\bm k}^{s})-\bm k}\right.\\
		\label{Ss+1Ss1}&\ \ \left.\times\left((\mathcal{M}_{s+1}(z))_{(\tilde{\Omega}_{\bm k}^{s+1}\setminus A_{\bm k}^{s+1})-\bm k}\right)^{-1} \mathcal{R}_{(\tilde{\Omega}_{\bm k}^{s}\setminus A_{\bm k}^{s})-\bm k}\mathcal{M}_{s+1}(z)\mathcal{R}_{A_{\bm k}^{s}-\bm k}\right),
	\end{align}
	where
	\begin{align*}
		\boE_s(z)&=\left(\vphantom{\left((\mathcal{M}_{s+1}(z))_{(\tilde{\Omega}_{\bm k}^{s+1}\setminus A_{\bm k}^{s+1})-\bm k}\right)^{-1}}\mathcal{R}_{A_{\bm k}^{s}-\bm k}\mathcal{M}_{s+1}(z)\mathcal{R}_{(\tilde{\Omega}_{\bm k}^{s+1}\setminus \tilde{\Omega}_{\bm k}^{s})-\bm k}\right.\\
		&\ \ \ \left.\times\left((\mathcal{M}_{s+1}(z))_{(\tilde{\Omega}_{\bm k}^{s+1}\setminus A_{\bm k}^{s+1})-\bm k}\right)^{-1}\mathcal{R}_{(\tilde{\Omega}_{\bm k}^{s+1}\setminus A_{\bm k}^{s})-\bm k}\mathcal{M}_{s+1}(z)\mathcal{R}_{A_{\bm k}^{s}-\bm k}\right)\\
		&\ \ +\left(\vphantom{\left((\mathcal{M}_{s+1}(z))_{(\tilde{\Omega}_{\bm k}^{s+1}\setminus A_{\bm k}^{s+1})-\bm k}\right)^{-1}}\mathcal{R}_{A_{\bm k}^{s}-\bm k}\mathcal{M}_{s+1}(z)\mathcal{R}_{(\tilde{\Omega}_{\bm k}^{s}\setminus A_{\bm k}^{s})-\bm k}\right.\\
		&\ \ \ \left.\times\left((\mathcal{M}_{s+1}(z))_{(\tilde{\Omega}_{\bm k}^{s+1}\setminus A_{\bm k}^{s+1})-\bm k}\right)^{-1}\mathcal{R}_{(\tilde{\Omega}_{\bm k}^{s+1}\setminus \tilde{\Omega}_{\bm k}^{s})-\bm k}\mathcal{M}_{s+1}(z)\mathcal{R}_{A_{\bm k}^{s}-\bm k}\right).
	\end{align*}
	From $\dist(A_{\bm k}^{s},\p \tilde{\Omega}_{\bm k}^s)>\frac{1}{2}\tilde{\zeta}_s$, \eqref{wphi}, \eqref{indpa}, \eqref{Det} and \eqref{375}, for $\bm x,\bm y\in A_{\bm k}^{s}-\bm k$, we get
	\begin{align}\label{E_sz}
		|\boE_s(z)(\bm x,\bm y)|\le 2D\left(\frac{\alpha}{2}\right)D(\alpha)\delta_s^{-\frac{1}{50}}e^{-\frac{\alpha}{2}\log^{\rho}\left(1+\frac{\tz_s}{2}\right)}\le \delta_s^{20}.
	\end{align}
	Since $\tilde{\Omega}_{\bm k}^s\setminus A_{\bm k}^s$ is $(s-1)$-$\good$ (cf. $(\bm c)_{s}$), \eqref{tsgnorm} and $\eqref{tsgdecay}$, we have
	\begin{align*}
		\|\mathcal{T}_{\tilde{\Omega}_{\bm k}^s\setminus A_{\bm k}^s}^{-1}\|\le \delta_{s-1}^{-3}.
	\end{align*}
	Equivalently,
	\begin{align}
		\label{368}\|((\mathcal{M}_{s+1}(\theta+\bm k\cdot\bm\omega))_{(\tilde{\Omega}_{\bm k}^s\setminus A_{\bm k}^s)-\bm k})^{-1}\|\le \delta_{s-1}^{-3}.
	\end{align}
	Within the domain specified by \eqref{zs1}, we claim the following estimates:
	\begin{align}
		\label{369}\|((\mathcal{M}_{s+1}(z))_{(\tilde{\Omega}_{\bm k}^s\setminus A_{\bm k}^s)-\bm k})^{-1}\|&\le2 \delta_{s-1}^{-3},\\
		\label{370}|((\mathcal{M}_{s+1}(z))_{(\tilde{\Omega}_{\bm k}^s\setminus A_{\bm k}^s)-\bm k})^{-1}(\bm x,\bm y)|&\le e^{-\alpha_{s-1}\log^{\rho}(1+\|\bm x-\bm y\|)}\ \text{for $\|\bm x-\bm y\|>10\tz_{s-1}$.}
	\end{align}
	\begin{proof}[Proof of the Claim]
		We proceed via the following steps:
		\begin{itemize}
			\item[\textbf{Step 1}]: \textbf{Operator Decomposition}.  Define the comparison operators:
		
		\begin{align*}
			\mathcal{T}_1=(\mathcal{M}_{s+1}(\theta+\bm k\cdot\bm\omega))_{(\tilde{\Omega}_{\bm k}^s \setminus A_{\bm k}^s)-\bm k},\ \mathcal{T}_2=(\mathcal{M}_{s+1}(z))_{(\tilde{\Omega}_{\bm k}^s \setminus A_{\bm k}^s)-\bm k}.
		\end{align*}
		The difference $\mathcal{D}_1=\mathcal{T}_1-\mathcal{T}_2$ is diagonal, with norm bounded by
		\begin{align*}
			\|\mathcal{D}_1\|\le 2\kappa_2\sqrt{4R^2+\frac{1}{4}}\delta_s^{\frac{1}{10}}
		\end{align*}
		 via \eqref{vdefn} and \eqref{qs-}. 
		 	\item[\textbf{Step 2}]: \textbf{Neumann Series Argument}.
		 Applying the Neumann expansion yields
		\begin{align*}
			\mathcal{T}_2^{-1}=(\mathcal{I}_{(\tilde{\Omega}_{\bm k}^s \setminus A_{\bm k}^s)-\bm k}-\mathcal{T}_1^{-1}\mathcal{D}_1)^{-1}\mathcal{T}_1^{-1}=\sum_{i=0}^{\infty}(\mathcal{T}_1^{-1}\mathcal{D}_1)^i \mathcal{T}_1^{-1}.
		\end{align*}
		The norm estimate follows by combining \eqref{368} with the Neumann series expansion. Specifically, from \eqref{368} we have the a priori bound $\|\mathcal{T}_1^{-1}\|\le \delta_{s-1}^{-3}$, which when combined with the smallness condition on $\|\mathcal{D}_1\|$ yields
		\begin{align*}
			\|	\mathcal{T}_2^{-1}\|\le \left(\sum_{i=0}^{\infty}\|\mathcal{T}_1^{-1}\|^i\cdot\|\mathcal{D}_1\|^i\right) \|\mathcal{T}_1^{-1}\|\le 2\delta_{s-1}^{-3}.
		\end{align*}
		\item[\textbf{Step 3}]: \textbf{Generalization to Subdomains}.
		For any subdomain $\tO_{\bm k'}^{t}\subset (\tilde{\Omega}_{\bm k}^s \setminus A_{\bm k}^s)$ with $0\le t\le s-1$ and $\bm k'\in P_t\setminus Q_t$, we similarly obtain
		\begin{align*}
			\|\boT_{\tO_{\bm k'}^{t}}^{-1}(z)\|\le 2\|\boT_{\tO_{\bm k'}^{t}}^{-1}\|\le 2\delta_{t}^{-3}.
		\end{align*}
		\item[\textbf{Step 4}]: \textbf{Off-Diagonal Decay}.
The decay estimate \eqref{370} follows by iterating the resolvent identity, following the same methodology used to establish $(\bm{e})_{s-1}$. We omit the repetitive details.
\end{itemize}
We complete proof of the claim.
	\end{proof}
	
	To estimate the difference between $\mathcal{S}_s$ and $\mathcal{S}_{s+1}$, we employ the resolvent identity combined with the decay properties of $\mathcal W$. For clarity, we first define the following sets: $X=(\tilde{\Omega}_{\bm k}^s\setminus A_{\bm k}^s)-\bm k$, $Z_1=\lg_{\frac{\tilde{\zeta}_s}{4}}\cap X$, $Z_2=\lg_{\frac{\tilde{\zeta}_s}{8}}\cap X$ and $Y=(\tilde{\Omega}_{\bm k}^{s+1}\setminus A_{\bm k}^{s+1})-\bm k$. For any $\bm m\in X$ and $\bm n\in Y$, the resolvent identity yields
	\begin{align*}
		&((\mathcal{M}_{s+1}(z))_Y)^{-1}(\bm m,\bm n)-\chi_X(\bm n)((\mathcal{M}_{s+1}(z))_X)^{-1}(\bm m,\bm n)\\
		=&-\ep\sum_{\bm l\in X \atop \bm l'\in Y\setminus X}((\mathcal{M}_{s+1}(z))_X)^{-1}(\bm m,\bm l)\mathcal{W}(\bm l,\bm l')((\mathcal{M}_{s+1}(z))_Y)^{-1}(\bm l',\bm n).
	\end{align*}
	If $\bm m\in Z_2$, since \eqref{wphi}, \eqref{quaeq}, \eqref{Det}, \eqref{375}, \eqref{369}, \eqref{370}, $\dist(Z_1,Y\setminus X)\ge\frac{\tilde{\zeta}_s}{8}\gg10\tz_{s-1}$ and $\dist(Z_2,X\setminus Z_1)\ge\frac{\tilde{\zeta}_s}{8}\gg10\tz_{s-1}$, we can get
	\begin{align*}
		({\rm I})=&\ |((\mathcal{M}_{s+1}(z))_Y)(\bm m,\bm n)-\chi_X(\bm n)((\mathcal{M}_{s+1}(z))_X)^{-1}(\bm m,\bm n)|\\
		\le&\ \sum_{\bm l\in Z_1 \atop \bm l'\in Y\setminus X}|((\mathcal{M}_{s+1}(z))_X)^{-1}(\bm m,\bm l)|\cdot|\mathcal{W}(\bm l,\bm l')|\cdot|((\mathcal{M}_{s+1}(z))_Y)^{-1}(\bm l',\bm n)|\\
		&\ \ +\sum_{\bm l\in X\setminus Z_1 \atop\bm l'\in Y\setminus X}|((\mathcal{M}_{s+1}(z))_X)^{-1}(\bm m,\bm l)|\cdot|\mathcal{W}(\bm l,\bm l')|\cdot|((\mathcal{M}_{s+1}(z))_Y)^{-1}(\bm l',\bm n)|\\
		\le&\  2\delta_{s-1}^{-3}\delta_{s}^{-\frac{1}{50}}e^{-\frac{9}{10}\alpha\log^{\rho}\left(1+\frac{\tz_s}{8}\right)}(\#(Z_1))D\left(\frac{\alpha}{10}\right)\\
		&\ \  +\delta_{s}^{-\frac{1}{50}}e^{-\alpha_{s-1}\log^{\rho}\left(1+\frac{\tz_s}{8}\right)+\alpha_{s-1}C(\rho)\log 2}(\#(X\setminus Z_1))D\left(\frac{\alpha}{10}\right)\\
		<&\ \delta_s^{20}.
	\end{align*}
	If $\bm m\in X\setminus Z_2$, by \eqref{wphi}, \eqref{Det}, \eqref{375} and \eqref{369}, we obtain
	\begin{align*}
		({\rm II})=&\ |((\mathcal{M}_{s+1}(z))_Y)^{-1}(\bm m,\bm n)-\chi_X(\bm n)((\mathcal{M}_{s+1}(z))_X)^{-1}(\bm m,\bm n)|\\
		\le&\ \sum_{\bm l\in X \atop \bm l'\in Y\setminus X}|((\mathcal{M}_{s+1}(z))_X)^{-1}(\bm m,\bm l)|\cdot|\mathcal{W}(\bm l,\bm l')|\cdot|((\mathcal{M}_{s+1}(z))_Y)^{-1}(\bm l',\bm n)|\\
		\le&\ 2\delta_{s-1}^{-3}\delta_{s}^{-\frac{1}{50}}(\#X)D(\alpha).
	\end{align*}
	For $\bm i\in A_{\bm k}^s-\bm k$, $\bm n\in Y$, since \eqref{wphi}, \eqref{Det}, $\dist(A_{\bm k}^s-\bm k,X\setminus Z_2)\ge\frac{\tilde{\zeta}_s}{16}\gg10\tz_{s-1}$ and \eqref{375}--\eqref{370},  we have
	\begin{align*}
		&\ \ |(\mathcal{R}_{A_{\bm k}^s-\bm k}\mathcal{M}_{s+1}(z)\mathcal{R}_X ((\mathcal{M}_{s+1}(z))_Y)^{-1})(\bm i,\bm n)\\
		&\ \ \ \ -(\mathcal{R}_{A_{\bm k}^s-\bm k}\mathcal{M}_{s+1}(z)\mathcal{R}_X ((\mathcal{M}_{s+1}(z))_X)^{-1}\mathcal{R}_X)(\bm i,\bm n)|\\
		\le&\ \ \sum_{\bm m\in Z_2}|\mathcal{W}(\bm i,\bm m)|\cdot({\rm I})+\sum_{\bm m\in X\setminus Z_2}|\mathcal{W}(\bm i,\bm m)|\cdot({\rm II})\\
		\le &\ \ D(\alpha)\delta_s^{20}+2\delta_{s-1}^{-3}\delta_{s}^{-\frac{1}{50}}(\#X)D(\alpha)D\left(\frac{\alpha}{10}\right)e^{-\frac{9}{10}\alpha\log^{\rho}\left(1+\frac{\tz_s}{16}\right)}<\delta_{s}^{15}.
	\end{align*}
	It then follows that
	\begin{align}
	\nonumber	&\ \ \mathcal{R}_{A_{\bm k}^s-\bm k}\mathcal{M}_{s+1}(z)\mathcal{R}_X ((\mathcal{M}_{s+1}(z))_Y)^{-1}\\
	\label{RXY}=&\ \ \mathcal{R}_{A_{\bm k}^s-\bm k}\mathcal{M}_{s+1}(z)\mathcal{R}_X ((\mathcal{M}_{s+1}(z))_X)^{-1}\mathcal{R}_X+O(\delta_s^{15}).
	\end{align}
	From \eqref{Ss+1Ss1}, \eqref{E_sz} and \eqref{RXY}, we have
	\begin{align*}
		&\ \mathcal{R}_{A_{\bm k}^s-\bm k}\mathcal{M}_{s+1}(z)\mathcal{R}_X ((\mathcal{M}_{s+1}(z))_Y)^{-1}\mathcal{R}_X \mathcal{M}_{s+1}(z)\mathcal{R}_{A_{\bm k}^s-\bm k}\\
		=&\ \mathcal{R}_{A_{\bm k}^s-\bm k}\mathcal{M}_{s+1}(z)\mathcal{R}_X ((\mathcal{M}_{s+1}(z))_X)^{-1}\mathcal{R}_X \mathcal{M}_{s+1}(z)\mathcal{R}_{A_{\bm k}^s-\bm k}+O(\delta_s^{15})\\
		=&\ \mathcal{R}_{A_{\bm k}^s-\bm k}\mathcal{M}_{s}(z)\mathcal{R}_X ((\mathcal{M}_{s}(z))_X)^{-1}\mathcal{R}_X \mathcal{M}_{s}(z)\mathcal{R}_{A_{\bm k}^s-\bm k}+O(\delta_s^{15}).
	\end{align*}
	and
	\begin{align*}
		\mathcal{S}_{s+1}(z)=&\ \mathcal{M}_{s}(z)_{A_{\bm k}^{s}-\bm k}-\left(\vphantom{(\mathcal{M}_{s}(z))_{X})^{-1}}\mathcal{R}_{A_{\bm k}^{s}-\bm k}\mathcal{M}_{s}(z)\mathcal{R}_{X}((\mathcal{M}_{s}(z))_{X})^{-1}\right.\\
		&\ \left.\times\ \mathcal{R}_{X}\mathcal{M}_{s}(z)\mathcal{R}_{A_{\bm k}^{s}-\bm k}\right)+O(\delta_s^{15})\\
		=&\ \mathcal{S}_s(z)+O(\delta_s^{15}),
	\end{align*}
	which implies \eqref{ss} for the $(s+1)$-th step. Building upon the estimates from \eqref{zs}, \eqref{detss} and \eqref{zs1}, we first note the lower bound
	\begin{align*}
		|\det \mathcal{S}_s(z)|\ge\delta_{s-1}\|z-\theta_s\|_{\T}\cdot\|z+\theta_s\|_{\T}.
	\end{align*}
	Applying Lemma \ref{det1} with the cardinality bound $\#(A_{\bm k}^s-\bm k)\le 2^s$ and using \eqref{ss}, we derive
	\begin{align*}
		\det \mathcal{S}_{s+1}(z)&=\det \mathcal{S}_s(z)+O((2^s)^2 (4|v|_R)^{2^s}\delta_s^{15})\\
		&=\det \mathcal{S}_s(z)+O(\delta_s^{10}).
	\end{align*}
	For the torus norm, we establish
	\begin{align*}
		\|z+\theta_s\|_{\T}&\ge\|\theta+\bm k\cdot\bm\omega+\theta_s\|_{\T}-\|z-\theta_s\|_{\T}-\|\theta+\bm k\cdot\bm\omega-\theta_s\|_{\T}\\
		&>\delta_s^{\frac{1}{100}}-\delta_{s}^{\frac{1}{10}}-\delta_s\\
		&>\frac{1}{2}\delta_s^{\frac{1}{100}}.
	\end{align*}
	This leads to the improved estimate
	\begin{align*}
		|\det \mathcal{S}_{s+1}(z)|\ge\delta_{s}^{\frac{1}{10}}|z-\theta_s+r_{s+1}(z)|,
	\end{align*}
	where $r_{s+1}(z)$ is analytic in the region \eqref{zs1} with $|r_{s+1}(z)| < \delta_s^8$. Finally, the equation
	\begin{align*}
		(z-\theta_s)+r_{s+1}(z)=0
	\end{align*}
	has a unique root $\theta_{s+1}$ in \eqref{zs1} satisfying
	\begin{align}\label{ts+1-ts}
		|\theta_{s+1}-\theta_s|=|r_{s+1}(\theta_{s+1})|<\delta_s^{8}.
	\end{align}
	For $|z| = \delta_s^{1/10}$, we obtain (since $|r_1(z)|<\delta_s^8$ and \eqref{ts+1-ts})
	\begin{align*}
		\frac{|r_{s+1}(\theta_1)-r_{s+1}(z)|}{|z-\theta_s+r_{s+1}(z)|}\le 4\delta_{s}^{7},
	\end{align*}
	which implies
	\begin{align*}
		\frac{|z-\theta_{s+1}|}{|z-\theta_{s}+r_{s+1}(z)|}=\frac{|z-\theta_{s}+r_{s+1}(\theta_1)|}{|z-\theta_{s}+r_{s+1}(z)|}\in[1-4\delta_{s}^{7},1+4\delta_{s}^{7}].
	\end{align*}
	By the maximum modulus principle, we have
	\begin{align*}
		\frac{1}{2}\le\frac{|z-\theta_{s+1}|}{|z-\theta_{s}+r_{s+1}(z)|}\le 2.
	\end{align*}
	Moreover, the root $\theta_{s+1}$ is unique for $\det \mathcal{M}_{s+1}(z) = 0$ in \eqref{zs1}. Since $\|z+\theta_s\|_{\T}>\frac{1}{2}\delta_s^{\frac{1}{100}}$ and $|\theta_{s+1}-\theta_s|<\delta_s^{8}$, we have
	\begin{align*}
		\frac{1}{2}\|z+\theta_s\|_{\T}\le \|z+\theta_{s+1}\|_{\T}\le 2\|z+\theta_s\|_{\T}.
	\end{align*}
	Thus, for $z$ in \eqref{zs1}, we establish the key determinant bound
	\begin{align}\label{detss+11}
		|\det \mathcal{S}_{s+1}(z)|\ge\delta_{s}\|z-\theta_{s+1}\|_{\T}\cdot\|z+\theta_{s+1}\|_{\T}.
	\end{align}
	Since $\delta_{s+1}=\delta_s^{10^{5\rho'}}$, we obtain $\delta_{s+1}^{\frac{1}{10^4}}\ll\frac{1}{2}\delta_s^{\frac{1}{10}}$. Recalling \eqref{zs1} and \eqref{ts+1-ts}, the estimate \eqref{detss+11} remains valid for
	\begin{align*}
		\|z-\theta_{s+1}\|_{\T}<\delta_{s+1}^{\frac{1}{10^4}}.
	\end{align*}
	
	For $\bm{k} \in Q_s^+$, we analyze $\mathcal{M}_{s+1}(z)$ in the region
	\begin{align}\label{zs2}
		\left\{z\in\C:\ |z+\theta_s|\le \delta_{s}^{\frac{1}{10}}\right\}.
	\end{align}
	
	An analogous argument demonstrates that the equation $\det \mathcal{M}_{s+1}(z) = 0$ possesses a unique root $\theta_{s+1}^{'}$ within the region defined by \eqref{zs2}. We now establish the antisymmetry property $\theta_{s+1} + \theta_{s+1}' = 0$. In fact, Lemma \ref{ef} guarantees that $\det \mathcal{M}_{s+1}(z)$ is an even function of $z$. The uniqueness of roots in \eqref{zs2} immediately implies $\theta_{s+1}' = -\theta_{s+1}$. Consequently, for all $z$ in \eqref{zs2}, the determinant bound \eqref{detss+11} remains valid. This establishes the estimate throughout the symmetric domain
		\begin{align*}
		\left\{z\in\C:\min_{\sigma=\pm1}\|z+\sigma\theta_{s+1}\|_{\T}<\delta_{s+1}^{\frac{1}{10^4}}\right\},
	\end{align*}
	thereby completing the proof of \eqref{detss} for the $(s+1)$-th step. Combining \eqref{qspm}--\eqref{ps+1} and the following
	\begin{align*}
		\|\theta+\bm k\cdot\bm\omega\pm\theta_{s+1}\|_{\T}<10\delta_{s+1}^{\frac{1}{100}},\ |\theta_{s+1}-\theta_s|<\delta_s^{8}\Rightarrow\|\theta+\bm k\cdot\bm\omega\pm\theta_s\|_{\T}<\delta_s,
	\end{align*}
	we get
	\begin{align*}
		\left\{\bm k\in\Z^d+\frac{1}{2}\sum_{i=0}^{s}\bm l_i:\ \min_{\sigma=\pm1}\|\theta+\bm k\cdot\bm\omega+\sigma\theta_{s+1}\|_{\T}<10\delta_{s+1}^{\frac{1}{100}}\right\}\subset P_{s+1},
	\end{align*}
	which proves \eqref{fs} at the $(s+1)$-th step. Finally, we want to estimate $\|\mathcal{T}_{\tilde{\Omega}_{\bm k}^{s+1}}^{-1}\|$. For $\bm k\in P_{s+1}$, by \eqref{ps+11}, we obtain
	\begin{align*}
		\theta+\bm k\cdot\bm\omega\in\left\{z\in\C:\ \min_{\sigma=\pm1}\|z+\sigma\theta_s\|_{\T}<\delta_s^{\frac{1}{10}}\right\},
	\end{align*}
	which together with \eqref{detss+11} implies
	\begin{align*}
		&\ |\det(\mathcal{T}_{A_{\bm k}^{s+1}}-\mathcal{R}_{A_{\bm k}^{s+1}}\mathcal{T}\mathcal{R}_{\tilde{\Omega}_{\bm k}^{s+1}\setminus A_{\bm k}^{s+1}}\mathcal{T}_{\tilde{\Omega}_{\bm k}^{s+1}\setminus A_{\bm k}^{s+1}}^{-1}\mathcal{R}_{\tilde{\Omega}_{\bm k}^{s+1}\setminus A_{\bm k}^{s+1}}\mathcal{T}\mathcal{R}_{A_{\bm k}^{s+1}})|\\
		=&\ |\det \mathcal{S}_{s+1}(\theta+\bm k\cdot\bm\omega)|\\
		\ge&\ \delta_{s}\|\theta+\bm k\cdot\bm\omega-\theta_{s+1}\|_{\T}\cdot\|\theta+\bm k\cdot\bm\omega+\theta_{s+1}\|_{\T}.
	\end{align*}
	By \eqref{ss}, Cramer's rule and Lemma \ref{chi}, one has
	\begin{align*}
		&\|(\mathcal{T}_{A_{\bm k}^{s+1}}-\mathcal{R}_{A_{\bm k}^{s+1}}\mathcal{T}\mathcal{R}_{\tilde{\Omega}_{\bm k}^{s+1}\setminus A_{\bm k}^{s+1}}\mathcal{T}_{\tilde{\Omega}_{\bm k}^{s+1}\setminus A_{\bm k}^{s+1}}^{-1}\mathcal{R}_{\tilde{\Omega}_{\bm k}^{s+1}\setminus A_{\bm k}^{s+1}}\mathcal{T}\mathcal{R}_{A_{\bm k}^{s+1}})^{-1}\|\\
		=&\ |\det \mathcal{S}_{s+1}(\theta+\bm k\cdot\bm\omega)|^{-1}\\
		&\ \ \cdot\|(\mathcal{T}_{A_{\bm k}^{s+1}}-\mathcal{R}_{A_{\bm k}^{s+1}}\mathcal{T}\mathcal{R}_{\tilde{\Omega}_{\bm k}^{s+1}\setminus A_{\bm k}^{s+1}}\mathcal{T}_{\tilde{\Omega}_{\bm k}^{s+1}\setminus A_{\bm k}^{s+1}}^{-1}\mathcal{R}_{\tilde{\Omega}_{\bm k}^{s+1}\setminus A_{\bm k}^{s+1}}\mathcal{T}\mathcal{R}_{A_{\bm k}^{s+1}})^{*}\|\\
		\le&\ 2^s\cdot (4|v|_R)^{2^s}\delta_{s}^{-1}\|\theta+\bm k\cdot\bm\omega-\theta_{s+1}\|_{\T}^{-1}\cdot\|\theta+\bm k\cdot\bm\omega+\theta_{s+1}\|_{\T}^{-1}.
	\end{align*}
	From Lemma \ref{scl} and \eqref{ts0}, we get
	\begin{align}
		\nonumber\|\mathcal{T}_{\tilde{\Omega}_{\bm k}^{s+1}}^{-1}\|&<4(1+\|\mathcal{T}_{\tilde{\Omega}_{\bm k}^{s+1}\setminus A_{\bm k}^{s+1}}\|)^2\\
		\nonumber&\ \ \times(1+\|(\mathcal{T}_{A_{\bm k}^{s+1}}-\mathcal{R}_{A_{\bm k}^{s+1}}\mathcal{T}\mathcal{R}_{\tilde{\Omega}_{\bm k}^{s+1}\setminus A_{\bm k}^{s+1}}\mathcal{T}_{\tilde{\Omega}_{\bm k}^{s+1}\setminus A_{\bm k}^{s+1}}^{-1}\mathcal{R}_{\tilde{\Omega}_{\bm k}^{s+1}\setminus A_{\bm k}^{s+1}}\mathcal{T}\mathcal{R}_{A_{\bm k}^{s+1}})^{-1}\|)\\
		\label{tos+1-11}&<\delta_{s}^{-2}\|\theta+\bm k\cdot\bm\omega-\theta_{s+1}\|_{\T}^{-1}\cdot\|\theta+\bm k\cdot\bm\omega+\theta_{s+1}\|^{-1}_{\T}.
	\end{align}
	
	\begin{itemize}
		\item[\textbf{Case 2.}] The case $(\bm C2)_s$ occurs, i.e.,
	\end{itemize}
	\begin{align*}
		\dist(\tilde{Q}_s^-,Q_s^+)\le100N_{s+1}^{10}.
	\end{align*}
	Then there exist $\bm i_s\in Q_{s}^{+}$ and $\bm j_s\in \tilde{Q}_s^-$ with $\|\bm i_s-\bm j_s\|\le100 N_{s+1}^{10}$ such that
	\begin{align*}
		\|\theta+\bm i_s\cdot\bm\omega+\theta_s\|_{\T}<\delta_s,\ \|\theta+\bm j_s\cdot\bm\omega-\theta_s\|_{\T}<\delta_s^{\frac{1}{100}}.
	\end{align*}
	Define
	\begin{align*}
		\bm l_s=\bm i_s-\bm j_s.
	\end{align*}
	Using \eqref{as1} and \eqref{as2} yields
	\begin{align*}
		Q_s^+,\tilde{Q}_s^-\subset P_s\subset \Z^d+\frac{1}{2}\sum_{i=0}^{s-1}\bm l_i.
	\end{align*}
	Thus $\bm i_s\equiv \bm j_s\ (\text{mod}\ \Z^d)$ and $\bm l_s\in\Z^d$. Define
	\begin{align}\label{os+12}
		O_{s+1}=Q_s^-\cup(Q_s^+-\bm l_s).
	\end{align}
	For every $\bm o\in O_{s+1}$, define its mirror point as
	\begin{align*}
		\bm o^*=\bm o+\bm l_s.
	\end{align*}
	Then we have
	\begin{align*}
		O_{s+1}\subset\left\{\bm o\in\Z^d+\frac{1}{2}\sum_{i=0}^{s-1}\bm l_i:\ \|\theta+\bm o\cdot\bm\omega-\theta_s\|_{\T}<2\delta_s^{\frac{1}{100}}\right\}
	\end{align*}
	and
	\begin{align*}
		O_{s+1}+\bm l_s\subset\left\{\bm o^*\in\Z^d+\frac{1}{2}\sum_{i=0}^{s-1}\bm l_i:\ \|\theta+\bm o^*\cdot\bm\omega+\theta_s\|_{\T}<2\delta_s^{\frac{1}{100}}\right\}.
	\end{align*}
	Then by \eqref{fs}, we obtain
	\begin{align}\label{os+13}
		O_{s+1}\cup(O_{s+1}+\bm l_s)\subset P_s.
	\end{align}
	Next, define
	\begin{align}\label{ps+12}
		P_{s+1}=\left\{\frac{1}{2}(\bm o+\bm o^*):\ \bm o\in O_{s+1}\right\}=\left\{\bm o+\frac{\bm l_s}{2}:\ \bm o\in O_{s+1}\right\}.
	\end{align}
	Notice that
	\begin{align*}
		\min&\left(\left\|\frac{\bm l_s}{2}\cdot\bm\omega+\theta_s\right\|_{\T},\left\|\frac{\bm l_s}{2}\cdot\bm\omega+\theta_s-\frac{1}{2}\right\|_{\T}\right)\\
		&=\frac{1}{2}\|\bm l_s\cdot\bm\omega+2\theta_s\|_{\T}\\
		&\le\frac{1}{2}(\|\theta+\bm i_s\cdot\bm\omega+\theta_s\|_{\T}+\|\theta+\bm j_s\cdot\bm\omega-\theta_s\|_{\T})<\delta_s^{\frac{1}{100}}.
	\end{align*}
	Since $\delta_s\ll1$, only one of
	\begin{align*}
		\left\|\frac{\bm l_s}{2}\cdot\bm\omega+\theta_s\right\|_{\T}<\delta_0^{\frac{1}{100}}\ \text{and}\ \left\|\frac{\bm l_s}{2}\cdot\bm\omega+\theta_s-\frac{1}{2}\right\|_{\T}<\delta_0^{\frac{1}{100}}
	\end{align*}
	occurs. First, we consider the case of
	\begin{align}\label{ls1}
		\left\|\frac{\bm l_s}{2}\cdot\bm\omega+\theta_s\right\|_{\T}<\delta_s^{\frac{1}{100}}.
	\end{align}
	Let $\bm k\in P_{s+1}$. Since $\bm k=\bm o+\frac{\bm l_s}{2}$ (for some $\bm o\in O_{s+1}$), we have
	\begin{align}\label{t+kos}
		\|\theta+\bm k\cdot\bm\omega\|_{\T}\le\|\theta+\bm o\cdot\bm\omega-\theta_s\|_{\T}+\left\|\frac{\bm l_s}{2}\cdot\bm\omega+\theta_s\right\|_{\T}<3\delta_s^{\frac{1}{100}},
	\end{align}
	which implies
	\begin{align}\label{ps+1s1}
		P_{s+1}\subset\left\{\bm k\in\Z^d+\frac{1}{2}\sum_{i=0}^s \bm l_i:\ \|\theta+\bm k\cdot\bm \omega\|_{\T}<3\delta_s^{\frac{1}{100}}\right\}.
	\end{align}
	Moreover, if $\bm k\ne \bm k'\in P_{s+1}$, we obtain
	\begin{align*}
		\|\bm k-\bm k'\|\ge\left(\frac{\g}{6\delta_s^{\frac{1}{100}}}\right)^{\frac{1}{\tau}}\gg 100N_{s+1}^{100}.
	\end{align*}
	Similar to the proof that appears in {\bf Case 1}  (i.e., the $(\bm C1)_s$ case), we can associate each $\bm k\in P_{s+1}$ with the blocks $\Omega_{\bm k}^{s+1}$, $2\Omega_{\bm k}^{s+1}$ and $\tilde{\Omega}_{\bm k}^{s+1}$ satisfying
	\begin{align*}
		&\lg_{100N_{s+1}^{10}}(\bm k)\subset \Omega_{\bm k}^{s+1}\subset\lg_{100N_{s+1}^{10}+50N_s^{100}}(\bm k),\\
		&\lg_{200N_{s+1}^{10}}(\bm k)\subset2 \Omega_{\bm k}^{s+1}\subset\lg_{200N_{s+1}^{10}+50N_s^{100}}(\bm k),\\
		&\lg_{N_{s+1}^{100}}(\bm k)\subset \tilde{\Omega}_{\bm k}^{s+1}\subset\lg_{N_{s+1}^{100}+50N_s^{100}}(\bm k)
	\end{align*}
	and
	\begin{align}\label{sb3}
		\left\{\begin{array}{l}
			\Omega_{\bm k}^{s+1}\cap\tilde{\Omega}_{\bm k'}^{s'}\ne\emptyset\ (s'<s+1)\Rightarrow\tilde{\Omega}_{\bm k'}^{s'}\subset\Omega_{\bm k}^{s+1},\\
				2\Omega_{\bm k}^{s+1}\cap\tilde{\Omega}_{\bm k'}^{s'}\ne\emptyset\ (s'<s+1)\Rightarrow\tilde{\Omega}_{\bm k'}^{s'}\subset2\Omega_{\bm k}^{s+1},\\
			\tilde{\Omega}_{\bm k}^{s+1}\cap\tilde{\Omega}_{\bm k'}^{s'}\ne\emptyset\ (s'<s+1)\Rightarrow\tilde{\Omega}_{\bm k'}^{s'}\subset\tilde{\Omega}_{\bm k}^{s+1},\\
			\dist(\tilde{\Omega}_{\bm k}^{s+1},\tilde{\Omega}_{\bm k'}^{s+1})>10\diam\tilde{\zeta}_{s+1}\ \text{for}\ \bm k\ne \bm k'\in P_{s+1}.
		\end{array}\right.
	\end{align}
	Furthermore, the translated domain
	\begin{align*}
		\tilde{\Omega}_{\bm k}^{s+1}-\bm k\subset\Z^d+\frac{1}{2}\sum_{i=0}^s \bm l_i
	\end{align*}
	is independent of $\bm k\in P_{s+1}$ and symmetrical about the origin. This completes the verification of properties $(\bm a)_{s+1}$ and $(\bm b)_{s+1}$ for case $(\bm C2)_{s}$.
	
	For each $\bm{k} \in P_{s+1}$, let $\bm{o}, \bm{o}^* \in P_s$ be the unique pair specified by \eqref{os+13}. We define
	\begin{align*}
		A_{\bm k}^{s+1}=A_{\bm o}^s\cup A_{\bm o^*}^s,
	\end{align*}
	where $\bm{o} \in O_{s+1}$ and $\bm{k} = \frac{1}{2}(\bm{o} + \bm{o}^*)$ per \eqref{ps+12}. This construction satisfies:
	\begin{align*}
		A_{\bm k}^{s+1}&\subset \Omega_{\bm o}^s\cup \Omega_{\bm o^*}^s\subset \Omega_{\bm k}^{s+1},\\
		\# A_{\bm k}^{s+1}&=\# A_{\bm o}^s+\# A_{\bm o^*}^s\le 2^{s+1}.
	\end{align*}
	The proof that $\tilde{\Omega}_{\bm k}^{s+1}\setminus A_{\bm k}^{s+1}$ is $s$-$\good$ appears on page 32 of \cite{CSZ24a}, which concludes the verification of $(\bm{c})_{s+1}$ for case $(\bm{C}2)_s$.
	
	We analyze the operator
	\begin{align*}
		\mathcal{M}_{s+1}(z):=\mathcal{T}_{\tilde{\Omega}_{\bm k}^{s+1}-\bm k}(z)=((v(z+\bm n\cdot\bm\omega)-E)\delta_{\bm n,\bm n'}+\ep \mathcal{W})_{\bm n\in\tilde{\Omega}_{\bm k}^{s+1}-\bm k}
	\end{align*}
	defined on the domain
	\begin{align}\label{zs3}
		\left\{z\in\C:\ |z|\le\delta_{s}^{\frac{1}{10^3}}\right\}.
	\end{align}
	If $\bm k'\in P_s$ and $\tilde{\Omega}_{\bm k'}^{s}\subset(\tilde{\Omega}_{\bm k}^{s+1}\setminus A_{\bm k}^{s+1})$, then $\bm k'\ne \bm o,\bm o^*$ and $\|\bm k'-\bm o\|,\|\bm k'-\bm o^*\|\le 4N_{s+1}^{100}$. Thus
	\begin{align*}
		\|\theta+\bm k'\cdot\bm\omega-\theta_s\|_{\T}&\ge\|(\bm k'-\bm o)\cdot\bm\omega\|_{\T}-\|\theta+\bm o\cdot\bm\omega-\theta_s\|_{\T}\\
		&\ge\frac{\g}{(4N_{s+1}^{100})^\tau}-2\delta_s^{\frac{1}{100}}>\delta_s^{\frac{1}{10^4}}
	\end{align*}
	and
	\begin{align*}
		\|\theta+\bm k'\cdot\bm\omega+\theta_s\|_{\T}&\ge\|(\bm k'-\bm o^*)\cdot\bm\omega\|_{\T}-\|\theta+\bm o^*\cdot\bm\omega+\theta_s\|_{\T}\\
		&\ge\frac{\g}{(4N_{s+1}^{100})^\tau}-2\delta_s^{\frac{1}{100}}>\delta_s^{\frac{1}{10^4}}.
	\end{align*}
	By \eqref{tsgnorm} and $\tilde{\Omega}_{\bm k}^{s+1}\setminus A_{\bm k}^{s+1}$ is $s$-$\good$ (cf. $(\bm c)_{s+1}$), we have
	\begin{align}
		\nonumber\|\mathcal{T}_{\tilde{\Omega}_{\bm k}^{s+1}\setminus A_{\bm k}^{s+1}}^{-1}\|&\le2\delta_{s}^{-3}\times\sup\limits_{\{\bm k'\in P_s:\ \tilde{\Omega}_{\bm k'}^s\subset(\tilde{\Omega}_{\bm k}^{s+1}\setminus A_{\bm k}^{s+1})\}}(\|\theta+\bm k'\cdot\bm\omega-\theta_s\|_{\T}^{-1}\cdot\|\theta+\bm k'\cdot\bm\omega+\theta_s\|_{\T}^{-1})\\
		\label{tosa2}&<\frac{1}{2}\delta_{s}^{-3\times\frac{1}{10^4}}.
	\end{align}
	One may restate \eqref{tosa2} as
	\begin{align*}
		\|((\mathcal{M}_{s+1}(\theta+\bm k\cdot\bm\omega))_{(\tilde{\Omega}_{\bm k}^{s+1}\setminus A_{\bm k}^{s+1})-\bm k})^{-1}\|<\frac{1}{2}\delta_{s}^{-3\times\frac{1}{10^4}}.
	\end{align*}
	From the estimate
	\begin{align}\label{qs-2}
		\nonumber\|z-(\theta+\bm k\cdot\bm\omega)\|_{\T}&\le|z|+\|\theta+\bm k\cdot\bm\omega\|_{\T}\\
		&<\delta_s^{\frac{1}{10^3}}+3\delta_s^{\frac{1}{100}}<2\delta_s^{\frac{1}{100}},
	\end{align}
	we apply a Neumann series argument to obtain the key inverse bound
	\begin{align}\label{ms+1-1}
		\|(\mathcal{M}_{s+1}(z)_{(\tilde{\Omega}_{\bm k}^{s+1}\setminus A_{\bm k}^{s+1})-\bm k})^{-1}\|<\delta_s^{-3\times\frac{1}{10^4}}.
	\end{align}
	By Lemma \ref{scl}, the inverse $(\mathcal{M}_{s+1}(z))^{-1}$ is governed by the Schur complement associated with $((\tilde{\Omega}_{\bm k}^{s+1}\setminus A_{\bm k}^{s+1})-\bm k)$:
	\begin{align*}
		\mathcal{S}_{s+1}(z)&=(\mathcal{M}_{s+1}(z))_{A_{\bm k}^{s+1}-\bm k}-\left(\vphantom{\left((\mathcal{M}_{s+1}(z))_{(\tilde{\Omega}_{\bm k}^{s+1}\setminus A_{\bm k}^{s+1})-\bm k}\right)^{-1} }\mathcal{R}_{A_{\bm k}^{s+1}-\bm k}\mathcal{M}_{s+1}(z)\mathcal{R}_{(\tilde{\Omega}_{\bm k}^{s+1}\setminus A_{\bm k}^{s+1})-\bm k}\right.\\
		&\ \ \left.\times\left((\mathcal{M}_{s+1}(z))_{(\tilde{\Omega}_{\bm k}^{s+1}\setminus A_{\bm k}^{s+1})-\bm k}\right)^{-1} \mathcal{R}_{(\tilde{\Omega}_{\bm k}^{s+1}\setminus A_{\bm k}^{s+1})-\bm k}\mathcal{M}_{s+1}(z)\mathcal{R}_{A_{\bm k}^{s+1}-\bm k}\right).
	\end{align*}
	We now turn to the analysis of $\det \mathcal{S}_{s+1}(z)$. Since
	\begin{align*}
		A_{\bm k}^{s+1}-\bm k&=(A_{\bm o}^s-\bm k)\cup(A_{\bm o^*}^s-\bm k),\\
		A_{\bm o}^s-\bm k&\subset\Omega_{\bm o}^s-\bm k,\ A_{\bm o^*}^s-\bm k\subset \Omega_{\bm o^*}^s-\bm k
	\end{align*}
	and
	\begin{align*}
		\dist(\Omega_{\bm o}^s-\bm k,\Omega_{\bm o^*}^s-\bm k)>10\tilde{\zeta}_s,
	\end{align*}
	we have
	\begin{align*}
		(\mathcal{M}_{s+1}(z))_{A_{\bm k}^{s+1}-\bm k}=((\mathcal{M}_{s+1}(z))_{A_{\bm o}^{s}-\bm k})\oplus ((\mathcal{M}_{s+1}(z))_{A_{\bm o^*}^{s}-\bm k}).
	\end{align*}
	From  $\dist(A_{\bm o}^{s},\p \tilde{\Omega}_{\bm o}^s)>\frac{1}{2}\tilde{\zeta}_s$ and $\dist(A_{\bm o^*}^{s},\p \tilde{\Omega}_{\bm o^*}^s)>\frac{1}{2}\tilde{\zeta}_s$, we have
	\begin{align*}
		\mathcal{R}_{A_{\bm o}^{s}-\bm k}\mathcal{M}_{s+1}(z)\mathcal{R}_{(\tilde{\Omega}_{\bm k}^{s+1}\setminus A_{\bm k}^{s+1})-\bm k}&=\mathcal{R}_{A_{\bm o}^{s}-\bm k}\mathcal{M}_{s+1}(z)\mathcal{R}_{(\tilde{\Omega}_{\bm o}^{s}\setminus A_{\bm o}^{s})-\bm k}+O(\delta_s^{20}),\\
		\mathcal{R}_{A_{\bm o^*}^{s}-\bm k}\mathcal{M}_{s+1}(z)\mathcal{R}_{(\tilde{\Omega}_{\bm k}^{s+1}\setminus A_{\bm k}^{s+1})-\bm k}&=\mathcal{R}_{A_{\bm o^*}^{s}-\bm k}\mathcal{M}_{s+1}(z)\mathcal{R}_{(\tilde{\Omega}_{\bm o^*}^{s}\setminus A_{\bm o}^{s})-\bm k}+O(\delta_s^{20}).
	\end{align*}
	For notational clarity, we define the following sets:
	\begin{align*}
		&X=(\tilde{\Omega}_{\bm o}^s\setminus A_{\bm o}^s)-\bm k,\ X^*=(\tilde{\Omega}_{\bm o^*}^s\setminus A_{\bm o^*}^s)-k,\ Y=(\tilde{\Omega}_{\bm k}^{s+1}\setminus A_{\bm k}^{s+1})-\bm k,\\
		&Z_1=\lg_{\frac{\tilde{\zeta}_s}{4}}\cap X,\ Z_2=\lg_{\frac{\tilde{\zeta}_s}{8}}\cap X,\ Z_1^*=\lg_{\frac{\tilde{\zeta}_s}{4}}\cap X^*,\ Z_2^*=\lg_{\frac{\tilde{\zeta}_s}{8}}\cap X^*.
	\end{align*}
	The Schur complement decomposes as
	\begin{align*}
		\mathcal{S}_{s+1}(z)=&\ ((\mathcal{M}_{s+1}(z))_{A_{\bm o}^{s}-\bm k})\oplus ((\mathcal{M}_{s+1}(z))_{A_{\bm o^*}^{s}-\bm k})\\
		&-\left(\vphantom{\left((\mathcal{M}_{s+1}(z))_{(\tilde{\Omega}_{\bm k}^{s+1}\setminus A_{\bm k}^{s+1})-\bm k}\right)^{-1}}(\mathcal{R}_{A_{\bm o}^{s}-\bm k}\oplus \mathcal{R}_{A_{\bm o*}^{s}-\bm k})\mathcal{M}_{s+1}(z)\mathcal{R}_{(\tilde{\Omega}_{\bm k}^{s+1}\setminus A_{\bm k}^{s+1})-\bm k}\right.\\
		&\ \times\left((\mathcal{M}_{s+1}(z))_{(\tilde{\Omega}_{\bm k}^{s+1}\setminus A_{\bm k}^{s+1})-\bm k}\right)^{-1}\\
		&\left.\vphantom{\left((\mathcal{M}_{s+1}(z))_{(\tilde{\Omega}_{\bm k}^{s+1}\setminus A_{\bm k}^{s+1})-\bm k}\right)^{-1}}\times \mathcal{R}_{(\tilde{\Omega}_{\bm k}^{s+1}\setminus A_{\bm k}^{s+1})-\bm k}\mathcal{M}_{s+1}(z)(\mathcal{R}_{A_{\bm o}^{s}-\bm k}\oplus \mathcal{R}_{A_{\bm o*}^{s}-\bm k})\right).
	\end{align*}
	Since $\tilde{\Omega}_{\bm o}^s\setminus A_{\bm o}^s$ is $(s-1)$-$\good$, \eqref{tsgnorm} yields
	\begin{align*}
		\|\mathcal{T}_{\tilde{\Omega}_{\bm o}^s\setminus A_{\bm o}^s}^{-1}\|&<\delta_{s-1}^{-3}.
	\end{align*}
	Equivalently, for the translated operator,
	\begin{align*}
		\|((\mathcal{M}_{s+1}(\theta+\bm k\cdot\bm\omega))_{X})^{-1}\|&<\delta_{s-1}^{-3}.
	\end{align*}
	From \eqref{qs-2} and following the arguments in \eqref{369}-\eqref{370}, we obtain
	\begin{align}
		\label{392}\|((\mathcal{M}_{s+1}(z))_{X})^{-1}\|&\le2 \delta_{s-1}^{-3},\\
\label{393}|((\mathcal{M}_{s+1}(z))_{X})^{-1}(\bm x,\bm y)|&\le e^{-\alpha_{s-1}\log^{\rho}(1+\|\bm x-\bm y\|)}\ \text{for $\|\bm x-\bm y\|>10\tz_{s-1}$.}
	\end{align}
	For $\bm{m} \in X$ and $\bm{n} \in Y$, the resolvent identity gives
	\begin{align*}
		&((\mathcal{M}_{s+1}(z))_Y)^{-1}(\bm m,\bm n)-\chi_X(\bm n)((\mathcal{M}_{s+1}(z))_X)^{-1}(\bm m,\bm n)\\
		=&-\ep\sum_{\bm l\in X \atop \bm l'\in Y\setminus X}((\mathcal{M}_{s+1}(z))_X)^{-1}(\bm m,\bm l)\mathcal{W}(\bm l,\bm l')((\mathcal{M}_{s+1}(z))_Y)^{-1}(\bm l',\bm n).
	\end{align*}
		If $\bm m\in Z_2$, since \eqref{wphi}, \eqref{quaeq}, \eqref{Det}, \eqref{ms+1-1}, \eqref{392}, \eqref{393}, $\dist(Z_1,Y\setminus X)\ge\frac{\tilde{\zeta}_s}{8}\gg10\tz_{s-1}$ and $\dist(Z_2,X\setminus Z_1)\ge\frac{\tilde{\zeta}_s}{8}\gg10\tz_{s-1}$, we  get
	\begin{align*}
		({\rm I})=&\ |((\mathcal{M}_{s+1}(z))_Y)(\bm m,\bm n)-\chi_X(\bm n)((\mathcal{M}_{s+1}(z))_X)^{-1}(\bm m,\bm n)|\\
		\le&\ \sum_{\bm l\in Z_1 \atop \bm l'\in Y\setminus X}|((\mathcal{M}_{s+1}(z))_X)^{-1}(\bm m,\bm l)|\cdot|\mathcal{W}(\bm l,\bm l')|\cdot|((\mathcal{M}_{s+1}(z))_Y)^{-1}(\bm l',\bm n)|\\
		&\ \ +\sum_{\bm l\in X\setminus Z_1 \atop \bm l'\in Y\setminus X}|((\mathcal{M}_{s+1}(z))_X)^{-1}(\bm m,\bm l)|\cdot|\mathcal{W}(\bm l,\bm l')|\cdot|((\mathcal{M}_{s+1}(z))_Y)^{-1}(\bm l',\bm n)|\\
		\le&\  2\delta_{s-1}^{-3}\delta_s^{-\frac{3}{10^4}}e^{-\frac{9}{10}\alpha\log^{\rho}\left(1+\frac{\tz_s}{8}\right)}(\#(Z_1))D\left(\frac{\alpha}{10}\right)\\
	&\ \ +\delta_s^{-\frac{3}{10^4}}e^{-\alpha_{s-1}\log^{\rho}\left(1+\frac{\tz_s}{8}\right)+\alpha_{s-1}C(\rho)\log 2}(\#(X\setminus Z_1))D\left(\frac{\alpha}{10}\right)\\
	<&\ \delta_s^{20}.
	\end{align*}
If $\bm m\in X\setminus Z_2$, by \eqref{wphi}, \eqref{Det}, \eqref{ms+1-1} and \eqref{392}, we obtain
	\begin{align*}
		({\rm II})=&\ |((\mathcal{M}_{s+1}(z))_Y)^{-1}(\bm m,\bm n)-\chi_X(\bm n)((\mathcal{M}_{s+1}(z))_X)^{-1}(\bm m,\bm n)|\\
		\le&\ \sum_{\bm l\in X \atop \bm l'\in Y\setminus X}|((\mathcal{M}_{s+1}(z))_X)^{-1}(\bm m,\bm l)|\cdot|\mathcal{W}(\bm l,\bm l')|\cdot|((\mathcal{M}_{s+1}(z))_Y)^{-1}(\bm l',\bm n)|\\
		\le&\ 2\delta_{s-1}^{-3}\delta_{s}^{-\frac{3}{10^4}}(\#X)D(\alpha).
	\end{align*}
	For $\bm i\in A_{\bm o}^s-\bm k$, $\bm n\in Y$, since \eqref{wphi}, \eqref{Det}, $\dist(A_{\bm o}^s-\bm k,X\setminus Z_2)\ge\frac{\tilde{\zeta}_s}{16}$ and \eqref{ms+1-1}--\eqref{393}, we have
	\begin{align*}
		&\ |\mathcal{R}_{A_{\bm o}^s-\bm k}\mathcal{M}_{s+1}(z)\mathcal{R}_X ((\mathcal{M}_{s+1}(z))_Y)^{-1}(\bm i,\bm n)\\
		&\ \ -\mathcal{R}_{A_{\bm o}^s-\bm k}\mathcal{M}_{s+1}(z)\mathcal{R}_X ((\mathcal{M}_{s+1}(z))_X)^{-1}\mathcal{R}_X(\bm i,\bm n)|\\
		\le&\ \sum_{\bm m\in Z_2}|\mathcal{W}(\bm i,\bm m)|\cdot({\rm I})+\sum_{\bm m\in X\setminus Z_2}|\mathcal{W}(\bm i,\bm m)|\cdot({\rm II})\\
		\le&\ D(\alpha)\delta_s^{20}+2\delta_{s-1}^{-3}\delta_{s}^{-\frac{3}{10^4}}(\#X)D(\alpha)D\left(\frac{\alpha}{10}\right)e^{-\frac{9}{10}\alpha\log^{\rho}\left(1+\frac{\tz_s}{16}\right)}<\delta_{s}^{15}.
	\end{align*}
	It then follows that
	\begin{align*}
		&\ \mathcal{R}_{A_{\bm o}^s-\bm k}\mathcal{M}_{s+1}(z)\mathcal{R}_X ((\mathcal{M}_{s+1}(z))_Y)^{-1}\\
		=&\ \mathcal{R}_{A_{\bm o}^s-\bm k}\mathcal{M}_{s+1}(z)\mathcal{R}_X ((\mathcal{M}_{s+1}(z))_X)^{-1}\mathcal{R}_X+O(\delta_s^{15}).
	\end{align*}
	Similarly,
	\begin{align*}
		&\ \mathcal{R}_{A_{\bm o^*}^s-\bm k}\mathcal{M}_{s+1}(z)\mathcal{R}_{X^*} ((\mathcal{M}_{s+1}(z))_Y)^{-1}\\
		=&\ \mathcal{R}_{A_{\bm o^*}^s-\bm k}\mathcal{M}_{s+1}(z)\mathcal{R}_{X^*} ((\mathcal{M}_{s+1}(z))_{X^*})^{-1}\mathcal{R}_{X^*}+O(\delta_s^{15}).
	\end{align*}
	As a result,
	\begin{align*}
		\mathcal{S}_{s+1}(z)=\mathcal{S}_s\left(z-\frac{\bm l_s}{2}\cdot\bm\omega\right)\oplus \mathcal{S}_s\left(z+\frac{\bm l_s}{2}\cdot\bm\omega\right)+O(\delta_s^{15}).
	\end{align*}
	From \eqref{ls1} and \eqref{zs3}, we have
	\begin{align*}
		\left\|z-\frac{\bm l_s}{2}\cdot\bm\omega-\theta_s\right\|_{\T}\le|z|+\left\|\frac{\bm l_s}{2}\cdot\bm\omega+\theta_s\right\|_{\T}<\delta_s^{\frac{1}{10^3}}+\delta_s^{\frac{1}{100}}<2\delta_s^{\frac{1}{10^3}}
	\end{align*}
	and
	\begin{align*}
		\left\|z+\frac{\bm l_s}{2}\cdot\bm\omega+\theta_s\right\|_{\T}|\le|z|+\left\|\frac{\bm l_s}{2}\cdot\bm\omega+\theta_s\right\|_{\T}<\delta_s^{\frac{1}{10^3}}+\delta_s^{\frac{1}{100}}<2\delta_s^{\frac{1}{10^3}}.
	\end{align*}
	Consequently, both translated points $z-\frac{\bm l_s}{2}\cdot\bm\omega$ and $z+\frac{\bm l_s}{2}\cdot\bm\omega$ lie within the domain defined by \eqref{zs}. Combining this with the determinant estimate \eqref{detss} yields
	\begin{align}
		\label{detss-}\left|\det \mathcal{S}_s\left(z-\frac{\bm l_s}{2}\cdot\bm\omega\right)\right|&\ge\delta_{s-1}	\left\|\left(z-\frac{\bm l_s}{2}\cdot\bm\omega\right)-\theta_s\right\|_{\T}\cdot\left\|\left(z-\frac{\bm l_s}{2}\cdot\bm\omega\right)+\theta_s\right\|_{\T},\\
		\label{detss+}\left|\det \mathcal{S}_s\left(z+\frac{\bm l_s}{2}\cdot\bm\omega\right)\right|&\ge\delta_{s-1}	\left\|\left(z+\frac{\bm l_s}{2}\cdot\bm\omega\right)-\theta_s\right\|_{\T}\cdot\left\|\left(z+\frac{\bm l_s}{2}\cdot\bm\omega\right)+\theta_s\right\|_{\T}.
	\end{align}
	Moreover, since $\#(A_{\bm k}^{s+1}-\bm k)\le 2^{s+1}$, \eqref{ss} and Lemma \ref{det1},  we have
	\begin{align}
		\nonumber&\ \ \sup_{\bm x\in A_{\bm k}^{s+1}-\bm k}\sum_{\bm y\in A_{\bm k}^{s+1}-\bm k}|\mathcal{S}_{s+1}(z)(\bm x,\bm y)|\\
\nonumber	\le&\ \   \sup_{\bm x\in A_{\bm k}^{s+1}-\bm k}\sum_{\bm y\in A_{\bm k}^{s+1}-\bm k}\left|\left(\mathcal{S}_s\left(z-\frac{\bm l_s}{2}\cdot\bm\omega\right)\oplus \mathcal{S}_s\left(z+\frac{\bm l_s}{2}\cdot\bm\omega\right)\right)(\bm x,\bm y)\right|+O(\delta_s^{10})\\
		\label{Ss+10}\le&\ \ 2|v|_R+\sum_{i=0}^{s}\delta_s<4|v|_R,
	\end{align}
	and
	\begin{align}
		\nonumber\det \mathcal{S}_{s+1}(z)&=\det \mathcal{S}_s\left(z-\frac{\bm l_s}{2}\cdot\bm\omega\right)\cdot\det \mathcal{S}_s\left(z+\frac{\bm l_s}{2}\cdot\bm\omega\right)+O((2^{s+1})^2(4|v|_R)^{2^{s+1}}\delta_s^{15})\\
		\label{dets2}&=\det \mathcal{S}_s\left(z-\frac{\bm l_s}{2}\cdot\bm\omega\right)\cdot\det \mathcal{S}_s\left(z+\frac{\bm l_s}{2}\cdot\bm\omega\right)+O(\delta_s^{10}).
	\end{align}
	Notice that
	\begin{align}
		\nonumber\left\|z+\frac{\bm l_s}{2}\cdot\bm\omega-\theta_s\right\|_{\T}&\ge\|\bm l_s\cdot\bm\omega\|_{\T}-\left\|z-\frac{\bm l_s}{2}\cdot\bm\omega-\theta_0\right\|_{\T}\\
		\label{398}&>\frac{\g}{(100N_{s+1}^{10})^\tau}-2\delta_s^{\frac{1}{10^3}}>\delta_s^{\frac{1}{10^4}},
	\end{align}
	and
	\begin{align}
		\nonumber\left\|z-\frac{\bm l_s}{2}\cdot\bm\omega+\theta_s\right\|_{\T}&\ge\|\bm l_s\cdot\bm\omega\|_{\T}-\left\|z+\frac{\bm l_s}{2}\cdot\bm\omega+\theta_0\right\|_{\T}\\
		\label{399}&>\frac{\g}{(100N_{s+1}^{10})^\tau}-2\delta_s^{\frac{1}{10^3}}>\delta_s^{\frac{1}{10^4}}.
	\end{align}
	Let $z_{s+1}$ satisfy
	\begin{align}\label{zs+1}
		z_{s+1}\equiv\frac{\bm l_s}{2}\cdot\bm\omega+\theta_s\ (\text{mod}\ \Z),\ |z_{s+1}|=\left\|\frac{\bm l_s}{2}\cdot\bm\omega+\theta_s\right\|_{\T}<\delta_s^{\frac{1}{100}}.
	\end{align}
	From estimates \eqref{detss-}--\eqref{399}, we derive the determinant bound
	\begin{align*}
		|\det \mathcal{S}_{s+1}(z)|\ge \delta_{s}^{\frac{1}{10}}|(z-z_{s+1})\cdot(z+z_{s+1})+r_{s+1}(z)|,
	\end{align*}
	where $r_{s+1}(z)$ is analytic in the domain \eqref{zs3} and satisfies
	\begin{align}\label{rs+12}
		|r_{s+1}(z)|<\delta_s^{8}.
	\end{align}
	By the Rouch\'e's  theorem, the equation
	\begin{align*}
		(z-z_{s+1})\cdot(z+z_{s+1})+r_{s+1}(z)=0
	\end{align*}
	has exactly two roots $\theta_{s+1}$ and $\theta_{s+1}^{'}$ in \eqref{zs3}, which are perturbations of $\pm z_{s+1}$. Notice that the zero sets coincide
	\begin{align*}
		\left\{|z|\le\delta_s^{\frac{1}{10^3}}:\ \det \mathcal{M}_{s+1}(z)=0\right\}=\left\{|z|\le\delta_s^{\frac{1}{10^3}}:\  \det \mathcal{S}_{s+1}(z)=0\right\}.
	\end{align*}
	Since $\det \mathcal{M}_{s+1}(z)$ is  even, we have
	\begin{align*}
		\theta_{s+1}^{'}=-\theta_{s+1}.
	\end{align*}
	Assuming both
	\begin{align*}
		|z_{s+1}-\theta_{s+1}|>|r_{s+1}(\theta_{s+1})|^{\frac{1}{2}}\ \text{and}\ |z_{s+1}-\theta_{s+1}|>|r_{s+1}(\theta_{s+1})|^{\frac{1}{2}}
	\end{align*}
leads to
	\begin{align*}
		|r_{s+1}(\theta_{s+1})|>|z_{s+1}-\theta_{s+1}|\cdot|z_{s+1}+\theta_{s+1}|>|r_{s+1}(\theta_{s+1})|,
	\end{align*}
	a contradiction. Thus, without loss of generality,
	\begin{align}\label{tzs+1}
		|\theta_{s+1}-z_{s+1}|\le|r_{s+1}(\theta_{s+1})|^{\frac{1}{2}}<\delta_s^{4}.
	\end{align}
	Moreover, for $|z|=\delta_{s}^{\frac{1}{10^3}}$, we have (since \eqref{zs+1} and \eqref{rs+12})
	\begin{align*}
		\frac{|r_{s+1}(z)-r_{s+1}(\theta_{s+1})|}{|(z-z_{s+1})\cdot(z+z_{s+1})+r_{s+1}(\theta_{s+1})|}\le 2\delta_s^{7}.
	\end{align*}
	Combined with $\theta_{s+1}^2-z_{s+1}^2+r_{s+1}(\theta_{s+1})=0$, this yields
	\begin{align*}
		&\ \frac{|(z-z_{s+1})\cdot(z+z_{s+1})+r_{s+1}(z)|}{|(z-\theta_{s+1})\cdot(z+\theta_{s+1})|}\\
		=&\ \frac{|(z-z_{s+1})\cdot(z+z_{s+1})+r_{s+1}(z)|}{|(z-z_{s+1})\cdot(z+z_{s+1})+r_{s+1}(\theta_{s+1})|}\\
		\in&\ \left[1-2\delta_s^{7},1+2\delta_s^{7}\right].
	\end{align*}
	By the maximum modulus principle, we have
	\begin{align*}
	\frac{1}{2}\le \frac{|(z-z_{s+1})\cdot(z+z_{s+1})+r_{s+1}(z)|}{|(z-\theta_{s+1})\cdot(z+\theta_{s+1})|}\le 2.
	\end{align*}
	Thus for $z$ in \eqref{zs3}, we obtain
	\begin{align}\label{detss+12}
		|\det \mathcal{S}_{s+1}(z)|\ge\delta_{s}\|z-\theta_{s+1}\|_{\T}\cdot\|z+\theta_{s+1}\|_{\T}.
	\end{align}
	Since $\delta_{s+1}^{\frac{1}{10^4}}<\frac{1}{2}\delta_{s}^{\frac{1}{10^3}}$, \eqref{zs+1} and \eqref{tzs+1}, we have the containment
	\begin{align*}
		\left\{z\in\C\ :\ \min_{\sigma=\pm1}|z+\sigma\theta_{s+1}|<\delta_{s+1}^{\frac{1}{10^4}}\right\}\subset\left\{z\in\C\ :\ |z|\le\delta_s^{\frac{1}{10^3}}\right\},
	\end{align*}
 proveing \eqref{detss} for step $(s+1)$. The inequalities
	\begin{align*}
		\|\theta+\bm k\cdot\bm\omega+\theta_{s+1}\|_{\T}<10\delta_{s+1}^{\frac{1}{100}}\text{ and } |\theta_{s+1}-z_{s+1}|<\delta_s^{4}
	\end{align*}
		imply
		\begin{align*}
			\left\|\theta+\left(\bm k+\frac{\bm l_s}{2}\right)\cdot\bm\omega+\theta_s\right\|_{\T}<\delta_s.
		\end{align*}
	For  $\bm k\in\Z^d+\frac{1}{2}\sum_{i=0}^{s}\bm l_i$ satisfying
	\begin{align*}
	\|\theta+\bm k\cdot\bm\omega+\theta_{s+1}\|_{\T}<10\delta_{s+1}^{\frac{1}{100}},
	\end{align*}
	we have
	\begin{align*}
		\bm k+\frac{\bm l_s}{2}\in\Z^d+\frac{1}{2}\sum_{i=0}^{s-1}\bm l_i\ \text{and}\ \left\|\theta+\left(\bm k+\frac{\bm l_s}{2}\right)\cdot\bm\omega+\theta_s\right\|_{\T}<\delta_s.
	\end{align*}
	Therefore, by \eqref{qspm}, we have $\bm k+\frac{\bm l_s}{2}\in Q_s^+$. Recalling \eqref{os+12} and \eqref{ps+12}, we have $\bm k\in P_{s+1}$. Thus
	\begin{align*}
		\left\{\bm k\in\Z^d+\frac{1}{2}\sum_{i=0}^{s}\bm l_i:\ \|\theta+\bm k\cdot\bm\omega+\theta_s\|_{\T}<10\delta_{s+1}^{\frac{1}{100}}\right\}\subset P_{s+1}.
	\end{align*}
	Similarly,
	\begin{align*}
		\left\{\bm k\in\Z^d+\frac{1}{2}\sum_{i=0}^{s}\bm l_i:\ \|\theta+\bm k\cdot\bm\omega-\theta_s\|_{\T}<10\delta_{s+1}^{\frac{1}{100}}\right\}\subset P_{s+1}.
	\end{align*}
	This establishes \eqref{fs} for the $(s+1)$-th step.
	
	Finally, we estimate $\|\mathcal{T}_{\tilde{\Omega}_{\bm k}^{s+1}}^{-1}\|$. For $\bm k\in P_{s+1}$, by \eqref{t+kos}, we have
	\begin{align*}
		\theta+\bm k\cdot\bm\omega\in\left\{z\in\C:\ \|z\|_{\T}\le\delta_{s}^{\frac{1}{10^3}}\right\}.
	\end{align*}
	Thus from \eqref{detss+12}, we obtain
	\begin{align*}
		&\ |\det(\mathcal{T}_{A_{\bm k}^{s+1}}-\mathcal{R}_{A_{\bm k}^{s+1}}\mathcal{T}\mathcal{R}_{\tilde{\Omega}_{\bm k}^{s+1}\setminus A_{\bm k}^{s+1}}\mathcal{T}_{\tilde{\Omega}_{\bm k}^{s+1}\setminus A_{\bm k}^{s+1}}^{-1}\mathcal{R}_{\tilde{\Omega}_{\bm k}^{s+1}\setminus A_{\bm k}^{s+1}}\mathcal{T}\mathcal{R}_{A_{\bm k}^{s+1}})|\\
		=&\ |\det \mathcal{S}_{s+1}(\theta+\bm k\cdot\bm\omega)|\\
		\ge&\ \delta_{s}\|\theta+\bm k\cdot\bm\omega-\theta_{s+1}\|\cdot\|\theta+\bm k\cdot\bm\omega+\theta_{s+1}\|.
	\end{align*}
	By \eqref{ss}, Cramer's rule and Lemma \ref{chi}, one has
	\begin{align*}
		&\ \|(\mathcal{T}_{A_{\bm k}^{s+1}}-\mathcal{R}_{A_{\bm k}^{s+1}}\mathcal{T}\mathcal{R}_{\tilde{\Omega}_{\bm k}^{s+1}\setminus A_{\bm k}^{s+1}}\mathcal{T}_{\tilde{\Omega}_{\bm k}^{s+1}\setminus A_{\bm k}^{s+1}}^{-1}\mathcal{R}_{\tilde{\Omega}_{\bm k}^{s+1}\setminus A_{\bm k}^{s+1}}\mathcal{T}\mathcal{R}_{A_{\bm k}^{s+1}})^{-1}\|\\
		=&\ |\det \mathcal{S}_{s+1}(\theta+\bm k\cdot\bm\omega)|^{-1}\\
		&\cdot\|(\mathcal{T}_{A_{\bm k}^{s+1}}-\mathcal{R}_{A_{\bm k}^{s+1}}\mathcal{T}\mathcal{R}_{\tilde{\Omega}_{\bm k}^{s+1}\setminus A_{\bm k}^{s+1}}\mathcal{T}_{\tilde{\Omega}_{\bm k}^{s+1}\setminus A_{\bm k}^{s+1}}^{-1}\mathcal{R}_{\tilde{\Omega}_{\bm k}^{s+1}\setminus A_{\bm k}^{s+1}}\mathcal{T}\mathcal{R}_{A_{\bm k}^{s+1}})^{*}\|\\
		\le&\ 2^{s+1}\cdot(4|v|_R)^{2^{s+1}}\delta_{s}^{-1}\|\theta+\bm k\cdot\bm\omega-\theta_{s+1}\|_{\T}^{-1}\cdot\|\theta+\bm k\cdot\bm\omega+\theta_{s+1}\|_{\T}^{-1}.
	\end{align*}
	From Lemma \ref{scl} and reference \eqref{tosa2}, we derive the following $\ell^2$-norm estimate
	\begin{align}
		\nonumber\|\mathcal{T}_{\tilde{\Omega}_{\bm k}^{s+1}}^{-1}\|&<4(1+\|\mathcal{T}_{\tilde{\Omega}_{\bm k}^{s+1}\setminus A_{\bm k}^{s+1}}\|)^2\\
		\nonumber&\ \ \times(1+\|(\mathcal{T}_{A_{\bm k}^{s+1}}-\mathcal{R}_{A_{\bm k}^{s+1}}\mathcal{T}\mathcal{R}_{\tilde{\Omega}_{\bm k}^{s+1}\setminus A_{\bm k}^{s+1}}\mathcal{T}_{\tilde{\Omega}_{\bm k}^{s+1}\setminus A_{\bm k}^{s+1}}^{-1}\mathcal{R}_{\tilde{\Omega}_{\bm k}^{s+1}\setminus A_{\bm k}^{s+1}}\mathcal{T}\mathcal{R}_{A_{\bm k}^{s+1}})^{-1}\|)\\
		\label{tos+1-12}&<\delta_{s}^{-2}\|\theta+\bm k\cdot\bm\omega-\theta_{s+1}\|_{\T}^{-1}\cdot\|\theta+\bm k\cdot\bm\omega+\theta_{s+1}\|_{\T}^{-1}.
	\end{align}
	For the case of
	\begin{align}\label{ls2}
		\left\|\frac{\bm l_s}{2}\cdot\bm\omega+\theta_s-\frac{1}{2}\right\|_{\T}<\delta_s^{\frac{1}{100}},
	\end{align}
	we obtain the inclusion
	\begin{align}\label{ps+1s2}
		P_{s+1}\subset\left\{\bm k\in\Z^d+\frac{1}{2}\sum_{i=0}^s \bm l_i:\ \left\|\theta+\bm k\cdot\bm\omega-\frac{1}{2}\right\|_{\T}<3\delta_s^{\frac{1}{100}}\right\}.
	\end{align}
	Now we examine $\mathcal{M}_{s+1}(z)$ in the shifted neighborhood
	\begin{align}\label{zs4}
		\left\{z\in\C:\ \left|z-\frac{1}{2}\right|\le \delta_{s}^{\frac{1}{10^3}}\right\}.
	\end{align}
	A similar argument shows that $\det\mathcal{M}_{s+1}(z)=0$ has two roots $\theta_{s+1}$ and $1-\theta_{s+1}$ in the set defined by \eqref{zs4} such that \eqref{Ss+10}--\eqref{tos+1-12} hold true for $z$ in \eqref{zs4}. Hence, if \eqref{ls2} holds, then \eqref{Ss+10}--\eqref{tos+1-12} hold true for $z$ in
	\begin{align*}
		\left\{z\in\C:\ \left\|z-\frac{1}{2}\right\|_{\T}\le \delta_{s}^{\frac{1}{10^3}}\right\}.
	\end{align*}
	By \eqref{ps+1s2}, for any $\bm k\in P_{s+1}$, we also have
	\begin{align*}
		\|\mathcal{T}_{\tilde{\Omega}_{\bm k}^{s+1}}^{-1}\|<\delta_{s}^{-2}\|\theta+\bm k\cdot\bm\omega-\theta_{s+1}\|_{\T}^{-1}\cdot\|\theta+\bm k\cdot\bm\omega+\theta_{s+1}\|_{\T}^{-1}.
	\end{align*}
	
	Therefore, we have established the desired estimates of $\|\mathcal{T}_{\tilde{\Omega}_{\bm k}^{s+1}}^{-1}\|$
	in both cases $(\bm C1)_{s}$ and $(\bm C2)_{s}$.
	
	\  \\
	
	\begin{itemize}
		\item[\textbf{Step2}]: \textbf{Off-diagonal estimates of $\mathcal{T}_{\tilde{\Omega}_{\bm k}^{s+1}}^{-1}$.}
	\end{itemize}
	The main result of this step is Theorem \ref{psqs}. Recalling
	\begin{align*}
		\delta_{s+1}=\delta_s^{10^{5\rho'}}
	\end{align*}
	and
	\begin{align*}
		Q_{s+1}^{\pm}=\left\{\bm k\in P_{s+1}:\ \|\theta+\bm k\cdot\bm\omega\pm\theta_{s+1}\|_{\T}<\delta_{s+1}\right\},\ Q_{s+1}=Q_{s+1}^+\cup Q_{s+1}^-,
	\end{align*}
	we have
	\begin{thm}\label{psqs}
		For $\bm k\in P_{s+1}\setminus Q_{s+1}$,  we have
		\begin{align}\label{tgsa}
			|\mathcal{T}_{\tilde{\Omega}_{\bm k}^{s+1}}^{-1}(\bm x,\bm y)|<e^{-\alpha'_s\log^{\rho}(1+\|\bm x-\bm y\|)}\ \text{for $\|\bm x-\bm y\|\ge\frac{\tilde{\zeta_s}}{10}$},
		\end{align}
		where $\alpha'_s=\alpha_s\left(1-\frac{20\times 10^{5\rho'}}{\alpha\log^{\rho-\rho'}N_{s+1}}\right)$ is defined in \eqref{alpha's}.
	\end{thm}
	
	To prove Theorem \ref{psqs}, we first establish two preparatory lemmas regarding iterative estimates:
		\begin{lem}\label{itetg}
		Assume $t\le s$ and $\lg'\subset\lg\subset\Z^d$ be finite sets. If $\lg'\subset\lg$ is $t$-good, $\bm u'\in\lg'$ and $\bm v\in\lg\setminus\lg'$, then there is some $\bm u'\in\lg \setminus\lg'$ such that
		\begin{align}\label{itetgeq1}
			|\boT_{\lg}^{-1}(\bm u,\bm v)|\le (\#\lg')\cdot e^{-\alpha_t\log^{\rho}(1+\|\bm u-\bm u'\|)}\cdot|\boT_{\lg}^{-1}(\bm u',\bm v)|\cdot\delta_t^{-3}.
		\end{align}
	\end{lem}
	\begin{proof}
		We refer to the Appendix \ref{app} for a detailed proof.
	\end{proof}
	\begin{lem}\label{itepsqt}
		Let $\lg\subset\Z^d$ be a finite set. Assume $t\le s$, $\bm l\in P_t\setminus Q_t$ and $\tO_{\bm l}^{s}\subset\lg$. If $\bm u\in2\Omega_{\bm l}^t$ and $\bm v\in \lg\setminus\tO_{\bm l}^{t}$, then there is some $\bm u'\in \lg\setminus\tO_{\bm l}^{t}$ such that
		\begin{align}
		\nonumber	|\boT_{\lg}^{-1}(\bm u,\bm v)|&\le (\#\tO_{\bm l}^t)\cdot e^{-\alpha_t\log^{\rho}(1+\|\bm u-\bm u'\|)}\cdot	|\boT_{\lg}^{-1}(\bm u',\bm v)|\cdot\delta_t^{-3}\\
	\label{itepsqteq1} &\le e^{-\alpha_t\left(1-\frac{6\times 10^{5\rho'}}{\alpha\log^{\rho-\rho'}N_{t+1}}\right)\log^{\rho}(1+\|\bm u-\bm u'\|)}\cdot	|\boT_{\lg}^{-1}(\bm u',\bm v)|,
		\end{align}
	\end{lem}
		\begin{proof}
		The proof is similar to that of Lemma \ref{itetg}, we omit the details.
	\end{proof}
\begin{rem}
	The second inequality is derived from $\dist(2\Omega_{\bm l}^t,\lg\setminus\tO_{\bm l}^t)\ge\frac{\tz_t}{3}$ and \eqref{indpa}.
\end{rem}
	
	\begin{proof}[Proof of Theorem \ref{psqs}]
		We claim that $\tO_{\bm k}^{s+1}\setminus2\Omega_{\bm k}^{s+1}$ is $s$-good. In fact, for $s'\le s-1$, assume 
		\begin{align*}
			\tO_{\bm l'}^{s'}\subset \tO_{\bm k}^{s+1}\setminus2\Omega_{\bm k}^{s+1},\ \tO_{\bm l'}^{s'}\subset\Omega_{\bm l}^{s'+1}\text{ and }\tO_{\bm l}^{s'+1}\cap 2\Omega_{\bm k}^{s+1}\ne\emptyset.
		\end{align*}
		Thus by \eqref{a1s}, we obtain 
		\begin{align*}
			\Omega_{\bm l}^{s'+1}\subset2\Omega_{\bm k}^{s+1},
		\end{align*}
		which contradicts 
		\begin{align*}
			\tO_{\bm l'}^{s'}\subset \tO_{\bm k}^{s+1}\setminus2\Omega_{\bm k}^{s+1}.
		\end{align*}
		 If there exists a $\bm k'$ such that $\bm k'\in Q_s$ and 
		 \begin{align*}
		 	\tO_{\bm k'}^{s}\subset \tO_{\bm k}^{s+1}\setminus2\Omega_{\bm k}^{s+1}\subset \tO_{\bm k}^{s+1},
		 \end{align*}
		  then by \eqref{311} and $\dist(\tilde{\Omega}_{\bm u}^{s+1},\tilde{\Omega}_{\bm u'}^{s+1})>10\tz_{s+1}\ \text{for}\ \bm u\ne\bm u'\in P_{s+1}$, we have
		\begin{align*}
			\tO_{\bm k'}^{s}\subset \Omega_{\bm k}^{s+1}\subset 2\Omega_{\bm k}^{s+1}\subset \tO_{\bm k}^{s+1}.
		\end{align*}
		This contradicts $\tO_{\bm k'}^{s}\subset\tO_{\bm k}^{s+1}\setminus2\Omega_{\bm k}^{s+1}$. We have proven the claim. 
		
		Since \eqref{tb0} and $\bm k\notin Q_{s+1}$, we have
		\begin{align}\label{tks}
			\|\mathcal{T}_{\tilde{\Omega}_{\bm k}^{s+1}}^{-1}\|<\delta_s^{-2}\|\theta+\bm k\cdot\bm\omega-\theta_s\|_{\T}^{-1}\cdot\|\theta+\bm k\cdot\bm\omega+\theta_s\|_{\T}^{-1}<\delta_s^{-2}\delta_{s+1}^{-2}<\delta_{s+1}^{-3}.
		\end{align}
		To obtain the desired estimates, we divide the remaining proof into several cases.		
		\begin{itemize}
			\item[\textbf{Case 1}:] $\bm x\in 2\Omega_{\bm k}^{s+1}$ or $\bm y\in 2\Omega_{\bm k}^{s+1}$. Without loss of generality, we assume that $\bm y\in2\Omega_{\bm k}^{s+1}$. From $\|\bm x-\bm y\|\ge\frac{\tz_{s+1}}{10}$ and $\diam(2\Omega_{\bm k}^{s+1})\ll\tz_{s+1}$, we have $\bm x\in \tilde{\Omega}_{\bm k}^{s+1}\setminus 2\Omega_{\bm k}^{s+1}$. Since $\tilde{\Omega}_{\bm k}^{s+1}\setminus 2\Omega_{\bm k}^{s+1}$ is $s$-good and \eqref{itetgeq1}, there is some $\bm x'\in 2\Omega_{\bm k}^{s+1}$ such that
			\begin{align}
			\nonumber	|\boT_{\tilde{\Omega}_{\bm k}^{s+1}}^{-1}(\bm x,\bm y)|\le&\ (\#(\tilde{\Omega}_{\bm k}^{s+1}\setminus 2\Omega_{\bm k}^{s+1}))\cdot e^{-\alpha_s\log^{\rho}(1+\|\bm x-\bm x'\|)}\\
				\label{60102}&\cdot|\boT_{\tilde{\Omega}_{\bm k}^{s+1}}^{-1}(\bm x',\bm y)|\cdot\delta_s^{-3}.
			\end{align}
			
			Next, we will extract $\log^{\rho}(1+\|\bm x-\bm y\|)$ from $\log^{\rho}(1+\|\bm x-\bm x'\|)$.  Since $\bm x',\bm y\in 2\Omega_{\bm k}^{s+1}$, $\|\bm y-\bm x'\|\le 4\zeta_{s+1}$. From $\|\bm x-\bm y\|\ge \frac{\tilde{\zeta}_{s+1}}{10}\gg4\zeta_{s+1}\gg1$ and \eqref{exl}, we have
			\begin{align}
				\nonumber&\ \ \log^{\rho}(1+\|\bm x-\bm x'\|)\ge\log^{\rho}(1+\|\bm x-\bm y\|-4\zeta_1)\\
				\nonumber	\ge&\ \ \left(1-2\rho\frac{4\zeta_{s+1}}{(1+\|\bm x-\bm y\|)\log(1+\|\bm x-\bm y\|)}\right)\log^{\rho}(1+\|\bm x-\bm y\|)\\
				\label{60103}	\ge&\ \ \left(1-\frac{80\rho\zeta_{s+1}}{\tz_{s+1}\log N_{s+1}}\right)\log^{\rho}(\|\bm x-\bm y\|+1).
			\end{align}
			Therefore, from \eqref{tks}, \eqref{60102}, \eqref{60103}, $1<\rho'<\rho<\rho'+1$ and $\|\bm x-\bm y\|\ge \frac{\tilde{\zeta}_{s+1}}{10}$, we get
			\begin{align*}
				|\boT_{\tilde{\Omega}_{\bm k}^{s+1}}^{-1}(\bm x,\bm y)|&\le (2\tilde{\zeta}_{s+1}+1)^d\cdot e^{-\alpha_s\left(1-\frac{80\rho\zeta_{s+1}}{\tz_{s+1}\log N_{s+1}}\right)\log^{\rho}(\|\bm x-\bm y\|+1)}\cdot\delta_{s}^{-5}\cdot\delta_{s+1}^{-2}\\
				&\le e^{-\alpha_{s}\left(1-\frac{6\times 10^{5\rho'}}{\alpha\log^{\rho-\rho'} N_{s+1}}\right)\log^{\rho}(\|\bm x-\bm y\|+1)}.
			\end{align*}
			We complete the proof in this case.
			\item[\textbf{Case 2}:] $\bm x\in \tilde{\Omega}_{\bm k}^{s+1}\setminus 2\Omega_{\bm k}^{s+1}$ and $\bm y\in \tilde{\Omega}_{\bm k}^{s+1}\setminus 2\Omega_{\bm k}^{s+1}$. In this case, for $1\le t\le s+1$, we define
			\begin{align}\label{wpt1}
				\widetilde{P}_t=\{\bm l\in P_t:\ \exists\bm l'\in Q_{t-1}\ {\rm s.t.}\ \tO_{\bm l'}^{t-1}\subset\tO_{\bm k}^{s+1},\tO_{\bm l'}^{t-1}\subset\Omega_{\bm l}^{t}\}.
			\end{align}			
			From \eqref{a1s},  \eqref{311} and \eqref{wpt1}, it follows that for $\bm l'\in\widetilde{P}_t\cap Q_t$ $(1\le t\le s)$, there is a $\bm l\in \widetilde{P}_{t+1}$ such that
			\begin{align*}
				\tO_{\bm l'}^{t}\subset \Omega_{\bm l}^{t+1}.
			\end{align*}
			Hence for any $\bm z\in \tO_{\bm k}^{s+1}\setminus2\Omega_{\bm k}^{s+1}$, if 
			\begin{align*}
				\bm z\in\bigcup_{\bm l\in\widetilde{P}_1}2\Omega_{\bm l}^{1},
			\end{align*}
			then there exists a $t\in[1,s]$ such that
			\begin{align*}
				\bm z\in \bigcup_{\bm l\in\widetilde{P}_t\setminus Q_t}2\Omega_{\bm l}^{t}.
			\end{align*}
			Therefore, for $\bm z\in \tO_{\bm k}^{s+1}\setminus2\Omega_{\bm k}^{s+1}$, we can define
			\begin{align*}
				O(\bm z)=\left\{\begin{array}{ll}
					\tO_{\bm k}^{s+1}\cap\lg_{\frac{1}{2}N_1}(\bm z), & \text{if }\bm z\in\tO_{\bm k}^{s+1}\setminus\left(\bigcup_{\bm l\in \widetilde{P}_1}2\Omega_{\bm l}^{1}\right),\\
					\tO_{\bm l}^{t}, & \text{if }\bm z\in 2\Omega_{\bm l}^{t}\text{ for some $\bm l\in\widetilde{P}_t\setminus Q_t$}.
				\end{array}\right.
			\end{align*}
			Moreover, for $\bm z\in\tilde{\Omega}_{\bm k}^{s+1}$, by Corollary \ref{itepsqt}, we can define
			\begin{align*}
				\hat{\bm z}=\left\{\begin{array}{ll}
					\bm z, & \bm z\in 2\Omega_{\bm k}^{s+1}\bigcup O(\bm y),\\
					\bm z', & \bm z\in \tilde{\Omega}_{\bm k}^{s+1}\setminus\left(2\Omega_{\bm k}^{s+1}\bigcup O(\bm y)\right).
				\end{array}\right.
			\end{align*}
			
			Let $\bm x_0:=\bm x$ and $\bm x_{l+1}=\hat{\bm x}_{l}$, $l\ge0$. For given $\{\bm x_{l}\}_{l\in\N}$, we define $l_1\ge1$ to be the smallest integer so that $\bm x_{l_1}\in 2\Omega_{\bm k}^{s+1}\bigcup O(\bm y)$. We then have 
			\begin{align*}
				\bm x_{i}\in \tilde{\Omega}_{\bm k}^{s+1}\setminus\left(2\Omega_{\bm k}^{s+1}\bigcup O(\bm y)\right)\text{ for }0\le i<l_1.
			\end{align*}
		We also divide the discussion into 3 cases:
			\item[\textbf{Case 2-1}:] $l_1>\left[\frac{2\alpha \log^{\rho}(\|\bm x-\bm y\|+1)+6\times 10^{5\rho'}\log^{\rho'}(N_{s+1}+1)}{\alpha\log^{\rho}\left(\frac{N_1}{2}+1\right)}\right]+1:=l^*$. Since $\bm x_{i}\in \tilde{\Omega}_{\bm k}^{s+1}\setminus\left(2\Omega_{\bm k}^{s+1}\bigcup O(\bm y)\right)$ for $0\le i<l_1$, we have
			\begin{align*}
				\bm x_{i+1}=\bm x'_{s+1}\in \tilde{\Omega}_{\bm k}^{1}\setminus O(\bm x_i)\ \text{(cf. Corollary \ref{itepsqt})}
			\end{align*}
			 and then 
			 \begin{align*}
			 	\|\bm x_{i+1}-\bm x_{i}\|\ge\frac{N_1}{2}\text{ for $0\le i<l_1$.}
			 \end{align*}
			 Thus, from \eqref{indpa},\eqref{ite0g0}, \eqref{itepsqteq1}, \eqref{tks}  and
			 \begin{align*}
			 	\alpha_s\left(1-\frac{6\times 10^{5\rho'}}{\alpha\log^{\rho-\rho'}N_{s+1}}\right)\ge\frac{\alpha}{2},
			 \end{align*}
			  we get
			\begin{align*}
				|\boT_{\tilde{\Omega}_{\bm k}^{s+1}}^{-1}(\bm x,\bm y)|&\le \prod_{i=0}^{l^*-1}\left(e^{-\alpha_s\left(1-\frac{6\times 10^{5\rho'}}{\alpha\log^{\rho-\rho'}N_{s+1}}\right)\log^{\rho}(\|\bm x_{i+1}-\bm x_{i}\|+1)}\right)|\boT_{\tilde{\Omega}_{\bm k}^{s+1}}^{-1}(\bm x_{l^*},\bm y)|\\
				&\le e^{-\frac{1}{2}\alpha\cdot l^*\log^{\rho}\left(\frac{N_1}{2}+1\right)}\delta_{s+1}^{-3}\\
				&\le e^{-\frac{1}{2}\alpha\cdot l^*\log^{\rho}\left(\frac{N_1}{2}+1\right)}\cdot e^{3\times 10^{5\rho'}\log^{\rho'}(N_{s+1}+1)}\\
				&\le e^{-\alpha\log^{\rho}(\|\bm x-\bm y\|+1)}.
			\end{align*}
			\item[\textbf{Case 2-2:}] $l_1\le l^*$ and $\bm x_{l_1}\in O(\bm y)$. According to $\|\bm x-\bm y\|\ge\frac{\tz_{s+1}}{10}$, $N_{s+1}\gg1$ and $1<\rho'<\rho$, we obtain
			\begin{align}
				\nonumber l^*&=\left[\frac{2\alpha \log^{\rho}(\|\bm x-\bm y\|+1)+6\times 10^{5\rho'}\log^{\rho'}(N_{s+1}+1)}{\alpha\log^{\rho}\left(\frac{N_1}{2}+1\right)}\right]+1\\
				\label{l*2}	&\le \log^{\rho}(1+\|\bm x-\bm y\|).
			\end{align}
			Then from \eqref{quaeq}, \eqref{ite0g0}, \eqref{itepsqteq1}, \eqref{tks} and $\|\bm x-\bm y\|\ge\frac{\tz_{s+1}}{10}$, we have
			\begin{align}
				\nonumber	|\boT_{\tilde{\Omega}_{\bm k}^{s+1}}^{-1}(\bm x,\bm y)|&\le \prod_{i=0}^{l_1-1}\left(e^{-\alpha_s\left(1-\frac{6\times 10^{5\rho'}}{\alpha\log^{\rho-\rho'}N_{s+1}}\right)\log^{\rho}(\|\bm x_{i+1}-\bm x_{i}\|+1)}\right)|\boT_{\tilde{\Omega}_{\bm k}^{s+1}}^{-1}(\bm x_{l_1},\bm y)|\\
				\nonumber	&\le e^{-\alpha_s\left(1-\frac{6\times 10^{5\rho'}}{\alpha\log^{\rho-\rho'}N_{s+1}}\right)\left(\log^{\rho}(\|\bm x_{l_1}-\bm x\|+1)-C(\rho)\log^{\rho}l_1\right)}\delta_{s+1}^{-3}\\
			\nonumber&\le e^{-\alpha_s\left(1-\frac{6\times 10^{5\rho'}}{\alpha\log^{\rho-\rho'}N_{s+1}}\right)\left(\log^{\rho}(\|\bm x-\bm y\|-\|\bm x_{l_1}-\bm y\|+1)-C(\rho)\log^{\rho}l_1\right)}\\
					\label{60104}&\ \ \cdot e^{3\times 10^{5\rho'}\log^{\rho'}(N_{s+1}+1)}.
			\end{align}
			By $\bm x_{l_1}\in O(\bm y)$, $\|\bm x_{l_1}-\bm y\|\le\frac{\tz_s}{2}$. Since $\|\bm x-\bm y\|\ge\frac{\tz_{s+1}}{10}\gg \frac{\tz_s}{2}\gg1$ and \eqref{exl}, we can get
			\begin{align}
				\nonumber&\ \ \log^{\rho}(\|\bm x-\bm y\|-\|\bm x_{l_1}-\bm y\|+1)\\
				\nonumber\ge&\ \ \log^{\rho}\left(\|\bm x-\bm y\|-\frac{\tz_s}{2}+1\right)\\
				\nonumber\ge&\ \ \left(1-\rho\frac{\tz_{s}}{(1+\|\bm x-\bm y\|)\log(1+\|\bm x-\bm y\|)}\right)\log^{\rho}(1+\|\bm x-\bm y\|)\\
				\label{60105}	\ge&\ \ \left(1-\frac{10\rho \tz_s}{\tz_{s+1}\log N_{s+1}}\right)\log^{\rho}(\|\bm x-\bm y\|+1).
			\end{align}
			From \eqref{l*2}, \eqref{60105}, $\|\bm x-\bm y\|\ge\frac{\tz_{s+1}}{10}\gg1$ and $1<\rho'<\rho<\rho'+1$, we get
			\begin{align*}
				&\ \ \log^{\rho}(\|\bm x-\bm y\|-\|\bm x_{l_1}-\bm y\|+1)-C(\rho)\log^{\rho}l_1\\
				\ge&\ \  \left(1-\frac{10^{5\rho'}}{\alpha\log^{\rho-\rho'}N_{s+1}}\right)\log^{\rho}(\|\bm x-\bm y\|+1)
			\end{align*}
			and then
			\begin{align}
				\nonumber	&\ \ \left(1-\frac{10^{5\rho'}}{\alpha\log^{\rho-\rho'} N_{s+1}}\right)\left(\log^{\rho}(\|\bm x-\bm y\|-\|\bm x_{l_1}-\bm y\|+1)-C(\rho)\log^{\rho}l_1\right)\\
				\label{60106}	\ge&\ \ \left(1-\frac{3\times 10^{5\rho'}}{\alpha\log^{\rho-\rho'}N_{s+1}}\right)\log^{\rho}(\|\bm x-\bm y\|+1).
			\end{align}
			Since $\|\bm x-\bm y\|\ge\frac{\tz_{s+1}}{10}\gg N_{s+1}\gg1$, \eqref{60104} and \eqref{60106}, we have
			\begin{align*}
				|\boT_{\tilde{\Omega}_{\bm k}^{s+1}}^{-1}(\bm x,\bm y)|&\le e^{-\alpha_s\left(1-\frac{10\times 10^{5\rho'}}{\alpha\log^{\rho-\rho'}N_{s+1}}\right)\log^{\rho}(\|\bm x-\bm y\|+1)}\cdot e^{3\times 10^{5\rho'}\log^{\rho'}(N_{s+1}+1)}\\
				&\le e^{-\alpha_s\left(1-\frac{20\times 10^{5\rho'}}{\alpha\log^{\rho-\rho'}N_{s+1}}\right)\log^{\rho}(\|\bm x-\bm y\|+1)}.
			\end{align*}
			\item[\textbf{Case 2-3}:] $l_1\le l^*$ and $\bm x_{l_1}\in 2\Omega_{\bm k}^{s+1}$. From a similar argument of \eqref{60104}, we obtain
			\begin{align}\label{60107}
				\nonumber	|\boT_{\tilde{\Omega}_{\bm k}^{1}}^{-1}(\bm x,\bm y)|&\le e^{-\alpha_s\left(1-\frac{6\times 10^{5\rho'}}{\alpha\log^{\rho-\rho'} N_{s+1}}\right)\left(\log^{\rho}(\|\bm x-\bm y\|-\|\bm x_{l_1}-\bm y\|+1)-C(\rho)\log^{\rho}l_1\right)}\\
				&\ \ \cdot|\boT_{\tilde{\Omega}_{\bm k}^{1}}^{-1}(\bm x_{l_1},\bm y)|
			\end{align}
			By Lemma \ref{itetg} (since $\tO_{\bm k}^{s+1}\setminus2\Omega_{\bm k}^{s+1}$ is $s$-good), there is some $\bm y'\in2\Omega_{\bm k}^{s+1}$ such that
			\begin{align}\label{60108}
				|\boT_{\tilde{\Omega}_{\bm k}^{s+1}}^{-1}(\bm x_{l_1},\bm y)|\le (\#\tO_{\bm k}^{s+1})\cdot e^{-\alpha_{s}\log^{\rho}(1+\|\bm y-\bm y'\|)}\cdot|\boT_{\tilde{\Omega}_{\bm k}^{1}}^{-1}(\bm x_{l_1},\bm y')|\cdot\delta_s^{-3}.
			\end{align}
			According to $\bm x_{l_1}, \bm y'\in2\Omega_{\bm k}^{s+1}$, we have $\|\bm x_{l_1}-\bm y'\|\le 4\zeta_{s+1}$. Since \eqref{quaeq} and $\|\bm x-\bm y\|\ge\frac{\tz_{s+1}}{10}\gg1$,
			\begin{align}
				\nonumber&\ \ \log^{\rho}(\|\bm x_{l_1}-\bm x\|+1)+\log^{\rho}(1+\|\bm y-\bm y'\|)\\
				\nonumber\ge&\ \ \log(1+\|\bm x_{l_1}-\bm x\|+\|\bm y-\bm y'\|)-C(\rho)\log^{\rho}2\\
				\nonumber\ge&\ \ \log(1+\|\bm x-\bm y\|-\|\bm x_{l_1}-\bm y'\|)-C(\rho)\log^{\rho}2\\
				\label{60109}\ge&\ \ \log(1+\|\bm x-\bm y\|-4\zeta_{s+1})-C(\rho)\log^{\rho}2.
			\end{align}
			From \eqref{803}, $\tz_{s+1}\gg\zeta_{s}$, $1<\rho'<\rho<\rho'+1$ and $N_{s+1}\gg1$, we also get
			\begin{align}
				\nonumber	&\ \ \log(1+\|\bm x-\bm y\|-4\zeta_{s+1})\\
				\nonumber\ge&\ \ \left(1-\frac{80\rho\zeta_{s+1}}{\tz_{s+1}\log N_{s+1}}\right)\log^{\rho}(\|\bm x-\bm y\|+1)\\
				\label{60110}	\ge&\ \ \left(1-\frac{ 10^{5\rho'}}{\alpha\log^{\rho-\rho'}N_{s+1}}\right)\log^{\rho}(\|\bm x-\bm y\|+1).
			\end{align}
			Combining \eqref{tk0}, \eqref{60107}--\eqref{60110}, $1<\rho'<\rho<\rho'+1$ and $\|\bm x-\bm y\|\ge\frac{\tz_{s+1}}{10}\gg1$ gives
			\begin{align*}
				|\boT_{\tilde{\Omega}_{\bm k}^{1}}^{-1}(\bm x,\bm y)|&\le e^{-\alpha_s\left(1-\frac{20\times 10^{5\rho'}}{\alpha\log^{\rho-\rho'}N_1}\right)\log^{\rho}(\|\bm x-\bm y\|+1)}.
			\end{align*}
		\end{itemize}
		This finishes the proof.
	\end{proof}

	\ \\

	\begin{itemize}
		\item[\textbf{Step 3}]: \textbf{Estimates of general $(s+1)$-good $\lg$}.
	\end{itemize}
	We now finalize the verification of property $\mathscr{P}_{s+1}$.
	First, we recall the precise definition of an $(s+1)$-$\good$ set: a finite set $\lg\subset\Z^d$ is $(s+1)$-$\good$ iff
	\begin{align}\label{s+1gdef}
		\left\{\begin{array}{l}
			\bm k'\in Q_{s'},\ \tO_{\bm k'}^{s'}\subset\lg,\ \tO_{\bm k'}^{s'}\subset\Omega_{\bm k}^{s'+1}\Rightarrow\tO_{\bm k}^{s'+1}\subset\lg\ \text{for $s'<s+1$},\\
			\{\bm k\in P_{s+1}:\ \tilde{\Omega}_{\bm k}^{s+1}\subset\lg\}\cap Q_{s+1}=\emptyset.
		\end{array}\right.
	\end{align}
	
	To establish $\mathscr{P}_{s+1}$, we will synthesize three key analytical tools: (a) The $\ell^2$-norm estimates for $\mathcal{T}_{\tilde{\Omega}_{\bm k}^{s+1}}^{-1}$ developed previously; (b) Schur's test for bounding matrix operators; (c) The resolvent identity technique.
	\begin{thm}\label{s+1g}
		If $\lg$ is $(s+1)$-$\good$, then
		\begin{align}
			\nonumber\|\boT_\lg^{-1}\|\le&\  2\delta_{s}^{-3}\sup_{\{\bm k\in P_{s+1}:\ \tO_{\bm k}^{s+1}\subset\lg\}}\left(\|\theta+\bm k\cdot\bm\omega-\theta_{s+1}\|_{\T}^{-1}\cdot\|\theta+\bm k\cdot\bm\omega+\theta_{s+1}\|_{\T}^{-1}\right)\\
			\label{s+1gnorm}<&\ \delta_{s+1}^{-3}.
		\end{align}
		and
		\begin{align}
			\label{s+1gdecay}	|\boT_{\lg}(\bm x,\bm y)|<e^{-\alpha_{s+1}\|\bm x-\bm y\|}\ \text{for $\|\bm x-\bm y\|\ge10\tz_{s+1}$}.
		\end{align}
	\end{thm}
	\begin{proof}
				First,  we prove \eqref{1gnorm} by Schur's test. Define
		\begin{align}\label{wpt2}
			\widetilde{P}_t=\{\bm k\in P_t:\ \exists \bm k'\in Q_{t-1}\ \text{s.t. }\tO_{\bm k'}^{t-1}\subset\lg,\ \tO_{\bm k'}^{t-1}\subset\Omega_{\bm k}^{t}\},\ 1\le t\le s+1.
		\end{align}
		From \eqref{311}, \eqref{s+1gdef} and \eqref{wpt2}, it follows that for $\bm k'\in\widetilde{P}_t\cap Q_t$ $(1\le t\le s)$, there exists a $\bm k\in\widetilde{P}_{t+1}$ such that
		\begin{align*}
			\tO_{\bm k'}^{t}\subset \Omega_{\bm k}^{t+1}
		\end{align*}
		and
		\begin{align*}
			\widetilde{P}_{s+1}\cap Q_{s+1}=\emptyset.
		\end{align*}
		Hence for any $\bm w\in\lg$, if
		\begin{align*}
			\bm w\in\bigcup_{\bm k\in\widetilde{P}_1}2\Omega_{\bm k}^1,
		\end{align*}
		then there exists a $t\in[1,s+1]$ such that
		\begin{align}\label{wincup}
			\bm w\in\bigcup_{\bm k\in\widetilde{P}_t\setminus Q_t}2\Omega_{\bm k}^{t}.
		\end{align}
		 For every $\bm w\in\lg$, define its block in $\lg$:
		\begin{align}\label{neighbor2}
			O(\bm w)=\left\{\begin{array}{ll}
				\lg_{\frac{1}{2}N_1}(\bm w)\cap \lg, & \text{if $\bm w\notin\bigcup_{\bm k\in\widetilde{P}_1}2\Omega_{\bm k}^{1}$,}\\
				\tilde{\Omega}_{\bm k}^t, & \text{if $\bm w\in2\Omega_{\bm k}^{t}$ for some $\bm k\in\widetilde{P}_t\setminus Q_t$.}
			\end{array}\right.
		\end{align}
		Since $\lg$ is $(s+1)$-good, $\boT_{O(\bm x)}^{-1}$ exist for all $\bm x\in\lg$. Hence we can define $\boL$ and $\boK$ as
		\begin{align}\label{defL2}
			\boL(\bm x,\bm y)=\left\{\begin{array}{ll}
				\boT_{O(\bm x)}^{-1}(\bm x,\bm y), & \text{for $\bm x\in\lg$ and $\bm y\in O(\bm x)$,}\\
				0, & \text{else,}
			\end{array}\right.
		\end{align}
		and
		\begin{align}\label{defK2}
			\boK(\bm x,\bm y)=\left\{\begin{array}{ll}
				\sum\limits_{\bm z\in O(\bm x)}\boL(\bm x,\bm z)\boW(\bm z,\bm y), & \text{for $\bm x\in\lg$ and $\bm y\in\lg\setminus O(\bm x)$,}\\
				0, & \text{else}.
			\end{array}\right.
		\end{align}
		Direct computations shows
		\begin{align*}
			\boL \boT_{\lg}=\boI_{\lg}+\ep\boK.
		\end{align*}
		The norm estimate \eqref{s+1gnorm} follows by applying an argument parallel to the proof of \eqref{1gnorm}, with appropriate modifications for the $(s+1)$-scale case.
			
			We estimate $|\boK(\bm x,\bm y)|$  first. We divide the discussion into three cases:
		\begin{itemize}
			\item[\textbf{Case 1}:] $\bm x\notin\bigcup_{\bm k\in\widetilde{P}_1}2\Omega_{\bm k}^{1}$ and $\bm y\in \lg\setminus O(\bm x)$. In the same manner as \eqref{Kle11}, we can obtain
			\begin{align}\label{2Kle11}
				|\boK(\bm x,\bm y)|\le e^{-\frac{5}{8}\alpha\log^{\rho}\left(1+\|\bm x-\bm y\|\right)}.
			\end{align}
			
			\item[\textbf{Case 2}:] $\bm x\in 2\Omega_{\bm k}^{t}$ for some $\bm k\in\widetilde{P}_t\setminus Q_t$ and $\bm y\in \lg\setminus O(\bm x)$. In this case, we define $X:=\lg_{\frac{\tz_t}{9}}(\bm k)\subset O(\bm x)$. Hence
			\begin{align}
				\nonumber|\boK(\bm x,\bm y)|&\le\sum_{\bm z\in O(\bm x)}|\boT_{O(\bm x)}^{-1}(\bm x,\bm z)|\cdot|\boW(\bm z,\bm y)|\\
				\label{2K12}&\le ({\rm I})+({\rm II}),
			\end{align}
			where
			\begin{align*}
				({\rm I})=\sum_{\bm z\in O(\bm x)\setminus X}|\boT_{O(\bm x)}^{-1}(\bm x,\bm z)|\cdot|\boW(\bm z,\bm y)|
			\end{align*}
			and
			\begin{align*}
				({\rm II})= \sum_{\bm z\in X}|\boT_{O(\bm x)}^{-1}(\bm x,\bm z)|\cdot|\boW(\bm z,\bm y)|.
			\end{align*}
			For $({\rm I})$, since \eqref{wphi}, \eqref{quaeq}, $\bm z\in O(\bm x)\setminus X$ ($\|\bm z-\bm x\|\ge\frac{\tz_t}{10}$), $\bm y\notin O(\bm x)$ ($\|\bm x-\bm y\|\ge\frac{\tz_t}{3}$) and \eqref{PsQsdecay}, we have
			\begin{align}
				\nonumber({\rm I})=&\ \ \sum_{\bm z\in O(\bm x)\setminus X}|\boT_{O(\bm x)}^{-1}(\bm x,\bm z)|\cdot|\boW(\bm z,\bm y)|\\
				\nonumber\le &\ \ \sum_{\bm z\in O(\bm x)\setminus X} e^{-\alpha'_{t-1}\log^{\rho}(1+\|\bm x-\bm z\|)}\cdot e^{-\alpha\log^{\rho}(1+\|\bm z-\bm y\|)}\\
				\nonumber\le &\ \ e^{-\alpha'_{t-1}\log^{\rho}(1+\|\bm x-\bm y\|)+\alpha'_{t-1}\cdot C(\rho)\log^{\rho}2}\sum_{\bm z\in O(\bm x)\setminus X}e^{-\frac{1}{10}\alpha\log^{\rho}(1+\|\bm z-\bm y\|)}\\
				\nonumber\le &\ \  D\left(\frac{\alpha}{10}\right)\cdot e^{-\alpha'_{t-1}\log^{\rho}(1+\|\bm x-\bm y\|)+\alpha'_{t-1}\cdot C(\rho)\log^{\rho}2}\\
				\label{2K1}	\le&\ \  \frac{1}{2}e^{-\frac{5}{8}\alpha'_{t-1}\log^{\rho}(1+\|\bm x-\bm y\|)},
			\end{align}
			where $D\left(\frac{\alpha}{10}\right)$ is defined in \eqref{Det}.  For $({\rm II})$, from \eqref{wphi}, \eqref{PsQsdecay} and $\bm k\notin Q_{t}$, we obtain
			\begin{align}
				\label{60201}\sum_{\bm z\in X}|\boT_{O(\bm x)}^{-1}(\bm x,\bm z)|\cdot|\boW(\bm z,\bm y)|\le \delta_{t-1}^{-2}\delta_{t}^{-2}\sum_{\bm z\in X} e^{-\alpha\log^{\rho}(1+\|\bm z-\bm y\|)}.
			\end{align}
			By $\bm z\in X$, $\bm y\in\lg\setminus O(\bm x)$ ($\|\bm y-\bm x\|\ge\frac{\tz_t}{3}\ge2\|\bm x-\bm z\|$) and \eqref{exl}, we get
			\begin{align}
				\nonumber&\ \ \log^{\rho}(1+\|\bm y-\bm z\|)\\
				\nonumber\ge&\ \ 	\log^{\rho}(1+\|\bm y-\bm x\|-\|\bm x-\bm z\|)\\
				\nonumber	\ge&\ \ \left(1-2\rho\frac{\|\bm x-\bm z\|}{(1+\|\bm y-\bm x\|)\log(1+\|\bm y-\bm x\|)}\right)\log^{\rho}(1+\|\bm y-\bm x\|)\\
				\label{60202}\ge&\ \ \left(1-\frac{1}{\log N_t}\right)\log^{\rho}(1+\|\bm y-\bm x\|).
			\end{align}
			Hence, combining \eqref{60201}, \eqref{60202},  $\|\bm x-\bm y\|\ge\frac{\tz_t}{3}\gg1$, \eqref{indpa} and $1<\rho'<\rho$ gives
			\begin{align}
				\nonumber({\rm II})=&\ \ \sum_{\bm z\in X}|\boT_{O(\bm x)}^{-1}(\bm x,\bm z)|\cdot|\boW(\bm z,\bm y)|\\
				\nonumber\le&\ \  \delta_{t-1}^{-2}\delta_t^{-2}\cdot e^{-\frac{3}{4}\alpha\log^{\rho}(1+\|\bm y-\bm x\|)}\cdot\sum_{\bm z\in X}e^{-\frac{1}{10}\alpha\log^{\rho}(1+\|\bm z-\bm y\|)}\\
				\nonumber	\le&\ \  \delta_{t-1}^{-2}\delta_t^{-2}\cdot D\left(\frac{\alpha}{10}\right)\cdot e^{-\frac{3}{4}\alpha\log^{\rho}(1+\|\bm y-\bm x\|)}\\
				\label{2K2}\le &\ \ \frac{1}{2}e^{-\frac{5}{8}\alpha'_{t-1}\log^{\rho}(1+\|\bm x-\bm y\|)},
			\end{align}
			where $D\left(\frac{\alpha}{10}\right)$ is defined in \eqref{Det}.	From \eqref{2K12}, \eqref{2K1} and \eqref{2K2}, we have
			\begin{align}\label{2Kle12}
				|\boK(\bm x,\bm y)|\le e^{-\frac{5}{8}\alpha'_{t-1}\log^{\rho}(1+\|\bm x-\bm y\|)}.
			\end{align}
			
			\item[\textbf{Case 3}:] $\bm x\in\lg$ and $\bm y\in O(\bm x)$. From \eqref{defK2}, we get
			\begin{align}\label{2Kle13}
				|\boK(\bm x,\bm y)|=0\le e^{-\frac{5}{8}\alpha'_{s}\log^{\rho}(1+\|\bm x-\bm y\|)}.
			\end{align}
		\end{itemize}
		In summary, we obtain
		\begin{align*}
			|\boK(\bm x,\bm y)|\le e^{-\frac{5}{8}\alpha'_s\log^{\rho}(1+\|\bm x-\bm y\|)}\text{ for all $\bm x,\bm y\in\lg$}.
		\end{align*}
		Therefore,
		\begin{align*}
			\sup_{\bm x\in\lg}\sum_{\bm y\in\lg}|\boK(\bm x,\bm y)|&\le D\left(\frac{5}{8}\alpha'_s\right)<+\infty,\\
			\sup_{\bm y\in\lg}\sum_{\bm x\in\lg}|\boK(\bm x,\bm y)|&\le D\left(\frac{5}{8}\alpha'_s\right)<+\infty,
		\end{align*}
		where $D\left(\frac{5}{8}\alpha'_s\right)$ is defined in \eqref{Det}. By Schur's test, we can get
		\begin{align*}
			\|\boK\|\le D\left(\frac{5}{8}\alpha'_s\right)\le D\left(\frac{5}{16}\alpha\right)<+\infty.
		\end{align*}
		From $\ep\ll1$, we have $\boI_{\lg}+\ep\boK$ is invertible and
		\begin{align}\label{2I+eK-1}
			\|(\boI_{\lg}+\ep\boK)^{-1}\|\le 2.
		\end{align}
		At this time,
		\begin{align}\label{2tlg-1}
			\boT_{\lg}^{-1}=(\boI_{\lg}+\ep\boK)^{-1}\boL
		\end{align}
		exists. Then, we estimate $\|\boL\|$ in order to estimate $\|\boT_{\lg}^{-1}\|$ by using similar methods.
		
		We deal with $\sup_{\bm x\in\lg}\sum_{\bm y\in\lg}|\boL(\bm x,\bm y)|$ first. We divide the discussion into two cases:
		\begin{itemize}
			\item[\textbf{Case 1}:] $\bm x\in\lg\setminus\bigcup_{\bm k\in\widetilde{P}_1}2\Omega_{\bm k}^{1}$. From \eqref{0gl2} and \eqref{defK2}, we have
			\begin{align}
				\nonumber\sum_{\bm y\in\lg}|\boL(\bm x,\bm y)|=&\ \sum_{\bm y\in O(\bm x)}|\boL(\bm x,\bm y)|\le (\# O(\bm x))\cdot\|\boT_{O(\bm x)}^{-1}\|\\
				\label{2L11}<&\ (N_1+1)^d\cdot2\kappa_1^{-1}\delta_0^{-2}<\delta_0^{-\frac{5}{2}}.
			\end{align}
			\item[\textbf{Case 2}:] $\bm x\in 2\Omega_{\bm k}^{t}$ for some $\bm k\in\widetilde{P}_t\setminus Q_t$. In this case, by \eqref{PsQsnorm} and \eqref{defK2}, we get
			\begin{align}
				\nonumber	\sum_{\bm y\in\lg}|\boL(\bm x,\bm y)|&=\sum_{\bm y\in O(\bm x)}|\boL(\bm x,\bm y)|\le (\# O(\bm x))\cdot\|\boT_{O(\bm x)}^{-1}\|\\
				\label{2L12}	&<\delta_{t-1}^{-3}\|\theta+\bm k\cdot\bm\omega-\theta_t\|_{\T}^{-1}\cdot\|\theta+\bm k\cdot\bm\omega+\theta_t\|_{\T}^{-1}<\delta_{t}^{-\frac{5}{2}}.
			\end{align}
		\end{itemize}
		Since $\theta,\theta_1,\cdots,\theta_{s+1}\in \D_R$, we obtain
		\begin{align*}
			\|\theta+\bm k\cdot\bm\omega\pm\theta_i\|\le\sqrt{4R^2+\frac{1}{4}},\ 1\le i\le s+1,
		\end{align*}
		and
		\begin{align}\label{2L13}
			\delta_{t}^{-3}\|\theta+\bm k\cdot\bm\omega-\theta_{t+1}\|_{\T}^{-1}\cdot\|\theta+\bm k\cdot\bm\omega+\theta_{t+1}\|_{\T}^{-1}>\delta_{t}^{-\frac{5}{2}}.
		\end{align}

		Combining \eqref{2L11}--\eqref{2L13}, we have
		\begin{align}
			\nonumber&\ \ \sup_{\bm x\in\lg}\sum_{\bm y\in\lg}|\boL(\bm x,\bm y)|\\
		\label{2L1}	\le&\ \  \delta_{s}^{-3}\sup_{\{\bm k\in P_{s+1}:\ \tO_{\bm k}^{s+1}\subset\lg\}}\left(\|\theta+\bm k\cdot\bm\omega-\theta_{s+1}\|_{\T}^{-1}\cdot\|\theta+\bm k\cdot\bm\omega+\theta_{s+1}\|_{\T}^{-1}\right).
		\end{align}
		
		Now, we estimate $\sup_{\bm y\in\lg}\sum_{\bm x\in\lg}|\boL(\bm x,\bm y)|$. We again divide the discussion into two cases:
		\begin{itemize}
			\item[\textbf{Case 1}:] $\bm y\in\lg\setminus\bigcup_{\bm k\in\widetilde{P}_1}\tO_{\bm k}^{1}$. In this case, 
			\begin{align*}
				\bm y\in O(\bm x)\text{ iff }\bm x\in \lg\cap\lg_{\frac{1}{2}N_1}(\bm y).
			\end{align*}
			 At this time, if $\bm x\in \lg\cap\lg_{\frac{1}{2}N_1}(\bm y)$, then $\bm x\in\lg\setminus\bigcup_{\bm k\in\widetilde{P}_1}2\Omega_{\bm k}^{1}$ ($O(\bm x)=\lg\cap\lg_{\frac{1}{2}N_1}(\bm x)$ is $0$-good). Hence by \eqref{indpa}, \eqref{0gl2}, we get
			\begin{align}
				\nonumber\sum_{\bm x\in\lg}|\boL(\bm x,\bm y)|=&\ \sum_{\bm x\in \lg\cap\lg_{\frac{1}{2}N_1}(\bm y)}|\boT_{O(\bm x)}^{-1}(\bm x,\bm y)|\\
			\label{2L21}	\le&\  (\# (\lg\cap\lg_{\frac{1}{2}N_1}(\bm y)))\cdot2\kappa_1^{-1}\delta_0^{-2}<\delta_0^{-\frac{5}{2}}.
			\end{align}
			\item[\textbf{Case 2}:] $\bm y\in \tO_{\bm k}^{t}$ for some $\bm k\in\widetilde{P}_t\setminus Q_t$. In this case, 
			\begin{align*}
				\bm y\in O(\bm x)\text{ iff }\bm x\in 2\Omega_{\bm k}^{t}\bigcup\left(\lg\cap\lg_{\frac{1}{2}N_1}(\bm y)\right).
			\end{align*}
			Therefore, from \eqref{indpa}, \eqref{0gl2} and \eqref{PsQsnorm}, we obtain
			\begin{align}
				\nonumber\sum_{\bm x\in\lg}	|\boL(\bm x,\bm y)|&\le \sum_{\bm x\in2\Omega_{\bm k}^t}|\boT_{O(\bm x)}^{-1}(\bm x,\bm y)|+\sum_{\bm x\in \left(\lg\cap\lg_{\frac{1}{2}N_1}(\bm y)\right)\setminus(2\Omega_{\bm k}^t)}|\boT_{O(\bm x)}^{-1}(\bm x,\bm y)|\\
				\nonumber&\le (\#(2\Omega_{\bm k}^{t}))\delta_{t-1}^{-2}\|\theta+\bm k\cdot\bm\omega-\theta_t\|_{\T}^{-1}\cdot\|\theta+\bm k\cdot\bm\omega+\theta_t\|_{\T}^{-1}\\
				\nonumber&\ \ +(\# (\lg\cap\lg_{\frac{1}{2}N_1}(\bm y)))\cdot2\kappa_1^{-1}\delta_0^{-2}\\
				\label{2L22}&\le \delta_{t-1}^{-3}\|\theta+\bm k\cdot\bm\omega-\theta_t\|_{\T}^{-1}\cdot\|\theta+\bm k\cdot\bm\omega+\theta_t\|_{\T}^{-1}.
			\end{align}
		\end{itemize}
		Combining \eqref{2L13}, \eqref{2L21} and \eqref{2L22} gives
		\begin{align}
		\nonumber &\ \ \sup_{\bm y\in\lg}\sum_{\bm x\in\lg}|\boL(\bm x,\bm y)|\\
		\label{2L2}\le&\ \  \delta_{s}^{-3}\sup_{\{\bm k\in P_{s+1}:\ \tO_{\bm k}^{s+1}\subset\lg\}}\left(\|\theta+\bm k\cdot\bm\omega-\theta_{s+1}\|_{\T}^{-1}\cdot\|\theta+\bm k\cdot\bm\omega+\theta_{s+1}\|_{\T}^{-1}\right).
		\end{align}
		Then, by \eqref{2L1}, \eqref{2L2} and Schur's test we know
		\begin{align}\label{2Lnorm1}
			\|\boL\|\le\delta_{s}^{-3}\sup_{\{\bm k\in P_{s+1}:\ \tO_{\bm k}^{s+1}\subset\lg\}}\left(\|\theta+\bm k\cdot\bm\omega-\theta_{s+1}\|_{\T}^{-1}\cdot\|\theta+\bm k\cdot\bm\omega+\theta_{s+1}\|_{\T}^{-1}\right).
		\end{align}
		Finally, from \eqref{s+1gdef}, \eqref{2I+eK-1}, \eqref{2tlg-1} and \eqref{2Lnorm1}, we prove
		\begin{align*}
			\|\boT_{\lg}^{-1}\|\le &\ 2\delta_{s}^{-3}\sup_{\{\bm k\in P_{s+1}:\ \tO_{\bm k}^{s+1}\subset\lg\}}\left(\|\theta+\bm k\cdot\bm\omega-\theta_{s+1}\|_{\T}^{-1}\cdot\|\theta+\bm k\cdot\bm\omega+\theta_{s+1}\|_{\T}^{-1}\right)\\
			<&\ \delta_{s+1}^{-3}.
		\end{align*}
		This completes the proof of \eqref{s+1gnorm}.
		
		We initiate the proof of the off-diagonal decay estimate stated in \eqref{s+1gdecay}. Throughout this argument, we maintain the standing assumption that
		\begin{align*}
			\|\bm x-\bm y\|\ge10\tz_{s+1}.
		\end{align*}
		The proof will proceed through an iterative scheme, leveraging the following key estimates \eqref{ite0g0} and \eqref{itetgeq1}. For $\bm z\in\lg$, from Corollary \ref{ite0g} and Lemma \ref{itepsqt}, there exists $\bm z'\in\lg\setminus O(\bm z)$ (cf. \eqref{neighbor2}) such that
		\begin{align}\label{2iteeq1}
			|\boT_{\lg}^{-1}(\bm z,\bm y)|\le e^{-\alpha''_s\log^{\rho}(1+\|\bm z-\bm z'\|)}\cdot|\boT_{\lg}^{-1}(\bm z',\bm y)|,
		\end{align}
		where $\alpha''_s=\alpha_s\left(1-\frac{30\times 10^{5\rho'}}{\alpha\log^{\rho-\rho'}N_{s+1}}\right)$.
		
		Next, iterating \eqref{2iteeq1} for $m\ge2$ steps leads to the following: there exist $\bm x_1,\bm x_2,\cdots,\bm x_m\in\lg$ such that
		\begin{align}\label{2iteeq2}
			|\boT_{\lg}^{-1}(\bm x,\bm y)|\le e^{-\alpha''_s\left(\sum_{k=0}^{m-1}\log^{\rho}(1+\|\bm x_{k+1}-\bm x_{k}\|)\right)}	|\boT_{\lg}^{-1}(\bm x_m,\bm y)|,
		\end{align}
		where $\bm x_0=\bm x$ and $\bm x_{k+1}=\bm x'$, $k=0,1,\cdots,m-1$.  We define $n\ge1$ to be the smallest integer so that $\bm x_n\in O(\bm y)$ (cf. \eqref{neighbor2}). We then have $\bm x_i\notin O(\bm y)$ for $i=0,1,\cdots,n-1$. We divide the discussion into two cases:
		\begin{itemize}
			\item[\textbf{Case 1}:] $n\le \frac{\alpha\log^{\rho}(1+\|\bm x-\bm y\|)+3\times 10^{5\rho'}\log(N_{s+1}+1)}{\frac{\alpha}{2}\log^{\rho}\left(\frac{N_1}{2}+1\right)}$. Using Lemma \ref{qua} and the triangle inequality implies
			\begin{align}
				\nonumber &\ \sum_{k=0}^{n-1}\log^{\rho}(1+\|\bm x_{k+1}-\bm x_{k}\|)\\
				\nonumber\ge&\ \log^{\rho}\left(1+\sum_{k=0}^{n-1}\|\bm x_{k+1}-\bm x_k\|\right)-C(\rho)\log^{\rho}n\\
				\label{60203}\ge&\ \log^{\rho}(1+\|\bm x_n-\bm x\|)-C(\rho)\log^{\rho}n.
			\end{align}
			Since $\bm x_n\in O(\bm y)$ and $\diam(O(\bm y))\le \tz_{s+1}$, we have
			\begin{align*}
				\|\bm x-\bm x_n\|\ge\|\bm x-\bm y\|-\|\bm x_n-\bm y\|\l\ge\|\bm x-\bm y\|-\tz_{s+1}.
			\end{align*}
			By $\|\bm x-\bm y\|\ge10\tz_{s+1}\ge N_{s+1}\gg1$ and $1<\rho'<\rho<\rho'+1$, applying Lemma \ref{exl} with $x=\|\bm x-\bm y\|$ and $y=\tz_{s+1}$ gives
			\begin{align}
			\nonumber&\ \log^{\rho}(1+\|\bm x_n-\bm x\|)\\
				\nonumber	\ge&\ \log^{\rho}(1+\|\bm x-\bm y\|-\tz_{s+1})\\
				\nonumber\ge&\ \left(1-2\rho\frac{\tz_{s+1}}{(1+\|\bm x-\bm y\|)\log(1+\|\bm x-\bm y\|)}\right)\log^{\rho}(1+\|\bm x-\bm y\|)\\
				\label{60204}\ge&\ \left(1-\frac{ 10^{5\rho'}}{\alpha\log^{\rho-\rho'} N_{s+1}}\right)\log^{\rho}(1+\|\bm x-\bm y\|).
			\end{align}
			Under the same conditions, we have
			\begin{align}
				\nonumber	C(\rho)\log^{\rho}n&\le C(\rho)\log^{\rho}(2\alpha\log^{\rho}(1+\|\bm x-\bm y\|))\le \log(1+\|\bm x-\bm y\|)\\
				\label{60205}	&\le \frac{10^{5\rho'}\log^{\rho}(1+\|\bm x-\bm y\|)}{\alpha\log^{\rho-\rho'}N_{s+1}}.
			\end{align}
			Combining \eqref{60203}, \eqref{60204} and \eqref{60205} shows
			\begin{align}\label{60206}
				\sum_{k=0}^{n-1}\log^{\rho}(1+\|\bm x_{k+1}-\bm x_{k}\|)\ge \left(1-\frac{2\times 10^{5\rho'}}{\alpha\log^{\rho-\rho'} N_{s+1}}\right)\log^{\rho}(1+\|\bm x-\bm y\|).
			\end{align}
			From $\|\bm x-\bm y\|\ge10\tz_{s+1}\ge N_{s+1}\gg1$, \eqref{indpa}, \eqref{tsgnorm}, \eqref{2iteeq2} and  \eqref{60206}, we have
			\begin{align*}
				|\boT_{\lg}^{-1}(\bm x,\bm y)|&\le e^{-\alpha''_s \left(1-\frac{2\times 10^{5\rho'}}{\alpha\log^{\rho-\rho'} N_{s+1}}\right)\log^{\rho}(1+\|\bm x-\bm y\|)}	\delta_{s+1}^{-3}\\
				&\le e^{-\alpha_{s+1}\log^{\rho}(1+\|\bm x-\bm y\|)}.
			\end{align*}

			\item[\textbf{Case 2}:] $n>\frac{\alpha\log^{\rho}(1+\|\bm x-\bm y\|)+3\times 10^{5\rho'}\log(N_{s+1}+1)}{\frac{\alpha}{2}\log^{\rho}\left(\frac{N_1}{2}+1\right)}$. In this case, from \eqref{tsgnorm}, \eqref{2iteeq2}, $\|\bm x_{k+1}-\bm x_k\|\ge\frac{N_1}{2}$ (since $\bm x_{k+1}\in\lg\setminus O(\bm x_k)$), $k=0,1\cdots,n-1$ and
			\begin{align*}
				\alpha''_s=\alpha_s\left(1-\frac{30\times 10^{5\rho'}}{\alpha\log^{\rho-\rho'}N_{s+1}}\right)>\frac{\alpha}{2},
			\end{align*}
			it follows that
			\begin{align*}
				|\boT_{\lg}^{-1}(\bm x,\bm y)|&\le e^{-\alpha''_s\left(\sum_{k=0}^{n-1}\log^{\rho}(1+\|\bm x_{k+1}-\bm x_{k}\|)\right)}	|\boT_{\lg}^{-1}(\bm x_n,\bm y)|\\
				&\le e^{-n\cdot\alpha''_s\log^{\rho}\left(\frac{N_1}{2}+1\right)}e^{3\times 10^{5\rho'}\log^{\rho}(N_{s+1}+1)}\\
				&\le e^{-\alpha\log^{\rho}(1+\|\bm x-\bm y\|)}\le e^{-\alpha_{s+1}\log^{\rho}(1+\|\bm x-\bm y\|)}.
			\end{align*}
		\end{itemize}
		This concludes the proof of Lemma \ref{s+1g}.
	\end{proof}

\section{Arithmetic spectral localization}\label{Arith}
	In this section, we will prove  Theorem \ref{asl} by combining  Green's function estimates and the Shnol's theorem.
\begin{proof}[Proof of Theorem \ref{asl}]
	%We prove that for $0<|\ep|\le\ep_0$, $\theta\in\Theta_{\tau_1,\g_1}$, $\mathcal{H}(\theta)$ has the only pure point spectrum with polynomially decaying eigenfunctions. 
	Let $\ep_0$ be given by Theorem \ref{ge}. Fix $\theta\in\T\setminus\Theta_{\tau}$. Let $E\in\sigma(\mathcal{H}(\theta))$ be a generalized eigenvalue of $\mathcal{H}(\theta)$ and $\psi=\{\psi(\bm n)\}_{\bm n\in\Z^d}\ne0$ be the corresponding generalized eigenfunction satisfying 
	\begin{align*}
		|\psi(\bm n)|\le(1+\|\bm n\|)^d.
	\end{align*}
	From Shnol's theorem  (cf. \cite{Han19}), it suffices to show that 
	\begin{align*}
			|\psi(\bm n)|<e^{-\frac{\alpha}{4\cdot 10^{6\rho}}\log^{\rho}\left(1+\|\bm n\|\right)} \text{ for $\|\bm n\|\gg1$}
	\end{align*}
	and $\psi\in\ell^2(\Z^d)$. For this purpose, note first that there exists some $\tilde{s}\in\N$ such that
	\begin{align}\label{4.1}
		\|2\theta+\bm n\cdot\bm\omega\|_{\T}>\frac{1}{\|n\|^{\tau}}\ \text{for all $n$ satisfying $\|n\|\ge N_{\tilde{s}}$}.
	\end{align}
	We claim that there exists some  $s_0>0$ such that for $s\ge s_0$,
	\begin{align}\label{4.2}
		\lg_{2N_s^{10^4}}\cap\left(\bigcup_{\bm k\in Q_s}\tilde{\Omega}_{\bm k}^{s}\right)\ne\emptyset.
	\end{align}
	Otherwise, there exists a  sequence $s_i\rightarrow+\infty$ (as $i\rightarrow\infty$) such that
	\begin{align}\label{4.3}
		\lg_{2N_{s_i}^{10^4}}\cap\left(\bigcup_{\bm k\in Q_{s_i}}\tilde{\Omega}_{\bm k}^{s_i}\right)=\emptyset.
	\end{align}
	Then we can enlarge $\tilde{\lg}_{N_{s_i}^{10^4}}$ to $\tilde{\lg}_i$ so that 
	\begin{align*}
		\lg_{N_{s_i}^{10^4}}\subset\tilde{\lg}_i\subset\lg_{N_{s_i}^{10^4}+50N_{s_i}^{100}}
	\end{align*}
	and
	\begin{align*}
		\tilde{\lg}_i\cap\tilde{\Omega}_{\bm k}^{s'}\ne\emptyset\Rightarrow\tilde{\Omega}_{\bm k}^{s'}\subset\tilde{\lg}_i\ \text{for $s'\le s_i$ and $\bm k\in P_{s'}$}.
	\end{align*}
	From \eqref{4.3}, we have
	\begin{align*}
		\tilde{\lg}_i\cap\left(\bigcup_{\bm k\in Q_{s_i}}\tilde{\Omega}_{\bm k}^{s_i}\right)=\emptyset,
	\end{align*}
	which shows that $\tilde{\lg}_i$ is $s_i$-good. Let $\tilde{\lg}_{i,0}=\lg_{\frac{1}{2}N_{s_i}^{10^4}}\cap\tilde{\lg}_i$. Using Poisson's identity  yields for $\bm n\in\lg_{N_{s_i}}$, 
	\begin{align*}
		|\psi(\bm n)|\le&\ \sum_{\bm n'\in\tilde{\lg}_i,\bm n''\notin\tilde{\lg}_i}|\mathcal{T}_{\tilde{\lg}_i}^{-1}(\bm n,\bm n')|\cdot|\mathcal{W}(\bm n',\bm n'')|\cdot|\psi(\bm n'')|\\
		\le &\ ({\rm I})+({\rm II}),
	\end{align*}
	where
	\begin{align*}
		({\rm I})&=\sum_{\bm n'\in\tilde{\lg}_{i,0},\bm n''\notin\tilde{\lg}_i}|\mathcal{T}_{\tilde{\lg}_i}^{-1}(\bm n,\bm n')|\cdot|\mathcal{W}(\bm n',\bm n'')|\cdot|\psi(\bm n'')|,\\
		({\rm II})&=\sum_{\bm n'\in\tilde{\lg}_i\setminus\tilde{\lg}_{i,0},\bm n''\notin\tilde{\lg}_i}|\mathcal{T}_{\tilde{\lg}_i}^{-1}(\bm n,\bm n')|\cdot|\mathcal{W}(\bm n',\bm n'')|\cdot|\psi(\bm n'')|.
	\end{align*}
	For $({\rm I})$, we have by Theorem \ref{ge}, \eqref{wphi}, and $|\psi(\bm n)|\le(1+\|\bm n\|)^d$ that 
\begin{align*}
		({\rm I})&\le \sum_{\bm n'\in\tilde{\lg}_{i,0},\bm n''\notin\tilde{\lg}_i}\|\mathcal{T}_{\tilde{\lg}_i}^{-1}\|\cdot e^{-\alpha\log^{\rho}(1+\|\bm n'-\bm n''\|)}\cdot(1+\|\bm n''\|)^d\\
		&\le \delta_{s_i}^{-3}\sum_{\bm n'\in\tilde{\lg}_{i,0},\bm n''\notin\tilde{\lg}_i} e^{-\alpha\log^{\rho}(1+\|\bm n'-\bm n''\|)}\cdot(1+\|\bm n'\|)^d(1+\|\bm n'-\bm n''\|)^d\\
		&\le \delta_{s_i}^{-3}(1+N_{s_i}^{10^4})^d\sum_{\bm n'\in\tilde{\lg}_{i,0},\bm n''\notin\tilde{\lg}_i} e^{-\alpha\log^{\rho}(1+\|\bm n'-\bm n''\|)}\cdot(1+\|\bm n'-\bm n''\|)^d.
\end{align*}
Since $\|\bm n'-\bm n''\|\ge\frac{1}{2}N_{s_i}^{10^4}\gg1$ and \eqref{Det} we get
\begin{align*}
	&\ \ \sum_{\bm n'\in\tilde{\lg}_{i,0},\bm n''\notin\tilde{\lg}_i} e^{-\alpha\log^{\rho}(1+\|\bm n'-\bm n''\|)}\cdot(1+\|\bm n'-\bm n''\|)^d\\
	\le&\ \  e^{-\frac{1}{2}\alpha\log^{\rho}\left(1+\frac{1}{2}N_{s_i}^{10^4}\right)}\sum_{\bm n'\in\tilde{\lg}_{i,0},\bm n''\notin\tilde{\lg}_i} e^{-\frac{1}{2}\alpha\log^{\rho}(1+\|\bm n'-\bm n''\|)}\cdot(1+\|\bm n'-\bm n''\|)^d\\
	\le &\ \ e^{-\frac{1}{2}\alpha\log^{\rho}\left(1+\frac{1}{2}N_{s_i}^{10^4}\right)}\sum_{\bm n'\in\tilde{\lg}_{i,0},\bm n''\notin\tilde{\lg}_i} e^{-\frac{1}{10}\alpha\log^{\rho}(1+\|\bm n'-\bm n''\|)}\\
	\le &\ \ e^{-\frac{1}{2}\alpha\log^{\rho}\left(1+\frac{1}{2}N_{s_i}^{10^4}\right)}\cdot(\#\tilde{\lg}_{i,0})\cdot D\left(\frac{\alpha}{10}\right).
\end{align*}
	Therefore, 
\begin{align*}
	({\rm I})\le e^{-\frac{1}{2}\alpha\log^{\rho}\left(1+\frac{1}{2}N_{s_i}^{10^4}\right)}\cdot\delta_{s_i}^{-3}\cdot(1+N_{s_i}^{10^4})^{2d}\cdot D\left(\frac{\alpha}{10}\right)\rightarrow0\ \text{as $i\rightarrow\infty$}.
\end{align*}
	For $({\rm II})$, we also have by Theorem \ref{ge}, $\|\bm n-\bm n'\|\ge\frac{1}{4}N_{s_i}^{10^4}\gg10\tz_{s_i}$, $\|\bm n-\bm n''\|\ge \frac{1}{2}N_{s_i}^{10^4}\gg1$, \eqref{wphi}, \eqref{quaeq} and $|\psi(\bm n)|\le(1+\|\bm n\|)^d$ that 
	\begin{align*}
			({\rm II})\le&\ \sum_{\bm n'\in\tilde{\lg}_i\setminus\tilde{\lg}_{i,0},\bm n''\notin\tilde{\lg}_i}e^{-\frac{1}{2}\alpha\log^{\rho}(1+\|\bm n-\bm n'\|)}\cdot e^{-\alpha\log^{\rho}(1+\|\bm n'-\bm n''\|)}\cdot(1+\|\bm n''\|)^d\\
			\le&\ \sum_{\bm n'\in\tilde{\lg}_i\setminus\tilde{\lg}_{i,0},\bm n''\notin\tilde{\lg}_i}\left(e^{-\frac{1}{2}\alpha\log^{\rho}(1+\|\bm n-\bm n''\|)+\frac{1}{2}\alpha C(\rho)\log^{\rho}2}\right.\\
			&\ \ \left.\cdot e^{-\frac{1}{2}\alpha\log^{\rho}(1+\|\bm n'-\bm n''\|)}\cdot (1+\|\bm n\|)^d(1+\|\bm n-\bm n''\|)^d\right)\\
			\le &\ (1+N_{s_i})^de^{-\frac{1}{4}\alpha\log^{\rho}\left(1+\frac{1}{2}N_{s_i}^{10^4}\right)} \sum_{\bm n'\in\tilde{\lg}_i\setminus\tilde{\lg}_{i,0},\bm n''\notin\tilde{\lg}_i} e^{-\frac{1}{2}\alpha\log^{\rho}(1+\|\bm n'-\bm n''\|)}\\
			\le &\ (1+N_{s_i})^d(1+2N_{s_i}^{10^4})^d e^{-\frac{1}{4}\alpha\log^{\rho}\left(1+\frac{1}{2}N_{s_i}^{10^4}\right)}\cdot D\left(\frac{\alpha}{2}\right)\rightarrow0\ \text{as $i\rightarrow\infty$}.
	\end{align*}
	It follows that $\psi(\bm n)=0\ {\rm for}\ \forall \bm n\in\Z^d$, which  contradicts $\psi\ne0$.  The Claim is proved.
	
	Next define
	\begin{align*}
		U_s=\lg_{8N_{s+1}^{10^4}}\setminus\lg_{4N_s^{10^4}},\ U_s^{*}=\lg_{10N_{s+1}^{10^4}}\setminus\lg_{3N_s^{10^4}}.
	\end{align*}
	We can also enlarge $U_s^{*}$ to $\tilde{U}_s^*$ so that
	\begin{align*}
		U_s^*\subset \tilde{U}_s^*\subset\lg_{50N_s^{100}}(U_s^*)
	\end{align*}
	and
	\begin{align*}
		\tilde{U}_s^*\cap\tilde{\Omega}_{\bm k}^{s'}\ne\emptyset\Rightarrow\tilde{\Omega}_{\bm k}^{s'}\subset\tilde{U}_s^*\ \text{for $s'\le s$ and $\bm k\in P_{s'}$}.
	\end{align*}
	Let $\bm n$ satisfy $\|\bm n\|>\max(4N_{\tilde{s}}^{10^4},4N_{s_0}^{10^4})$. Then there exists some $s\ge\max(\tilde{s},s_0)$ such that
	\begin{align}\label{4.4}
		\bm n\in U_s\subset \tilde{U}_s^*.
	\end{align}
	By virtue of condition \eqref{4.2} and without loss of generality, we may restrict our consideration to the case where there exists a $\bm k\in Q_s^+$ such that
	\begin{align*}
		\lg_{2N_s^{10^4}}\cap\tilde{\Omega}_{\bm k}^s\ne\emptyset.
	\end{align*}
	Then for $\bm k\ne\bm k'\in Q_s^+$, we have
	\begin{align*}
		\|\bm k-\bm k'\|>\left(\frac{\g}{2\delta_s}\right)^{\frac{1}{\tau}}\gg N_{s+1}^{10^5}\gg\diam (\tilde{U}_s^*).
	\end{align*}
	Therefore,
	\begin{align*}
		\tilde{U}_s^*\cap\left(\bigcup_{\bm l\in Q_s^+}\tilde{\Omega}_{\bm l}^s\right)=\emptyset.
	\end{align*}
	Now, if there exists some $\bm l\in Q_s^-$ such that
	\begin{align*}
		\tilde{U}_s^*\cap\tilde{\Omega}_{\bm l}^s\ne\emptyset,
	\end{align*}
	then
	\begin{align*}
		 2N_s^{10^4}-100N_s^{100}\le\|\bm l\|-\|\bm k\|\le\|\bm l+\bm k\|\le\|\bm l\|+\|\bm k\|<11N_{s+1}^{10^4}.
	\end{align*}
	Recalling
	\begin{align*}
		Q_s\subset P_s\subset \Z^d+\frac{1}{2}\sum_{i=0}^{s-1}\bm l_i,
	\end{align*}
	we have $\bm l+\bm k\in\Z^d$. According to \eqref{4.1}, we obtain 
	\begin{align*}
		\frac{1}{(11N_{s+1}^{10^4})^{\tau}}<\|2\theta+(\bm l+\bm k)\cdot\bm\omega\|_{\T}\le\|\theta+\bm l\cdot\bm\omega-\theta_s\|_{\T}+\|\theta+\bm k\cdot\bm\omega+\theta_s\|_{\T}<2\delta_s,
	\end{align*}
	which contradicts
	\begin{align*}
		\delta_s^{-1}\gg N_{s+1}^{10^4\tau}.
	\end{align*}
	We thus have shown
	\begin{align*}
		\tilde{U}_s^*\cap\left(\bigcup_{\bm l\in Q_s}\tilde{\Omega}_{\bm l}^s\right)=\emptyset.
	\end{align*}
	This implies that $\tilde{U}_s^*$ is $s$-$\good$. 
	
	Finally, by recalling \eqref{4.4}, we  can set
	\begin{align*}
		\hat{U}_s=\lg_{\frac{1}{2}N_s^{10^4}}(U_s).
	\end{align*}
	Then
	\begin{align*}
		|\psi(\bm n)|\le&\sum_{\bm n'\in\tilde{U}_s^*,\bm n''\notin\tilde{U}_s^*}|\mathcal{T}_{\tilde{U}_s^*}^{-1}(\bm n,\bm n')|\cdot|\mathcal{W}(\bm n',\bm n'')|\cdot|\psi(\bm n'')|\\
		\le &({\rm III})+({\rm IV}),
	\end{align*}
	where
	\begin{align*}
		({\rm III})&=\sum_{\bm n'\in\hat{U}_s,\bm n''\notin\tilde{U}_s^*}|\mathcal{T}_{\tilde{U}_s^*}^{-1}(\bm n,\bm n')|\cdot|\mathcal{W}(\bm n',\bm n'')|\cdot|\psi(\bm n'')|,\\
		({\rm IV})&=\sum_{\bm n'\in\tilde{U}_s^*\setminus\hat{U}_s,\bm n''\notin\tilde{U}_s^*}|\mathcal{T}_{\tilde{U}_s^*}^{-1}(\bm n,\bm n')|\cdot|\mathcal{W}(\bm n',\bm n'')|\cdot|\psi(\bm n'')|.
	\end{align*}
	For $({\rm III})$, we have by Theorem \ref{ge}, \eqref{wphi}, and $|\psi(\bm n)|\le(1+\|\bm n\|)^d$ that 
	\begin{align*}
		({\rm III})\le&\  \sum_{\bm n'\in\hat{U}_s,\bm n''\notin\tilde{U}_s^*} \|\boT_{\tilde{U}_s^*}^{-1}\|\cdot e^{-\alpha\log^{\rho}(1+\|\bm n'-\bm n''\|)}\cdot(1+\|\bm n''\|)^d\\
		\le &\ \delta_s^{-3}\sum_{\bm n'\in\hat{U}_s,\bm n''\notin\tilde{U}_s^*}e^{-\alpha\log^{\rho}(1+\|\bm n'-\bm n''\|)}(1+\|\bm n'\|)^d(1+\|\bm n'-\bm n''\|)^d\\
		\le &\ \delta_s^{-3}(1+9N_{s+1}^{10^4})^d\sum_{\bm n'\in\hat{U}_s,\bm n''\notin\tilde{U}_s^*}e^{-\alpha\log^{\rho}(1+\|\bm n'-\bm n''\|)}(1+\|\bm n'-\bm n''\|)^d.
	\end{align*}
	Since $\|\bm n'-\bm n''\|\ge\frac{1}{4}N_{s}^{10^4}\gg1$ and \eqref{Det} we get
	\begin{align*}
		&\ \ \sum_{\bm n'\in\hat{U}_s,\bm n''\notin\tilde{U}_s^*}e^{-\alpha\log^{\rho}(1+\|\bm n'-\bm n''\|)}(1+\|\bm n'-\bm n''\|)^d\\
		\le&\ \  e^{-\frac{1}{2}\alpha\log^{\rho}\left(1+\frac{1}{4}N_{s_i}^{10^4}\right)}\sum_{\bm n'\in\hat{U}_s,\bm n''\notin\tilde{U}_s^*} e^{-\frac{1}{2}\alpha\log^{\rho}(1+\|\bm n'-\bm n''\|)}\cdot(1+\|\bm n'-\bm n''\|)^d\\
		\le &\ \ e^{-\frac{1}{2}\alpha\log^{\rho}\left(1+\frac{1}{4}N_{s_i}^{10^4}\right)}\sum_{\bm n'\in\hat{U}_s,\bm n''\notin\tilde{U}_s^*} e^{-\frac{1}{10}\alpha\log^{\rho}(1+\|\bm n'-\bm n''\|)}\\
		\le &\ \ e^{-\frac{1}{2}\alpha\log^{\rho}\left(1+\frac{1}{2}N_{s_i}^{10^4}\right)}\cdot(\#\hat{U}_s)\cdot D\left(\frac{\alpha}{10}\right).
	\end{align*}
	Therefore, 
	\begin{align*}
		({\rm III})\le&\  e^{-\frac{1}{2}\alpha\log^{\rho}\left(1+\frac{1}{4}N_{s}^{10^4}\right)}\cdot\delta_{s}^{-3}\cdot(1+9N_{s+1}^{10^4})^{2d}\cdot D\left(\frac{\alpha}{10}\right)\\
		\le&\ \frac{1}{2}e^{-\frac{1}{4}\alpha\log^{\rho}\left(1+\frac{1}{4}N_{s}^{10^4}\right)}.
	\end{align*}
		For $({\rm IV})$, we also have by Theorem \ref{ge}, $\|\bm n-\bm n'\|\ge\frac{1}{4}N_{s}^{10^4}\gg10\tz_{s}$, $\|\bm n-\bm n''\|\ge \frac{1}{2}N_{s}^{10^4}\gg1$, \eqref{wphi}, \eqref{quaeq} and $|\psi(\bm n)|\le(1+\|\bm n\|)^d$ that 
	\begin{align*}
		({\rm IV})\le& \sum_{\bm n'\in\tilde{U}_s^*\setminus\hat{U}_s,\bm n''\notin\tilde{U}_s^*}e^{-\frac{1}{2}\alpha\log^{\rho}(1+\|\bm n-\bm n'\|)}\cdot e^{-\alpha\log^{\rho}(1+\|\bm n'-\bm n''\|)}\cdot(1+\|\bm n''\|)^d\\
		\le& \sum_{\bm n'\in\tilde{U}_s^*\setminus\hat{U}_s,\bm n''\notin\tilde{U}_s^*}\left(e^{-\frac{1}{2}\alpha\log^{\rho}(1+\|\bm n-\bm n''\|)+\frac{1}{2}\alpha C(\rho)\log^{\rho}2}\right.\\
		&\ \ \left.\cdot e^{-\frac{1}{2}\alpha\log^{\rho}(1+\|\bm n'-\bm n''\|)}\cdot (1+\|\bm n\|)^d(1+\|\bm n-\bm n''\|)^d\right)\\
		\le &(1+9N_{s+1}^{10^4})^de^{-\frac{1}{3}\alpha\log^{\rho}\left(1+\frac{1}{2}N_{s_i}^{10^4}\right)} \sum_{\bm n'\in\tilde{U}_s^*\setminus\hat{U}_s,\bm n''\notin\tilde{U}_s^*} e^{-\frac{1}{2}\alpha\log^{\rho}(1+\|\bm n'-\bm n''\|)}\\
		\le &(1+9N_{s+1}^{10^4})^d(1+11N_{s+1}^{10^4})^d e^{-\frac{1}{3}\alpha\log^{\rho}\left(1+\frac{1}{2}N_{s_i}^{10^4}\right)}\cdot D\left(\frac{\alpha}{2}\right)\\
		\le &\frac{1}{2}e^{-\frac{1}{4}\alpha\log^{\rho}\left(1+\frac{1}{4}N_{s}^{10^4}\right)}.
	\end{align*}
Hence,
\begin{align*}
		|\psi(\bm n)|\le e^{-\frac{1}{4}\alpha\log^{\rho}\left(1+\frac{1}{4}N_{s}^{10^4}\right)}.
\end{align*}
	Combining the above estimates and $\|\bm n\|\le 8N_{s+1}^{10^4}$, we have
	\begin{align*}
		|\psi(\bm n)|<e^{-\frac{\alpha}{4\cdot 10^{6\rho}}\log^{\rho}\left(1+\|\bm n\|\right)}.
	\end{align*}
	We complete the proof of arithmetic spectral localization. 
\end{proof}

	\section{Sub-polynomial bounds of moments}\label{subpoly}
	For convenience, we will divide this proof into several parts.
	
	\subsection{Basic assumptions}
There are some assumption used throughout the proof:
\begin{itemize}
	\item[(1)] $1<\frac{1+\rho'}{2}<\rho'<\rho<\rho'+1$.
\item[(2)] $\boH(\theta)$ is self-adjoint for all $\theta\in\T$, which implies $v$ is real on $\T$.
\item[(3)] $v(\T)=[a,b]$, since $v$ is continuous on $\T$.
\item[(4)] There is some $\beta>0$ such that $\{z\in\C:\ \Re z\in[a-2\beta,b+2\beta],\ \Im z\in[-\beta,\beta]\}\subset v(\D_{R/2})$ by using open mapping theorem and $v$ is analytic on $\D_R$.
\item[(5)] $\ep$ is sufficiently small such that $\sigma(\boH(\theta))\subset[a-\beta,b+\beta]$ for all $\theta\in\T$ and $|\ep|\cdot2e^{\frac{19}{20}\alpha C(\rho)\log^{\rho}2}\cdot D\left(\frac{\alpha}{20}\right)<\frac{\beta}{2}$, where $D(\alpha)$ is defined in \eqref{Det}.
\item[(6)] $E\in\R$.
\end{itemize}

\subsection{Connection between moments and Green's functions}\label{92}
In this part, we will establish a connection between moments and Green's functions. We have

\begin{lem}\label{cbmgf}
	For all $ \theta\in \mathbb{T},t, {T}> 0,$ and $\bm n\in\mathbb{Z}^{d}$, we obtain
		\begin{align}\label{cbpmgf}
			\nonumber&\ |\ji e^{-\sqrt{-1}t\mathcal{H}(\theta)}\delta_{\bm 0},\delta_{\bm n}\jd|^2\\
		\nonumber\le&\  \frac{(b-a+4\beta)e^2}{2\pi^2}\int_{a-2\beta}^{b+2\beta}\left|\left\ji\left(\mathcal{H}(\theta)-E-\sqrt{-1}t^{-1}\right)^{-1}\delta_{\bm 0},\delta_{\bm n}\right\jd\right|^2dE\\
	&\ +\frac{2e^2}{\beta^2\pi^2}(b-a+6\beta+2t^{-1})^2e^{-\frac{9}{5}\alpha\log^{\rho}(1+\|\bm n\|)}.
	\end{align}
	and
	\begin{align}\label{cbtapmgf}
		\nonumber&\ \frac{2}{T}\int^{\infty}_{0}e^{-2t/T}|\langle e^{-\sqrt{-1}t\boH(\theta)}\delta_{\bm 0},\delta_{\bm n}\rangle|^{2}\mathrm{d}t\\
	\leq&\ \frac{1}{T\pi}\int_{a-2\beta}^{b+2\beta}|\langle (\boH(\theta)-E-\sqrt{-1}t)^{-1}\delta_{\bm 0},\delta_{\bm n}\rangle|^{2}dE+\frac{4}{\beta T\pi}e^{-\frac{9}{5}\alpha\log^{\rho}(1+\|\bm n\|)}.
	\end{align}
\end{lem}
	\begin{proof}
	We first deal with \eqref{cbpmgf}. Given $t>0$, we will consider the following contour $\mathscr{C}=\mathscr{C}_1\cup\mathscr{C}_2\cup\mathscr{C}_3\cup\mathscr{C}_4$, where
	\begin{align*}
		\mathscr{C}_1&=\{z=E+\sqrt{-1}y:\ E\in[a-2\beta,b+2\beta],\ y=t^{-1}\},\\
		\mathscr{C}_2&=\{z=E+\sqrt{-1}y:\ E=a-2\beta,\ y\in[-\beta,t^{-1}]\},\\
		\mathscr{C}_3&=\{z=E+\sqrt{-1}y:\ E\in[a-2\beta,b+2\beta],\ y=-\beta\},\\
		\mathscr{C}_4&=\{z=E+\sqrt{-1}y:\ E=b+2\beta,\ y\in[-\beta,t^{-1}]\}.
	\end{align*}
	By Cauchy integral formula, we can get
	\begin{align*}
		\ji e^{-\sqrt{-1}t\mathcal{H}(\theta)}\delta_{\bm 0},\delta_{\bm n}\jd=-\frac{1}{2\pi\sqrt{-1}}\int_{\mathscr{C}}e^{-\sqrt{-1}tz}\ji(\mathcal{H}(\theta)-z)^{-1}\delta_{\bm0},\delta_{\bm n}\jd dz.
	\end{align*}
	Notice that for $z\in\mathscr{C}$, we have $\Im z\le t^{-1}$ and hence $|e^{-\sqrt{-1}tz}|=e^{t\Im z}\le e$. Thus
	\begin{align}\label{1234}
		|\ji e^{-\sqrt{-1}t\mathcal{H}(\theta)}\delta_{\bm 0},\delta_{\bm n}\jd|\le\frac{e}{2\pi}\sum_{j=1}^{4}\int_{\mathscr{C}_j}|\ji(\mathcal{H}(\theta)-z)^{-1}\delta_{\bm0},\delta_{\bm n}\jd|\cdot |dz|.
	\end{align}
	If $z\in\mathscr{C}_2\cup\mathscr{C}_3\cup\mathscr{C}_4$, then $\dist(z,\sigma(\mathcal{H}(\theta)))\ge\beta$ and hence Lemma \ref{cte} shows that
	\begin{align}
		\nonumber&\ \sum_{j=2}^{4}\int_{\mathscr{C}_j}|\ji(\mathcal{H}(\theta)-z)^{-1}\delta_{\bm0},\delta_{\bm n}\jd|\cdot |dz|\\
		\label{234}\le&\  \frac{2}{\beta}(b-a+6\beta+2t^{-1})e^{-\frac{9}{10}\alpha\log^{\rho}(1+\|\bm n\|)}.
	\end{align}
	Using the Cauchy–Schwarz inequality, we obtain
	\begin{align*}
		&\ \left(\int_{\mathscr{C}_1}|\ji(\mathcal{H}(\theta)-z)^{-1}\delta_{\bm0},\delta_{\bm n}\jd|\cdot |dz|\right)^2\\
		\le&\  (b-a+4\beta)\cdot\int_{a-2\beta}^{b+2\beta}|\ji(\mathcal{H}(\theta)-E-\sqrt{-1}t^{-1})^{-1}\delta_{\bm0},\delta_{\bm n}\jd|^2dE,
	\end{align*}
	which combines $(x+y)^2\le 2(x^2+y^2)$, \eqref{1234} and \eqref{234} implying that
	\begin{align*}
		&\ |\ji e^{-\sqrt{-1}t\mathcal{H}(\theta)}\delta_{\bm 0},\delta_{\bm n}\jd|^2\\
		\le&\  \frac{(b-a+4\beta)e^2}{2\pi^2}\int_{a-2\beta}^{b+2\beta}\left|\left\ji\left(\mathcal{H}(\theta)-E-\sqrt{-1}t^{-1}\right)^{-1}\delta_{\bm 0},\delta_{\bm n}\right\jd\right|^2dE\\
		&\ \ +\frac{2e^2}{\beta^2\pi^2}(b-a+6\beta+2t^{-1})^2e^{-\frac{9}{5}\alpha\log^{\rho}(1+\|\bm n\|)}.
	\end{align*}
	
	We then analyze \eqref{cbtapmgf}. Let 
	\begin{align*}
		f(t):=e^{-\frac{t}{T}}\ji e^{-\sqrt{-1}t\mathcal{H}(\theta)}\delta_{\bm 0},\delta_{\bm n}\jd\chi_{[0,\infty)}(t),
	\end{align*}
	where $\chi_{[0,\infty)}$ denotes the characteristic function of $[0,\infty)$. Its Fourier transform is given by
	\begin{align*}
		\hat{f}(-E)&=\int_{-\infty}^{\infty}e^{\sqrt{-1}tE}f(t)dt\\
		&=\int_0^{\infty}e^{-t(T^{-1}-\sqrt{-1}E)}\left(\int_{-\infty}^{\infty}e^{-\sqrt{-1}tu}\mu_{\delta_{\bm 0},\delta_{\bm n}}(du)\right)dt\\
		&=\int_{-\infty}^{\infty}\left(\int_0^{\infty}e^{-t(T^{-1}+\sqrt{-1}(u-E))}dt\right)\mu_{\delta_{\bm 0},\delta_{\bm n}}(du)\\
		&=\int_{-\infty}^{\infty}\frac{-\sqrt{-1}}{u-E-\sqrt{-1}T^{-1}}\mu_{\delta_{\bm 0},\delta_{\bm n}}(du)\\
		&=-\sqrt{-1}\left\ji \left(\mathcal{H}(\theta)-E-\sqrt{-1}T^{-1}\right)^{-1}\delta_{\bm 0},\delta_{\bm n}\right\jd,
	\end{align*}
	where  Fubini's theorem is used and $\mu_{\delta_{\bm 0},\delta_{\bm n}}$ is the spectral measure associated with $\mathcal{H}(\theta)$ and $\delta_{\bm 0},\delta_{\bm n}$. Thus, by Parseval's theorem on the Fourier transform,
	\begin{align*}
		&\ \frac{2}{T}\int_0^{\infty}e^{\frac{-2t}{T}}|\ji e^{-\sqrt{-1}t\mathcal{H}(\theta)}\delta_{\bm 0},\delta_{\bm n}\jd|^2dt\\
		=&\ \frac{1}{T\pi}\int_{-\infty}^{\infty}\left|\left\ji\left(\mathcal{H}(\theta)-E-\sqrt{-1}T^{-1}\right)^{-1}\delta_{\bm 0},\delta_{\bm n}\right\jd\right|^2dE.
	\end{align*}
	Recall that $\sigma(\mathcal{H}(\theta))\subset [a-\beta,b+\beta]$. By Lemma \ref{cte}, for any $E\in(b+2\beta,\infty)$, we have
	\begin{align*}
		\left|\left\ji\left(\mathcal{H}(\theta)-E-\sqrt{-1}T^{-1}\right)^{-1}\delta_{\bm 0},\delta_{\bm n}\right\jd\right|^2\le\frac{1}{\left(E-b-\frac{3}{2}\beta\right)^2}e^{-\frac{9}{5}\alpha\log^{\rho}(1+\|\bm n\|)}.
	\end{align*}
	Similarly, for any $E\in(-\infty,a-2\beta)$, we also have 
	\begin{align*}
		\left|\left\ji\left(\mathcal{H}(\theta)-E-\sqrt{-1}T^{-1}\right)^{-1}\delta_{\bm 0},\delta_{\bm n}\right\jd\right|^2\le\frac{1}{\left(E-a+\frac{3}{2}\beta\right)^2} e^{-\frac{9}{5}\alpha\log^{\rho}(1+\|\bm n\|)}.
	\end{align*}
	Therefore,
	\begin{align*}
		&\ \frac{2}{T}\int_0^{\infty}e^{\frac{-2t}{T}}|\ji e^{-\sqrt{-1}t\mathcal{H}(\theta)}\delta_{\bm 0},\delta_{\bm n}\jd|^2dt\\
		\le&\ \frac{1}{T\pi}\int_{a-2\beta}^{b+2\beta}\left|\left\ji\left(\mathcal{H}(\theta)-E-\sqrt{-1}T^{-1}\right)^{-1}\delta_{\bm 0},\delta_{\bm n}\right\jd\right|^2dE\\
		&\ +\frac{4}{\beta T\pi}e^{-\frac{9}{5}\alpha\log^{\rho}(1+\|\bm n\|)}.
	\end{align*}		
	We complete this proof.
\end{proof}

\subsection{Green's function estimate}
Because of subsection \ref{92}, we only need to bound Green's functions. We first have the following conclusion
\begin{lem}\label{93l1}
	If $a-2\beta\le E\le b+2\beta$, $\delta_s^{3}\le t^{-1}\le\beta$, then for $\theta\in\T$ and any $\bm k\in P_s$, we have
	\begin{align}\label{93a}
		\|\mathcal{T}_{\tilde{\Omega}_{\bm k}^s}^{-1}(E+\sqrt{-1}t^{-1};\theta)\|&\le\delta_{s}^{-3}
	\end{align}
and
	\begin{align}\label{93b}
		|\mathcal{T}_{\tilde{\Omega}_{\bm k}^s}^{-1}(E+\sqrt{-1}t^{-1};\theta)(\bm x,\bm y)|&< e^{-\alpha'_{s-1}\log^{\rho}(1+\|\bm x-\bm y\|)}\ \text{for $\|\bm x-\bm y\|\ge\frac{\tilde{\zeta}_{s-1}}{10}$}.
	\end{align}
\end{lem}
\begin{rem}
	Strictly speaking, the definitions of both $P_s$ and $\tO_{\bm k}^s$ depend on the energe $E+\sqrt{-1}t^{-1}$. For notational simplicity, we may suppress this dependence in our notation. We adopt the same notational convention for subsequent concepts $s$-regular sets and $s$-good sets.
\end{rem}

\begin{rem}
	We demand $a-2\beta\le E\le b+2\beta$, $t^{-1}\le\beta$ only for $E+\sqrt{-1}t^{-1}\in v(\D_{R/2})$, which meets the condition of Theorem \ref{ge} holds. If \eqref{93a} is established, the proof of \eqref{93b} is similar to Theorem \ref{psqs}, we omit the related details. 
\end{rem}
	\begin{proof}
		Since $\mathcal{H}(\theta)$ is self-adjoint and $\delta_s^{3}\le t^{-1}\le\beta$, we can obtain
		\begin{align*}
			\|\mathcal{T}_{\tilde{\Omega}_{\bm k}^s}^{-1}(E+\sqrt{-1}t^{-1};\theta)\|=\frac{1}{\dist(E+\sqrt{-1}t^{-1},\sigma(\mathcal{H}_{\tilde{\Omega}_{\bm k}^s}(\theta)))}\le \delta_{s}^{-3}.
		\end{align*}
		We finish this proof.
\end{proof}
Next, we need to estimate Green's function of more general sets. Assume that the finite set $\lg$ is $s$-regular, namely
\begin{align}\label{sreg}
	\lg\cap \tilde{\Omega}_{\bm k}^{s'}\Rightarrow \tilde{\Omega}_{\bm k}^{s'}\subset \lg\ \text{for}\ s'\le s\ \text{and}\ \bm k\in P_{s'}.
\end{align}
We have
\begin{lem}\label{93lsr}
	Under the assumptions of Lemma \ref{93l1}, if $\lg$ is $s$-regular, then for $\theta\in\T$,
	\begin{align}
		\label{tsrnorm}\|\mathcal{T}_{\lg}^{-1}(E+\sqrt{-1}t^{-1};\theta)\|\le\delta_{s}^{-3},
	\end{align}
	and
	\begin{align}
		\label{tsrdecay}|\mathcal{T}_{\lg}^{-1}(E+\sqrt{-1}t^{-1};\theta)(\bm x,\bm y)|<e^{-\alpha_s\log^{\rho}(1+\|\bm x-\bm y\|)}\text{ for $\|\bm x-\bm y\|>10\tz_s$}.
	\end{align}
\end{lem}

\begin{proof}
		From self-adjointness of  $\mathcal{H}(\theta)$ and $\delta_s^{3}\le t^{-1}\le\beta$, we get
	\begin{align*}
		\|\mathcal{T}_{\lg}^{-1}(E+\sqrt{-1}t^{-1};\theta)\|=\frac{1}{\dist(E+\sqrt{-1}t^{-1},\sigma(\mathcal{H}_{\lg}(\theta)))}\le \delta_{s}^{-3}.
	\end{align*}
	Define 
\begin{align}\label{wpt3}
	\widetilde{P}_t=\{\bm k\in P_t:\ \exists \bm k'\in Q_{t-1}\ {\rm s.t.},\ \tilde{\Omega}_{\bm k'}^{t-1}\subset \lg,\tilde{\Omega}_{\bm k'}^{t-1}\subset\Omega_{\bm k}^t\}\ (1\le t\le s).
\end{align}
Similar to \eqref{wincup}, for any $\bm l\in \lg$, if
\begin{align*}
	\bm l\in\bigcup_{\bm k\in \tilde{P}_1}2\Omega_{\bm k}^1,
\end{align*}
then there exists $1\le t\le s-1$ such that
\begin{align*}
	\bm l\in\bigcup_{\bm k\in \tilde{P}_t\setminus Q_t}2\Omega_{\bm k}^{t}
\end{align*}
or 
\begin{align*}
	\bm l\in \bigcup_{\bm k\in \tilde{P}_s}2\Omega_{\bm k}^{s}.
\end{align*}
For every $\bm l\in \lg$, define
\begin{align*}
	O(\bm l)=\left\{\begin{array}{ll}
		\lg_{\frac{1}{2}N_1}(\bm l)\cap  \lg, & \text{if }\bm l\notin\bigcup_{\bm k\in \tilde{P}_1}2\Omega_{\bm k}^1,\\
		\tilde{\Omega}_{\bm k}^t,  &\text{if }\bm l\in2\Omega_{\bm k}^t\ \text{for some }\bm k\in\widetilde{P}_t\setminus Q_t,\\
		\tilde{\Omega}_{\bm k}^s,  &\text{if }\bm l\in2\Omega_{\bm k}^s\ \text{for some }\bm k\in\widetilde{P}_s.
	\end{array}\right.
\end{align*}
In light of \eqref{PsQsnorm}, \eqref{PsQsdecay}, Lemma \ref{0g}, Lemma \ref{93l1} and \eqref{sreg}, the demonstration for \eqref{tsrdecay} parallels that of \eqref{tgsa}, thus we suppress the routine details.
\end{proof}

	Now, we can state our main conclusion of this part:
\begin{thm}\label{93main}
	If $a-2\beta\le E\le b+2\beta$ and $\delta_{s}^{3}\le t^{-1}<\min(\delta_{s-1}^{3},\beta)$, then for $\theta\in\T$ and any $\bm n\in\Z^d$ with $\|\bm n\|\ge e^{(\log t)^{\frac{2}{1+\rho'}}}\gg N_s^{10^3}$, we have
	\begin{align*}
		|\mathcal{T}^{-1}(E+\sqrt{-1}t^{-1};\theta)(\bm 0,\bm n)|< e^{-\frac{3}{4}\alpha_s\log^{\rho}\left(1+\|\bm n\|\right)}.
	\end{align*}
\end{thm}
\begin{proof}
	Since $\mathcal{H}(\theta)$ is self-adjoint, one has that for any $\lg\subset\Z^d$, $\dist(\sigma(\mathcal{H}_{\lg}(\theta)),E+\sqrt{-1}t^{-1})\ge t^{-1}$, and hence,
	\begin{align}\label{t}
		\|\mathcal{T}_{\lg}^{-1}(E+\sqrt{-1}t^{-1};\theta)\|\le t\le \delta_{s}^{-3}.
	\end{align}
	By Lemma \ref{93A}, for fixed $\bm n$ with $\|\bm n\|\ge e^{(\log t)^{\frac{2}{1+\rho'}}}\gg N_s^{10^3}$, there exists a $s$-regular set $O_{\bm n}$ satisfying
	\begin{align*}
		\lg_{\frac{1}{5}\|\bm n\|}(\bm n)\subset O_{\bm n}\subset\lg_{\frac{1}{5}\|\bm n\|+50N_{s}^{100}}(\bm n).
	\end{align*}
	For simplicity, we omit the dependence of $\mathcal{T}(E+\sqrt{-1}t^{-1};\theta)$ on $E,t,\theta$. Using resolvent identity shows that
	\begin{align*}
		\mathcal{T}^{-1}(\bm 0,\bm n)=-\sum_{\bm l\notin O_{\bm n},\bm l'\in O_{\bm n}}\mathcal{T}^{-1}(\bm 0,\bm l)\mathcal{W}(\bm l,\bm l')\mathcal{T}_{O_{\bm n}}^{-1}(\bm l',\bm n).
	\end{align*}
	Let $O_{\bm n,0}=\lg_{\frac{1}{10}\|\bm n\|}(\bm n)\cap O_{\bm n}$, we have
	\begin{align*}
		|\mathcal{T}^{-1}(\bm 0,\bm n)|&\le \sum_{\bm l\notin O_{\bm n},\bm l'\in O_{\bm n}}|\mathcal{T}^{-1}(\bm 0,\bm l)|\cdot|\mathcal{W}(\bm l,\bm l')|\cdot|\mathcal{T}_{O_{\bm n}}^{-1}(\bm l',\bm n)|\\
		&= ({\rm I})+({\rm II}),
	\end{align*}
	where
	\begin{align*}
		({\rm I})&=\sum_{\bm l\notin O_{\bm n},\bm l'\in O_{\bm n,0}}|\mathcal{T}^{-1}(\bm 0,\bm l)|\cdot|\mathcal{W}(\bm l,\bm l')|\cdot|\mathcal{T}_{O_{\bm n}}^{-1}(\bm l',\bm n)|,\\
		({\rm II})&=\sum_{\bm l\notin O_{\bm n},\bm l'\in O_{\bm n}\setminus O_{\bm n,0}}|\mathcal{T}^{-1}(\bm 0,\bm l)|\cdot|\mathcal{W}(\bm l,\bm l')|\cdot|\mathcal{T}_{O_{\bm n}}^{-1}(\bm l',\bm n)|.
	\end{align*}
	For $({\rm I})$, we have by \eqref{wphi}, \eqref{t} and $\|\bm l-\bm l'\|\ge \frac{1}{10}\|\bm n\|$ that 
	\begin{align*}
		({\rm I})&\le \|\mathcal{T}^{-1}\|\cdot\|\mathcal{T}_{O_{\bm n}}^{-1}\|\cdot(\# O_{\bm n,0})\cdot D\left(\frac{\alpha}{10}\right)e^{-\frac{9}{10}\alpha\log^{\rho}\left(1+\frac{1}{10}\|\bm n\|\right)}.
	\end{align*}
	For $({\rm II})$, we also have by \eqref{quaeq}, \eqref{t} and \eqref{tsrdecay} that 
	\begin{align*}
		({\rm II})&\le\|\mathcal{T}^{-1}\|\cdot(\# O_{\bm n})\cdot D\left(\frac{\alpha}{10}\right)e^{-\alpha_{s}\log^{\rho}\left(1+\frac{1}{10}\|\bm n\|\right)+C(\rho)\alpha_{s}\log^{\rho}2}.
	\end{align*}
	Therefore, since Lemma \ref{93lsr} and $t^{-1}<\delta_{s-1}^{3}$, we can obtain
	\begin{align*}
		|\mathcal{T}^{-1}(\bm 0,\bm n)|<e^{-\frac{3}{4}\alpha_s\log^{\rho}\left(1+\|\bm n\|\right)}.
	\end{align*}
	We finish this proof.
\end{proof}

\subsection{The completion of proof}
In this part, we will combine the previous content to complete  the proof of Theorem \ref{sne}.
\begin{proof}[Proof of Theorem \ref{sne}]
	Let $T_0=\max\{\beta^{-1},\delta_0^{-3}\}$. If $t\ge T_0$, there is some $s\ge1$ such that
	\begin{align*}
		\delta_{s}^{3}\le t^{-1}<\min(\delta_{s-1}^{3},\beta).
	\end{align*}
	We let
	\begin{align*}
		\left(|\mathscr{X}_{\mathcal{H}(\theta)}|_{\delta_{\bm 0}}^p\right)(t)&=\sum_{\bm n\in\Z^d}(1+\|\bm n\|)^p\left|\ji e^{-\sqrt{-1}t\mathcal{H}(\theta)}\delta_{\bm 0},\delta_{\bm n}\jd\right|^2\\
		&= ({\rm I})+({\rm II}),
	\end{align*}
	where
	\begin{align*}
		({\rm I})	&=\sum_{\|\bm n\|\le e^{(\log t)^{\frac{2}{1+\rho'}}}}(1+\|\bm n\|)^p\left|\ji e^{-\sqrt{-1}t\mathcal{H}(\theta)}\delta_{\bm 0},\delta_{\bm n}\jd\right|^2,\\
		({\rm II})&=\sum_{\|\bm n\|>e^{(\log t)^{\frac{2}{1+\rho'}}}}(1+\|\bm n\|)^p\left|\ji e^{-\sqrt{-1}t\mathcal{H}(\theta)}\delta_{\bm 0},\delta_{\bm n}\jd\right|^2.
	\end{align*}
	For $({\rm I})$, direct computation shows
	\begin{align*}
		({\rm I})\le&\ (1+e^{(\log t)^{\frac{2}{1+\rho'}}})^p\sum_{\|\bm n\|\le e^{(\log t)^{\frac{2}{1+\rho'}}}}\left|\ji e^{-\sqrt{-1}t\mathcal{H}(\theta)}\delta_{\bm 0},\delta_{\bm n}\jd\right|^2\\
		\le&\  (1+e^{(\log t)^{\frac{2}{1+\rho'}}})^p.
	\end{align*}
	For $({\rm II})$, by Lemma \ref{cbmgf} and Theorem \ref{93main}, we have
	\begin{align*}
		({\rm II})\le&\sum_{\|\bm n\|> e^{(\log t)^{\frac{2}{1+\rho'}}}}\frac{(b-a+4\beta)^2e^2}{2\pi^2}(1+\|\bm n\|)^{p}e^{-\frac{3}{2}\alpha_s\log^{\rho}\left(1+\|\bm n\|\right)}\\
		&+\sum_{\|\bm n\|> e^{(\log t)^{\frac{2}{1+\rho'}}}}\frac{2e^2}{\beta^2\pi^2}(b-a+6\beta+2t^{-1})^2(1+\|\bm n\|)^{p}e^{-\frac{9}{5}\alpha\log^{\rho}\left(1+\|\bm n\|\right)}\\
		\le&\ 1.
	\end{align*}
	Hence, we get
	\begin{align*}
		\left(|\mathscr{X}_{\mathcal{H}(\theta)}|_{\delta_{\bm 0}}^p\right)(t)\le 2^p e^{p(\log t)^{\frac{2}{1+\rho'}}}.
	\end{align*}
	Let $a(\bm n,T)=\frac{2}{T}\int_0^{\infty}e^{\frac{-2t}{T}}|\ji e^{-\sqrt{-1}t\mathcal{H}(\theta)}\delta_{\bm 0},\delta_{\bm n}\jd|^2dt$, then
	\begin{align*}
		\left(|\bar{\mathscr{X}}_{\mathcal{H}(\theta)}|_{\delta_{\bm 0}}^p\right)(T)&=\sum_{\bm n\in\Z^d}(1+\|\bm n\|)^p a(\bm n,T)\\
		&= ({\rm III})+({\rm IV}),
	\end{align*}
	where
	\begin{align*}
		({\rm III})	&=\sum_{\|\bm n\|\le e^{(\log T)^{\frac{2}{1+\rho'}}}}(1+\|\bm n\|)^p a(\bm n,T),\\
		({\rm IV})&=\sum_{\|\bm n\|> e^{(\log T)^{\frac{2}{1+\rho'}}}}(1+\|\bm n\|)^p a(\bm n,T).
	\end{align*}
	For $({\rm III})$, direct computation shows
	\begin{align*}
		({\rm III})\le&\  (1+e^{(\log T)^{\frac{2}{1+\rho'}}})^p\sum_{\|\bm n\|\le e^{(\log T)^{\frac{2}{1+\rho'}}}}a(\bm n,T)\\
		\le&\  (1+e^{(\log T)^{\frac{2}{1+\rho'}}})^p.
	\end{align*}
	For $({\rm IV})$, by Lemma \ref{cbmgf} and Theorem \ref{93main}, we also have
	\begin{align*}
		({\rm IV})\le&\sum_{\|\bm n\|> e^{(\log T)^{\frac{2}{1+\rho'}}}}(1+\|\bm n\|)^{p}e^{-\frac{3}{2}\alpha_s\log^{\rho}\left(1+\|\bm n\|\right)}\\
		&+\sum_{\|\bm n\|> e^{(\log T)^{\frac{2}{1+\rho'}}}}\frac{4}{\beta\pi}((1+\|\bm n\|)^{p}e^{-\frac{9}{5}\alpha\log^{\rho}\left(1+\|\bm n\|\right)}\\
		\le&\ 1.
	\end{align*}
	Therefore, we also obtain
	\begin{align*}
		\left(|\bar{\mathscr{X}}_{\mathcal{H}(\theta)}|_{\delta_{\bm 0}}^p\right)(T)\le2^p e^{p(\log T)^{\frac{2}{1+\rho'}}}.
	\end{align*}
	We complete this proof.
\end{proof}
	
	\section*{Acknowledgments}
Y. Wu was  partially supported by NSFC  (12301227).
\section*{Data Availability}
The manuscript has no associated data.
\section*{Declarations}
{\bf Conflicts of interest} \ The authors  state  that there is no conflict of interest.

	\newpage

	\appendix{}
	\section{}\label{app}
	
	\begin{proof}[Proof of Lemma \ref{el}]
		Direct computation shows
		\begin{align}\label{51001}
			\log^{\rho}(1+x-y)=\left(1+\frac{\log\left(1-\frac{y}{1+x}\right)}{\log(1+x)}\right)^{\rho}\log^{\rho}(1+x).
		\end{align}
		Since $x>y$, we have
		\begin{align*}
			\frac{\log\left(1-\frac{y}{1+x}\right)}{\log(1+x)}>-1,
		\end{align*}
		which combines $(1+z)^{\rho}\ge1+\rho z$ $(\rho>1,z>-1)$ showing that
		\begin{align}\label{51002}
			\left(1+\frac{\log\left(1-\frac{y}{1+x}\right)}{\log(1+x)}\right)^{\rho}\ge 1+\rho\frac{\log\left(1-\frac{y}{1+x}\right)}{\log(1+x)}.
		\end{align}
		From $0<2y<1+x$ and $\log(1-z)\ge -2z$ $\left(0<z<\frac{1}{2}\right)$, we can get
		\begin{align}\label{51003}
			\log\left(1-\frac{y}{1+x}\right)\ge-\frac{2y}{1+x}.
		\end{align}
		Combining \eqref{51001}--\eqref{51003} gives
		\begin{align*}
			\log^{\rho}(1+x-y)\ge\left(1-2\rho\frac{y}{(1+x)\log(1+x)}\right)\log^{\rho}(1+x).
		\end{align*}
		We finish this proof.
	\end{proof}
	
	\begin{proof}[Proof of Lemma \ref{chi}]
		By Hadamard's inequality, we have for any $\bm i,\bm j\in\lg, $
		\begin{align*}
			|\ji \delta_{\bm i}, \mathcal{S}_{\lg}^* \delta_{\bm j}\jd|&\le \prod_{\bm l\ne\bm i}\left(\sum_{\bm k\ne\bm j}|\ji \delta_{\bm l}, \mathcal{S}_{\lg}\delta_{\bm k}\jd|^2\right)^{\frac{1}{2}}\\
			&\le \prod_{\bm l\ne\bm i}\left(\sum_{\bm k\ne\bm j}|\ji \delta_{\bm l}, \mathcal{S}_{\lg}\delta_{\bm k}\jd|\right)\\
			&\le \left(\sup_{\bm x\in\lg}\sum_{\bm y\in\lg}|\ji \delta_{\bm x}, \mathcal{S}_{\lg}\delta_{\bm y}\jd|\right)^{\#\lg-1}.
		\end{align*}
		Moreover, we have
		\begin{align*}
			\|\mathcal{S}_{\lg}^*\|&\le \left(\left(\sup_{\bm i\in\lg}\sum_{\bm j\in\lg}|\ji \delta_{\bm i}, \mathcal{S}_{\lg}^* \delta_{\bm j}\jd|\right)\cdot\left(\sup_{\bm j\in\lg}\sum_{\bm i\in\lg}|\ji \delta_{\bm i}, \mathcal{S}_{\lg}^* \delta_{\bm j}\jd|\right)\right)^{\frac{1}{2}}\\
			&\le (\#\lg)\cdot\left(\sup_{\bm x\in\lg}\sum_{\bm y\in\lg}|\ji \delta_{\bm x}, \mathcal{S}_{\lg}\delta_{\bm y}\jd|\right)^{\#\lg-1}.
		\end{align*}
		We complete this proof.
	\end{proof}
	
	\begin{proof}[Proof of Lemma \ref{cte}]
		For any fixed $\bm y\in \Z^d$ and any cut off $K > 0$, the
		multiplication operator
		\begin{align}\label{cutoff}
			\boM= e^{\lambda'\min \{\log^{\rho}(1+\|\bm x-\bm y\|), K\}}\delta_{\bm x,\bm x'},\ \bm x\in\Z^d
		\end{align}
		is bounded and invertible on $ \ell^2(\Z^d)$ for any $ \lambda'\in(0,+\infty)$. The Combes-Thomas
		estimate is based on the observation that for any $ \bm x $ with $ \log^{\rho}(1+\|\bm x-\bm y\|) \le K$ and
		any $ z \notin \sigma(\boA) \cup \sigma(\boM \boA \boM^{-1})$,
		\begin{align}
			\nonumber\boG(z)(\bm x,\bm y) e^{\lambda'\log^{\rho}(1+\|\bm x-\bm y\|)}&=\ji \delta_{\bm x} , \boM(\boA-z)^{-1}\boM^{-1}\delta_{\bm y}\jd\\
			\nonumber&=\ji \delta_{\bm x} ,(\boM \boA \boM^{-1}-z)^{-1}\delta_{\bm y}\jd\\
			\label{s23}&=\ji \delta_{\bm x} ,(\boA + \boB -z)^{-1}\delta_{\bm y}\jd,
		\end{align}
		where
		\begin{eqnarray}\label{s24}
			\boB := \boM \boA \boM^{-1}- \boA.
		\end{eqnarray}
		For a proof of \eqref{s22}, this will be done by estimating the
		$\ell^2$-norm of $\boB$. By Schur's test, we need to estimate $\boB(\bm x,\bm x')$ for all $\bm x,\bm x'\in\Z^d$. We divide the discussion into two cases:
		\begin{itemize}
			\item[\textbf{Case 1}:] $\bm x=\bm x'$ or $\min\{\log^{\rho}(1+\|\bm x-\bm y\|),\log^{\rho}(1+\|\bm x'-\bm y\|)\}\ge K$. In this case, direct computation shows that $\boB(\bm x,\bm x')=0$.
			\item[\textbf{Case 2}:] $\bm x\ne\bm x'$ and  $\max\{\log^{\rho}(1+\|\bm x-\bm y\|),\log^{\rho}(1+\|\bm x'-\bm y\|)\}<K$. From \eqref{quaeq}, \eqref{cutoff} and \eqref{s24}, we have
			\begin{align*}
				|\boB(\bm x,\bm x')|&=|\boA(\bm x,\bm x')|\cdot\left|e^{\lambda'\log^{\rho}(1+\|\bm x-\bm y\|)-\lambda'\log^{\rho}(1+\|\bm x'-\bm y\|)}-1\right|\\
				&\le |\boA(\bm x,\bm x')|\cdot\left(e^{\left|\lambda'\log^{\rho}(1+\|\bm x-\bm y\|)-\lambda'\log^{\rho}(1+\|\bm x'-\bm y\|)\right|}-1\right)\\
				&\le 2e^{\lambda'C(\rho)\log^{\rho}2}|\boA(\bm x,\bm x')|\cdot e^{\lambda'\log^{\rho}(1+\|\bm x-\bm x'\|)}.
			\end{align*}
			\item[\textbf{Case 3}:] $\min\{\log^{\rho}(1+\|\bm x-\bm y\|),\log^{\rho}(1+\|\bm x'-\bm y\|)\}<K$ and  $\max\{\log^{\rho}(1+\|\bm x-\bm y\|),\log^{\rho}(1+\|\bm x'-\bm y\|)\}\ge K$. Without loss of generality, we assume that $\log^{\rho}(1+\|\bm x-\bm y\|)< K$ and $\log^{\rho}(1+\|\bm x'-\bm y\|)\ge K$. Also by \eqref{quaeq}, \eqref{cutoff} and \eqref{s24}, we can get
			\begin{align*}
						|\boB(\bm x,\bm x')|&=|\boA(\bm x,\bm x')|\cdot\left|e^{\lambda'\log^{\rho}(1+\|\bm x-\bm y\|)-\lambda'K}-1\right|\\
				&\le |\boA(\bm x,\bm x')|\cdot\left(e^{\left|\lambda'\log^{\rho}(1+\|\bm x-\bm y\|)-\lambda' K\right|}-1\right)\\
				&=|\boA(\bm x,\bm x')|\cdot\left(e^{\lambda' K-\lambda'\log^{\rho}(1+\|\bm x-\bm y\|)}-1\right)\\
				&\le |\boA(\bm x,\bm x')|\cdot\left(e^{\lambda' \log^{\rho}(1+\|\bm x'-\bm y\|)-\lambda'\log^{\rho}(1+\|\bm x-\bm y\|)}-1\right)\\
				&\le 2e^{\lambda'C(\rho)\log^{\rho}2}|\boA(\bm x,\bm x')|\cdot e^{\lambda'\log^{\rho}(1+\|\bm x-\bm x'\|)}.
			\end{align*}
		\end{itemize}
		To sum up, we obtain $\boB(\bm x,\bm x)=0$ and
		\begin{align}\label{Bxx'}
		|\boB(\bm x,\bm x')|\le 2e^{\lambda'C(\rho)\log^{\rho}2}|\boA(\bm x,\bm x')|\cdot e^{\lambda'\log^{\rho}(1+\|\bm x-\bm x'\|)},\ \bm x\ne\bm x',
	\end{align}
	 which is independent from $K$. According to the self-adjointness of $\boA$, \eqref{s21} and \eqref{Bxx'}, we have
	 \begin{align*}
	 	\|\boB\|\le 2e^{\lambda'C(\rho)\log^{\rho}2}\cdot S_{\lambda'}\le 2e^{\lambda C(\rho)\log^{\rho}2}\cdot S_{\lambda}<\infty
	 \end{align*}
	 for all $\lambda'\le \lambda$. The spectrum of $\boB + \boA $ has distance at most $ \|\boB\| $ from the
		spectrum of $ \boA$. Hence $ \boB +\boA -z $ is invertible as long as $ \mathscr{D}=\dist(z,\sigma(\boA)) > \|B\|$, for which
		\begin{align}\label{s27}
			\|(\boB + \boA -z)^{-1}\| \leq \frac{1}{\mathscr{D} - \|\boB\|}	 \leq \frac{1}{\mathscr{D} - 2e^{\lambda' C(\rho)\log^{\rho}2}\cdot S_{\lambda'}}.
		\end{align}
		In view of \eqref{s23} and the fact that $K$ was arbitrary, we have thus proved \eqref{s22}.
	\end{proof}

	\begin{proof}[Proof of Lemma \ref{ite1}]
			Using the resolvent identity implies
		\begin{align*}
			\boT_{\lg}^{-1}(\bm u,\bm v)=-\ep \sum_{\bm w\in\tO_{\bm k}^1\atop \bm w'\in\lg\setminus\tO_{\bm k}^1}\boT_{\tO_{\bm k}^1}^{-1}(\bm u,\bm w)\cdot\boW(\bm w,\bm w')\cdot\boT_{\lg}^{-1}(\bm w',\bm v),
		\end{align*}
		and then
		\begin{align}
			\nonumber|\boT_{\lg}^{-1}(\bm u,\bm v)|&\le \max_{\bm w\in\tO_{\bm k}^1\atop \bm w'\in\lg\setminus\tO_{\bm k}^1}\left|\boT_{\tO_{\bm k}^1}^{-1}(\bm u,\bm w)\cdot\boW(\bm w,\bm w')\cdot e^{\frac{\alpha}{10}\log^{\rho}(1+\|\bm w-\bm w'\|)}\cdot\boT_{\lg}^{-1}(\bm w',\bm v)\right|\\
			\nonumber&\ \ \times\left(\sum_{\bm w\in\tO_{\bm k}^1\atop \bm w'\in\lg\setminus\tO_{\bm k}^1}e^{-\frac{\alpha}{10}\log^{\rho}(1+\|\bm w-\bm w'\|)}\right)\\
			\nonumber&\le \max_{\bm w\in\tO_{\bm k}^1\atop \bm w'\in\lg\setminus\tO_{\bm k}^1}\left|\boT_{\tO_{\bm k}^1}^{-1}(\bm u,\bm w)\cdot\boW(\bm w,\bm w')\cdot e^{\frac{\alpha}{10}\log^{\rho}(1+\|\bm w-\bm w'\|)}\cdot\boT_{\lg}^{-1}(\bm w',\bm v)\right|\\
			\label{1gri2}&\ \ \times(\#\tO_{\bm k}^1)\cdot D\left(\frac{\alpha}{10}\right),
		\end{align}
		where $D\left(\frac{\alpha}{10}\right)$ is defined in \eqref{Det}. Since $\tO_{\bm k}^1$ and $\lg\setminus\tO_{\bm k}^1$ are finite sets, there are $\bm u^*\in\tO_{\bm k}^1$ and $\bm u'\in \lg\setminus\tO_{\bm k}^1$ such that
		\begin{align}
			\nonumber&|\boT_{\tO_{\bm k}^1}^{-1}(\bm u,\bm u^*)\cdot\boW(\bm u^*,\bm u')\cdot e^{\frac{\alpha}{10}\log^{\rho}(1+\|\bm u^*-\bm u'\|)}\cdot\boT_{\lg}^{-1}(\bm u',\bm v)|\\
			=&\ \
			\label{1gmax}\max_{\bm w\in\tO_{\bm k}^1\atop \bm w'\in\lg\setminus\tO_{\bm k}^1}|\boT_{\tO_{\bm k}^1}^{-1}(\bm u,\bm w)\cdot\boW(\bm w,\bm w')\cdot e^{\frac{\alpha}{10}\log^{\rho}(1+\|\bm w-\bm w'\|)}\cdot\boT_{\lg}^{-1}(\bm w',\bm v)|.
		\end{align}
		Let $X:=\lg_{\frac{\tz_1}{9}}(\bm k)\subset \tO_{\bm k}^1$. We will divide  the discussion into two cases, $\bm u^*\in X$ and $\bm u^*\in\tO_{\bm k}^1\setminus X$.
		\begin{itemize}
			\item[\textbf{Case 1}:] $\bm u^*\in X$. Since $\bm k\in\widetilde{P}_1\subset P_1\setminus Q_1$ (cf. \eqref{tk0}) and \eqref{wphi}, we can get
					\begin{align}
				\nonumber&\ \ |\boT_{\tO_{\bm k}^1}^{-1}(\bm u,\bm u^*)\cdot\boW(\bm u^*,\bm u')\cdot e^{\frac{\alpha}{10}\log^{\rho}(1+\|\bm u^*-\bm u'\|)}|\\
				\label{1gTW01}\le &\ \ 2\delta_0^{-3}\delta_1^{-2}\cdot e^{-\frac{9}{10}\alpha\log^{\rho}(1+\|\bm u-\bm u'\|)}.
			\end{align}
			\item[\textbf{Case 2}:] $\bm u^*\in\tO_{\bm k}^1\setminus X$. In this case, $\|\bm u^*-\bm u\|\ge\frac{\tz_1}{10}$. From \eqref{tg1a}, \eqref{quaeq} and  \eqref{wphi}, we have
					\begin{align}
				\nonumber&\ \ |\boT_{\tO_{\bm k}^1}^{-1}(\bm u,\bm u^*)\cdot\boW(\bm u^*,\bm u')\cdot e^{\frac{\alpha}{10}\log^{\rho}(1+\|\bm u^*-\bm u'\|)}|\\
				\nonumber\le&\ \  e^{-\alpha'_0\log^{\rho}(1+\|\bm u-\bm u^*\|)}\cdot e^{-\frac{9}{10}\alpha\log^{\rho}(1+\|\bm u^*-\bm u'\|)}\\
				\label{1gTW02}\le&\ \  e^{-\alpha'_0\log^{\rho}(1+\|\bm u-\bm u'\|)+\alpha'_0 C(\rho)\log^{\rho}2}.
			\end{align}
		\end{itemize}
		Combining \eqref{1gTW01}, \eqref{1gTW02} and $\delta_1\ll1$ gives
		\begin{align}
			\nonumber&\ \ |\boT_{\tO_{\bm k}^1}^{-1}(\bm u,\bm u^*)\cdot\boW(\bm u^*,\bm u')\cdot e^{\frac{\alpha}{10}\log^{\rho}(1+\|\bm u^*-\bm u'\|)}|\\
			\label{1gTW03}\le&\ \ 2\delta_0^{-3}\delta_1^{-2}\cdot e^{-\alpha'_0\log^{\rho}(1+\|\bm u-\bm u'\|)}.
		\end{align}
		By \eqref{1gri2}, \eqref{1gmax}, \eqref{1gTW03}, $\bm u'\in\lg\setminus\tO_{\bm k}^1$ ($\|\bm u-\bm u'\|\ge\frac{\tz_1}{3}$), $\#(\tO_{\bm k}^1)\le (10N_1^{100}+1)^d$, \eqref{indpa} and $\delta_1\ll1$, we obtain
		\begin{align}\label{51801}
		\nonumber	|\boT_{\lg}^{-1}(\bm u,\bm v)|&\le (\#(\tO_{\bm k}^1))\cdot e^{-\alpha'_0\log^{\rho}(1+\|\bm u-\bm u'\|)}\cdot|\boT_{\lg}^{-1}(\bm u',\bm v)|\cdot\delta_1^{-3}\\
			&\le e^{-\alpha''_0\log^{\rho}(1+\|\bm u-\bm u'\|)}\cdot|\boT_{\lg}^{-1}(\bm u',\bm v)|
		\end{align}
		where $\alpha''_0=\frac{3}{4}\alpha\left(1-\frac{30\times 10^{5\rho'}}{\alpha\log^{\rho-\rho'}N_1}\right)$.
	
		We finish this proof.
	\end{proof}
	\begin{proof}[Proof of Lemma \ref{itetg}]
			Using the resolvent identity implies
		\begin{align*}
			\boT_{\lg}^{-1}(\bm u,\bm v)=-\ep \sum_{\bm w\in\lg'\atop \bm w'\in\lg\setminus\lg'}\boT_{\lg'}^{-1}(\bm u,\bm w)\cdot\boW(\bm w,\bm w')\cdot\boT_{\lg}^{-1}(\bm w',\bm v),
		\end{align*}
		and then
		\begin{align}
			\nonumber|\boT_{\lg}^{-1}(\bm u,\bm v)|&\le \max_{\bm w\in\lg'\atop \bm w'\in\lg\setminus\lg'}\left|\boT_{\lg'}^{-1}(\bm u,\bm w)\cdot\boW(\bm w,\bm w')\cdot e^{\frac{\alpha}{10}\log^{\rho}(1+\|\bm w-\bm w'\|)}\cdot\boT_{\lg}^{-1}(\bm w',\bm v)\right|\\
			\nonumber&\ \ \times\left(\sum_{\bm w\in\lg'\atop \bm w'\in\lg\setminus\lg'}e^{-\frac{\alpha}{10}\log^{\rho}(1+\|\bm w-\bm w'\|)}\right)\\
			\nonumber&\le \max_{\bm w\in\lg'\atop \bm w'\in\lg\setminus\lg'}\left|\boT_{\lg'}^{-1}(\bm u,\bm w)\cdot\boW(\bm w,\bm w')\cdot e^{\frac{\alpha}{10}\log^{\rho}(1+\|\bm w-\bm w'\|)}\cdot\boT_{\lg}^{-1}(\bm w',\bm v)\right|\\
			\label{tgri}&\ \ \times(\#\lg')\cdot D\left(\frac{\alpha}{10}\right),
		\end{align}
		where $D\left(\frac{\alpha}{10}\right)$ is defined in \eqref{Det}. Since $\lg'$ and $\lg\setminus\lg'$ are finite sets, there are $\bm u^*\in\lg'$ and $\bm u'\in \lg\setminus\lg'$ such that
		\begin{align}
			\nonumber&|\boT_{\lg'}^{-1}(\bm u,\bm u^*)\cdot\boW(\bm u^*,\bm u')\cdot e^{\frac{\alpha}{10}\log^{\rho}(1+\|\bm u^*-\bm u'\|)}\cdot\boT_{\lg}^{-1}(\bm u',\bm v)|\\
			=&\ \
			\label{tgmax}\max_{\bm w\in\lg'\atop \bm w'\in\lg\setminus\lg'}|\boT_{\lg'}^{-1}(\bm u,\bm w)\cdot\boW(\bm w,\bm w')\cdot e^{\frac{\alpha}{10}\log^{\rho}(1+\|\bm w-\bm w'\|)}\cdot\boT_{\lg}^{-1}(\bm w',\bm v)|.
		\end{align}
		If $\|\bm u^*-\bm u\|\le 10\tz_t$, from $\lg'$ is $t$-good (cf. \eqref{tsgnorm}) and \eqref{wphi}, we can get
		\begin{align}
			\nonumber&\ \ |\boT_{\lg'}^{-1}(\bm u,\bm u^*)\cdot\boW(\bm u^*,\bm u')\cdot e^{\frac{\alpha}{10}\log^{\rho}(1+\|\bm u^*-\bm u'\|)}|\\
			\label{tgTW01}\le &\ \ \delta_t^{-3}\cdot e^{-\frac{9}{10}\alpha\log^{\rho}(1+\|\bm u-\bm u'\|)}.
		\end{align}
		If $\|\bm u^*-\bm u\|> 10\tz_t$, since $\lg'$ is $t$-good (cf. \eqref{tsgdecay}), \eqref{wphi} and \eqref{quaeq}, we have
		\begin{align}
			\nonumber&\ \ |\boT_{\lg'}^{-1}(\bm u,\bm u^*)\cdot\boW(\bm u^*,\bm u')\cdot e^{\frac{\alpha}{10}\log^{\rho}(1+\|\bm u^*-\bm u'\|)}|\\
			\nonumber\le&\ \  e^{-\alpha_t\log^{\rho}(1+\|\bm u-\bm u^*\|)}\cdot e^{-\frac{9}{10}\alpha\log^{\rho}(1+\|\bm u^*-\bm u'\|)}\\
			\label{tgTW02}\le&\ \  e^{-\alpha_t\log^{\rho}(1+\|\bm u-\bm u'\|)+\alpha_t C(\rho)\log^{\rho}2}.
		\end{align}
		Combining \eqref{tgTW01}, \eqref{tgTW02} and $\delta_t\ll1$ gives
		\begin{align}
			\nonumber&\ \ |\boT_{\lg'}^{-1}(\bm u,\bm u^*)\cdot\boW(\bm u^*,\bm u')\cdot e^{\frac{\alpha}{10}\log^{\rho}(1+\|\bm u^*-\bm u'\|)}|\\
			\label{tgTW03}\le&\ \ \delta_t^{-3}\cdot e^{-\alpha_t\log^{\rho}(1+\|\bm u-\bm u'\|)}.
		\end{align}
		By \eqref{tgri}, \eqref{tgmax}, \eqref{tgTW03} and $\delta_t\ll1$, we obtain
		\begin{align}\label{60101}
			|\boT_{\lg}^{-1}(\bm u,\bm v)|\le (\#\lg')\cdot e^{-\alpha_t\log^{\rho}(1+\|\bm u-\bm u'\|)}\cdot|\boT_{\lg}^{-1}(\bm u',\bm v)|\cdot\delta_t^{-3}.
		\end{align}
		We finish this proof.
	\end{proof}

	\begin{lem}[Schur complement lemma]\label{scl}
		Let $\lg_1$ and $\lg_2$ be finite subsets of $\Z^d$ with $\lg_1\cap\lg_2=\emptyset$. Suppose $\mathcal{A}\in\mathbf{M}_{\lg_1}^{\lg_1}$, $\mathcal{B}\in\mathbf{M}_{\lg_2}^{\lg_1}$, $\mathcal{C}\in\mathbf{M}_{\lg_1}^{\lg_2}$, $\mathcal{D}\in\mathbf{M}_{\lg_2}^{\lg_2}$ and
		\begin{align*}
			\mathcal{M}=\left(\begin{array}{cc}
				\mathcal{A} & \mathcal{B}\\
				\mathcal{C} & \mathcal{D}
			\end{array}\right)\in\mathbf{M}_{\lg_1\cup\lg_2}^{\lg_1\cup\lg_2}.
		\end{align*}
		Assume further that $\mathcal{A}$ is invertible and $\|\mathcal{B}\|,\|\mathcal{C}\|\le1$. Then we have
		\begin{itemize}
			\item[(1)]
			\begin{align*}
				\det \mathcal{M}=\det \mathcal{A}\cdot\det \mathcal{S},
			\end{align*}
			where
			\begin{align*}
				\mathcal{S}=\mathcal{D}-\mathcal{C}\mathcal{A}^{-1}\mathcal{B}\in\mathbf{M}_{\lg_2}^{\lg_2}
			\end{align*}
			is called the Schur complement of $\mathcal{A}$.
			\item[(2)] $\mathcal{M}$ is invertible iff $\mathcal{S}$ is invertible and
			\begin{align}\label{sc}
				\|\mathcal{S}^{-1}\|\le\|\mathcal{M}^{-1}\|<4(1+\|\mathcal{A}^{-1}\|)^2(1+\|\mathcal{S}^{-1}\|).
			\end{align}
		\end{itemize}
	\end{lem}
	\begin{proof}[Proof of Lemma \ref{scl}]
		(1) Since $\mathcal{A}$ is invertible, we have
		\begin{align*}
			\left(\begin{array}{cc}
				\mathcal{I}_{\lg_1} & \bm 0\\
				-\mathcal{C}\mathcal{A}^{-1}& \mathcal{I}_{\lg_2}
			\end{array} \right)\mathcal{M} \left(\begin{array}{cc}
				\mathcal{I}_{\lg_1} & -\mathcal{A}^{-1}\mathcal{B}\\
				\bm 0 & \mathcal{I}_{\lg_2}
			\end{array}\right)=\left(\begin{array}{cc}
				\mathcal{A} & \bm 0\\
				\bm 0 & \mathcal{S}
			\end{array}\right),
		\end{align*}
		which implies
		\begin{align*}
			\det\mathcal{M}=\det\mathcal{A}\cdot\det\mathcal{S}.
		\end{align*}
		
		(2) Direct computation shows
		\begin{align*}
			\mathcal{M}^{-1}=\left(\begin{array}{cc}
				\mathcal{A}^{-1}+\mathcal{A}^{-1}\mathcal{B}\mathcal{S}^{-1}\mathcal{C}\mathcal{A}^{-1} & -\mathcal{A}^{-1}\mathcal{B}\mathcal{S}^{-1}\\
				-\mathcal{S}^{-1}\mathcal{C}\mathcal{A}^{-1} & \mathcal{S}^{-1}
			\end{array}\right),
		\end{align*}
		
		which combines $\|\boA\boB\|\le \|\boA\|\cdot\|\boB\|$ implying \eqref{sc}.
	\end{proof}
	
	\begin{lem}\label{ef}
		Let $\bm l\in\frac{1}{2}\Z^d$ and let $\lg\subset\Z^d+\bm l$ be a finite set which is symmetrical about the origin (i.e., $\bm n\in\lg\Leftrightarrow-\bm n\in\lg$). Then
		\begin{align*}
			\det (\mathcal{T}(z))_{\lg}=\det((v(z+\bm n\cdot\bm\omega)-E)\delta_{\bm n,\bm n'}+\ep \mathcal{W})_{\lg}
		\end{align*}
		is an even function of $z$.
	\end{lem}

	\begin{proof}
		Define the unitary map
		\begin{align*}
			\mathcal{U}_\lg:\ell^2(\lg)\rightarrow\ell^2(\lg)\ \text{with}\ (\mathcal{U}_{\lg}\psi)(\bm n)=\psi(-\bm n).
		\end{align*}
		Then
		\begin{align*}
			\mathcal{U}_{\lg}^{-1}(\mathcal{T}(z))_{\lg}\mathcal{U}_{\lg}=((v(z-\bm n\cdot\bm\omega)-E)\delta_{\bm n,\bm n'}+\ep \mathcal{W})_{\lg}=(\mathcal{T}(-z))_{\lg},
		\end{align*}
		which implies
		\begin{align*}
			\det (\mathcal{T}(z))_{\lg}=\det (\mathcal{T}(-z))_{\lg}.
		\end{align*}
		We completes this proof.
	\end{proof}

		\begin{lem}\label{det1}
		Let $\mathcal{A},\mathcal{B}:\ell^2(\Z^d)\rightarrow\ell^2(\Z^d)$ be linear operators and let $\lg$ be a finite subset of $\Z^d$. If $\sup_{\bm x\in\lg}\sum_{\bm y\in\lg}|\ji\delta_{\bm x},\mathcal{A}_{\lg}\delta_{\bm y}\jd|\le M$ and $\sup_{\bm x\in\lg}\sum_{\bm y\in\lg}|\ji\delta_{\bm x},\mathcal{B}_{\lg}\delta_{\bm y}\jd|\le \epsilon$, then
		\begin{align}\label{detd}
			\left|\det(\mathcal{A}_{\lg}+\mathcal{B}_{\lg})-\det \mathcal{A}_{\lg}\right|\le \ep (\#\lg)^2(M+\ep)^{\#\lg-1}.
		\end{align}
	\end{lem}

	\begin{proof}
		Let $f(t)=\det(\mathcal{A}_{\lg}+t\mathcal{B}_{\lg})$. Then
		\begin{align*}
			f'(t)={\rm tr}(\mathcal{B}_{\lg}(\mathcal{A}_{\lg}+t\mathcal{B}_{\lg})^*).
		\end{align*}
		Since $\sup_{\bm x\in\lg}\sum_{\bm y\in\lg}|\ji\delta_{\bm x},\mathcal{A}_{\lg}\delta_{\bm y}\jd|\le M$ and $\sup_{\bm x\in\lg}\sum_{\bm y\in\lg}|\ji\delta_{\bm x},\mathcal{B}_{\lg}\delta_{\bm y}\jd|\le \epsilon$, we have
		\begin{align*}
		\sup_{\bm x\in\lg}\sum_{\bm y\in\lg}|\ji\delta_{\bm x},(\boA_{\lg}+t\boB_{\lg})\delta_{\bm y}\jd|&\le 	\sup_{\bm x\in\lg}\sum_{\bm y\in\lg}|\ji\delta_{\bm x},\mathcal{A}_{\lg}\delta_{\bm y}\jd|+|t|\cdot\sup_{\bm x\in\lg}\sum_{\bm y\in\lg}|\ji\delta_{\bm x},\mathcal{B}_{\lg}\delta_{\bm y}\jd|\\
		&\le M+\epsilon |t|.
		\end{align*}
		By Lemma \ref{chi}, we get for any $\bm i,\bm j\in\lg$,
		\begin{align*}
			|\ji \delta_{\bm i}, (\mathcal{A}_{\lg}+t\mathcal{B}_{\lg})^* \delta_{\bm j}\jd|\le&\ \left(\sup_{\bm x\in\lg}\sum_{\bm y\in\lg}|\ji\delta_{\bm x},(\boA_{\lg}+t\boB_{\lg})\delta_{\bm y}\jd|\right)^{\#\lg-1}\\
			\le&\  (M+\epsilon |t|)^{\#\lg-1}.
		\end{align*}
		Therefore,
		\begin{align*}
			|f'(t)|&\le \sum_{\bm i\in\lg}|\ji \delta_{\bm i},\mathcal{B}_{\lg}(\mathcal{A}_{\lg}+t\mathcal{B}_{\lg})^* \delta_{\bm i}\jd|\\
			&\le\epsilon (\#\lg)^2\max_{\bm i,\bm j\in\lg}|\ji \delta_{\bm i}, (\mathcal{A}_{\lg}+t\mathcal{B}_{\lg})^* \delta_{\bm j}\jd|\\
			&\le \epsilon(\#\lg)^2(M+\epsilon|t|)^{\#\lg-1}.
		\end{align*}
		According to the mean-value theorem, we obtain for some $\xi\in(0,1),$
		\begin{align*}
			|\det(\mathcal{A}_{\lg}+\mathcal{B}_{\lg})-\det \mathcal{A}_{\lg}|=|f(1)-f(0)|=|f'(\xi)|\le \epsilon (\#\lg)^2(M+\epsilon)^{\#\lg-1}.
		\end{align*}
		This completes the proof.
	\end{proof}
	
	\begin{lem}\label{93A}
	For any set $ B\subseteq \mathbb{Z}^{d},$ there is a $s$-regular deformation $ B^{*}$ satisfying
\[
B\subseteq B^*\subseteq \Lambda_{50N^{100}_{s}}(B)=\{\bm x\in \mathbb{Z}^{d}: \dist(\bm x,B)\leq 50N^{100}_{s}\}.
\]
\end{lem}

	\begin{proof}
				We start from
		\begin{align*}
			J_{0,0}=B.
		\end{align*}
		Inductively define
		\begin{align*}
			J_{r,0}\subsetneq J_{r,1}\subsetneq \cdot\subsetneq J_{r,t_r}:=J_{r+1,0},
		\end{align*}
		where
		\begin{align*}
			J_{r,t+1}=J_{r,t}\cup\left(\bigcup_{\{\bm k\in P_{s-r}:\ \lg_{2N_{s-r}^{100}}(\bm k)\cap J_{r,t}\}}\tilde{\Omega}_{\bm k}^{s-r}\right),
		\end{align*}
		and $t_r$ is the largest integer satisfying the $\subsetneq$ relationship. Thus,
		\begin{align}\label{93sub}
			\bm k\in P_{s-r},\ \lg_{2N_{s-r}^{100}}(\bm k)\cap J_{r+1,0}\ne\emptyset\Rightarrow \lg_{2N_{s-r}^5}(\bm k)\subset J_{r+1,0}.
		\end{align}
		We claim that $t_r\le 2$, $0\le r\le s-1$. Otherwise, there exist three different points $\bm k_1,\bm k_2,\bm k_3\in P_{s-r}$ such that
		\begin{align*}
			J_{r,0}\cap \lg_{2N_{s-r}^{100}}(\bm k_1)\ne\emptyset,\ \lg_{2N_{s-r}^{100}}(\bm k_1)\cap \lg_{2N_{s-r}^{100}}(\bm k_2)\ne\emptyset,\ \lg_{2N_{s-r}^{100}}(\bm k_2)\cap \lg_{2N_{s-r}^{100}}(\bm k_3)\ne\emptyset.
		\end{align*}
		Thus $\max(\|\bm k_1-\bm k_2\|,\|\bm k_2-\bm k_3\|,\|\bm k_1-\bm k_3\|)\le10N_s^{100}$. According to \eqref{as1} and \eqref{as2}, then two of the three points $\bm k_i$ must satisfy
		\begin{align*}
			\|(\bm k_i-\bm k_j)\cdot\bm\omega\|_{\T}<6\delta_{s-r}^{\frac{1}{100}},\ i\ne j
		\end{align*}
		which implies $\|\bm k_i-\bm k_j\|>\left(\frac{\g}{6\delta_{s-r}^{\frac{1}{100}}}\right)^{\frac{1}{\tau}}\gg N_{s-r}^{100}$, a contradiction. Thus, we have shown
		\begin{align*}
			J_{r+1,0}=J_{r,t_r}\subset \lg_{10N_{s-r}^{100}}(J_{r,0}).
		\end{align*}
		SInce
		\begin{align*}
			\sum_{r=0}^{s-1}20 N_{s-r}^{100}<50N_{s}^{100},
		\end{align*}
		we find $J_{s,0}$ to satisfy
		\begin{align}\label{93B}
			B=J_{0,0}\subset J_{s,0}\subset \lg_{50N_s^{100}}(B).
		\end{align}
		Assume that for some $\bm k\in P_{s'}$ $(1\le s'\le s)$, $J_{s,0}\cap \tilde{\Omega}_{\bm k}^{s'}\ne\emptyset$. From $\tilde{\Omega}_{\bm k}^{s'}\subset \lg_{1.5N_{s'}^{100}}(\bm k)$, we obtain
		\begin{align*}
			J_{s,0}\cap \lg_{1.5N_{s'}^{100}}(\bm k)\ne\emptyset.
		\end{align*}
		Recalling \eqref{93B}, we have
		\begin{align*}
			J_{s,0}\subset \lg_{50N_{s'-1}^{100}}(J_{s-s'+1,0}).
		\end{align*}
		From $50N_{s'-1}^{100}\ll 0.5 N_{s'}^{100}$, it follows that
		\begin{align*}
			J_{s-s'+1,0}\cap \lg_{2N_{s'}^{100}}(\bm k)\ne\emptyset.
		\end{align*}
		By \eqref{93sub}, we get
		\begin{align*}
			\tilde{\Omega}_{\bm k}^{s'}\subset\lg_{2N_{s'}^{100}}(\bm k)\ne\emptyset\subset J_{s-s'+1,0}\subset J_{s,0}.
		\end{align*}
		Then $B^*=J_{s,0}$ is the set set that satisfies all conditions.
	\end{proof}
	
	\bibliographystyle{alpha}
	%\bibliography{Powerlaw}

\end{document}